\documentclass[a4paper]{article}

\usepackage[style=alphabetic]{biblatex}
\usepackage{bbold}
\usepackage{nicematrix}
\usepackage{qcircuit}
\usepackage{calc}
\let\1\relax
\usepackage{footnote}
\usepackage{graphicx}
\usepackage{mathtools}
\usepackage[font=small]{caption}
\usepackage[labelformat=simple]{subcaption}

\usepackage{prelude}

\newtheorem{algorithm}[theorem]{Algorithm}

\let\vec\relax
\let\nbd\nobreakdash
\DeclareMathOperator\vec{vec}

\DeclareMathOperator\sep{sep}
\newcommand\blockmatrix[4]{\left(\begin{array}{c|c} #1 & #2 \\ \hline #3 & #4 \end{array}\right)}
\newcommand\blockvector[2]{\left(\begin{array}{c} #1 \\ #2 \end{array}\right)}
\newcommand{\Gate}[1]{\gate{\raisebox{0pt}[.7em][.5em]{\makebox[1.2em]{$#1$}}}}

\DefineBibliographyStrings{english}{in = {}}
\bibliography{dissipativeGS}

\crefrangelabelformat{enumi}{#3#1#4--#5#2#6}
\crefformat{footnote}{#2\footnotemark[#1]#3}

\begin{document}

\title{Dissipative ground state preparation and the Dissipative Quantum Eigensolver}
\author{Toby S.\ Cubitt\thanks{%
    \href{mailto:t.cubitt@ucl.ac.uk}{t.cubitt@ucl.ac.uk};
    \href{mailto:toby@phasecraft.io}{toby@phasecraft.io}}}
\affil{Phasecraft Ltd., UK}
\affil{Department of Computer Science, UCL, London, UK}
\maketitle

\begin{abstract}
  For any local Hamiltonian $H$, I construct a local CPT map and stopping condition which converges to the ground state subspace of $H$.
  Like any ground state preparation algorithm, this algorithm necessarily has exponential run-time in general (otherwise BQP$=$QMA), even for gapped, frustration-free Hamiltonians (otherwise BQP$\supseteq$NP).
  However, this dissipative quantum eigensolver has a number of interesting characteristics, which give advantages over previous ground state preparation algorithms.
  \begin{itemize}
  \item%
    The entire algorithm consists simply of iterating the same set of local measurements repeatedly.
  \item%
    The expected overlap with the ground state subspace increases monotonically with the length of time this process is allowed to run.
  \item%
    It converges to the ground state subspace unconditionally, without any assumptions on or prior information about the Hamiltonian.
  \item%
    The algorithm does not require any variational optimisation over parameters.
  \item%
    It is often able to find the ground state in low circuit depth in practice.
  \item%
    It has a simple implementation on certain types of quantum hardware, in particular photonic quantum computers.
  \item%
    The process is immune to errors in the initial state.
  \item%
    It is inherently error- and noise-resilient, i.e.\ to errors during execution of the algorithm and also to faulty implementation of the algorithm itself, without incurring any computational overhead: the overlap of the output with the ground state subspace degrades smoothly with the error \emph{rate}, independent of the algorithm's run-time.
  \end{itemize}

  I give rigorous proofs of the above claims, and benchmark the algorithm on some concrete examples numerically.
\end{abstract}

\clearpage

\tableofcontents

\clearpage

\section{Introduction}\label{sec:introduction}
Finding ground states of quantum many-body systems is one of the most important---and one of the most notoriously difficult---problems in physics, both in scientific research and in practical applications.
Large amounts of computational resources in academia and industry are devoted to solving this or closely-related problems in computational modelling of quantum chemistry, materials science and condensed matter physics.

A number of successful classical variational algorithms for solving ground state problems have been developed over the last half century, and are now routinely used in computational modelling.
Ranging from density functional theory (DFT)~\cite{DFT}, widely used in quantum chemistry and materials modelling; to quantum Monte-Carlo, the method of choice for many condensed matter and high-energy physics problems; to tensor-network methods~\cite{tensor_networks} (of which the density matrix renormalisation group (DMRG)~\cite{DMRG} algorithm for 1D systems can be viewed as an example), which are increasingly seeing use in condensed matter theory.

However, although these methods are highly successful at solving ground state problems in certain cases, they fail to give meaningful results in others.
Indeed, we know from a complexity-theoretic perspective that all these methods must necessarily fail even for certain classical systems.
The ground state energy problem is NP-hard even for simple 2D many-body spin-lattice systems with nearest-neighbour interactions~\cite{Barahona}.
I.e.\ it is in a precise sense at least as hard as every other problem in the complexity class NP.
NP-hardness holds even for classical Hamiltonians that are frustration-free and have a constant, uniformly bounded spectral gap.

The argument is standard, but bears repeating.
Rewriting e.g.\ a 3SAT problem as a Hamiltonian in the standard way gives a classical, local Hamiltonian with spectral gap $\geq 1$, which is frustration-free iff the 3SAT problem is satisfiable.
If there were a polynomial-time algorithm that was guaranteed to find ground states of gapped, classical, local, frustration-free Hamiltonians, then one could run it for the appropriate time on a 3SAT Hamiltonian (regardless of whether or not it's frustration-free), and check whether the resulting state satisfies the 3SAT problem.
(The latter can be done efficiently because 3SAT is in NP.)
If it does, we've found a satisfying assignment; if it doesn't, we know the 3SAT problem is unsatisfiable, otherwise the frustration-free ground state algorithm would have found a solution.

For quantum systems, the situation is more challenging still.
The quantum ground state energy problem (often called the ``local Hamiltonian problem'') is QMA-hard~\cite{Kitaev_book}: a complexity class containing---and believed to be strictly larger than---the class NP.
As in the classical setting, it is still QMA-hard for simple 2D quantum spin-lattice models with nearest-neighbour interactions~\cite{KKR,Oliveira-Terhal,quantum_Schaeffer,quantum_Schaeffer_long}, and even for 1D spin chains~\cite{AGIK,Gottesman-Irani}.
(On the other hand, all QMA-hard many-body systems to date exhibit a spectral gap that decreases with system size~\cite{QMA-critical}.)
It is QMA${}_1$-hard for frustration-free quantum Hamiltonians---a class that sits somewhere between NP and QMA, but is still believed to be strictly larger than NP---by similar arguments to the classical case but applied to the QSAT problem.~\cite{QSAT}

In the early 1980s, Feynman~\cite{Feynman} suggested that, if quantum many-body systems were so difficult to solve computationally, perhaps quantum mechanics had greater computational power than classical mechanics, therefore one should use quantum computers to solve these problems.
However, the problem Feynman was targeting was \emph{not} ground state problems, but the (also classically challenging) problem of simulating the time-dynamics of quantum many-body systems.
We now know in great detail~\cite{Feynman,Lloyd,lincomb_unitaries,signal_processing,Haah_simulation} that time-dynamics simulation \emph{is} tractable on quantum computers: it is in the class BQP---indeed, it is complete for that class---whereas we do not know of any general sub-exponential-time classical algorithm for simulating quantum time-dynamics.
However, the quantum ground state energy problem is QMA-hard, not in BQP.
We therefore do not expect even quantum computers to be able to solve ground state problems in sub-exponential time in general.\footnote{Unless BQP$=$QMA, which would be the quantum analogue of NP collapsing to P in the infamous P vs.\ NP problem and is generally considered just as unlikely.}

Nonetheless, as ground state problems are of such practical importance, and classical algorithms are often successful despite the theoretical exponential worst-case complexity, a number of quantum algorithms for the ground state problem have been proposed and studied.
We review a number of these briefly in \cref{sec:background}.
The most relevant of these to this work is the dissipative state engineering algorithm of~\cite{VWC,Barbara}.
(Ref.~\cite{VWC} also contains other results, discussed in more detail in \cref{sec:background}.)

The dissipative state engineering algorithm provably finds the ground state of any frustration-free Hamiltonian, i.e.\ any Hamiltonian $H=\sum_{i=1}^m h_i$ such that the ground state $\ket{\psi_0}$ of $H$ is simultaneously the ground state of each local term $h_i$ individually.
The algorithm is deceptively simple and elegant, and is closely related to classical local search algorithms for constrain satisfaction problems:
Given a many-body Hamiltonian, for each local term of the Hamiltonian in turn, it performs a projective measurement onto the local ground space of that term or its complement.
If the measurement does not successfully project onto the local ground space, the state of the measured qubits is ``resampled'' (e.g.\ by replacing it with the maximally mixed state, i.e.\ resetting those qubits to a random computational basis state).
This process corresponds to iterating the following completely positive, trace-preserving (CPT) map:
\begin{equation}\label{eq:dissipative_engineering}
  \cE(\rho) = \frac{1}{m} \sum_{i=1}^m\left(
    \Pi_i\rho\Pi_i + \tr((\1-\Pi_i)\rho)\frac{\1}{D}
  \right),
\end{equation}
where $\Pi_i$ is the projector onto the ground space of the local term $h_i$.

Ref.~\cite{VWC} proves that, for any frustration-free Hamiltonian, iterating this CPT map is guaranteed to eventually converge to the ground state (albeit in exponential time in the worst case, consistent with the complexity theoretic considerations just discussed).

The dissipative state engineering algorithm has the interesting feature of promising to be inherently resilient to errors and noise, without any additional error-correction or fault-tolerance overhead.
Since the output of the algorithm is the fixed point of a quantum Markov process (i.e.\ the fixed point of the CPT map \cref{eq:dissipative_engineering}), the algorithm will converge to the ground state from any initial state.
Moreover, if the state is hit by an error part way through the process, it will inexorably converge back to the fixed point and hence to the desired ground state.
Since noise and errors are by far the biggest obstacle to practical implements of quantum algorithms, this feature is particularly attractive for near-term implementations.

The big limitation of this algorithm is that it only works for frustration-free Hamiltonians.
Not only does this restrict its applicability, it is also complexity-theoretically intractable to know in advance whether a given Hamiltonian is frustration-free or not.
This has thus far limited the applicability of dissipative state engineering.
The subsequent literature on ground state preparation has largely dismissed dissipative state engineering as a viable method for finding ground states on these grounds, and focused on other alternatives (see \cref{sec:background}).

\section{Main Results}

In this work, I develop a new quantum algorithm for preparing ground states of Hamiltonians, which overcomes the limitations of dissipative state engineering and also many of the limitations of other methods.
In particular, this Dissipative Quantum Eigensolver (DQE) algorithm retains all the advantages of dissipative state engineering, but is not restricted to the class of frustration-free Hamiltonians; it works for any Hamiltonian, without restriction and without any prior knowledge of, or assumptions on, any spectral or other properties of the Hamiltonian.

At a high level, there are two main ingredients that allow DQE to overcome the limitations of the original dissipative state engineering process:
\begin{description}
\item[Weak measurements:]%
  Rather than performing projective measurements onto the ground spaces of the local terms in the Hamiltonian, the DQE algorithm performs weak measurements of the local terms, i.e.\ generalised quantum measurements where one of the measurement operators is close to the identity.
  Intuitively, for frustration-free Hamiltonians, projective measurements are OK because successfully projecting onto the local ground space cannot take the state out of the global ground space (by definition of frustration-freeness).
  For frustrated Hamiltonians, however, successfully projecting onto the local ground space might drive the state completely out of the global ground space.

  On the other hand, intuitively, a successful weak measurement only ``nudges'' the state a towards a locally lower-energy state.
  This might still be the wrong direction globally.
  But the weak measurement only weakly perturbs the global state, rather than projecting it to completely the wrong subspace.
  By iterating the local weak measurements, the state is nudged incrementally towards the global ground state, rather than bouncing back and forth between local ground spaces.\footnote{Though the details are different, this is somewhat reminiscent of the approach used to generalise the constructive Quantum Lovasz Local Lemma~\cite{QLLL} from commuting Hamiltonians~\cite{commuting_QLLL,commuting_QLLL2} to general Hamiltonians~\cite{constructive_QLLL}.}

\item[Stopping conditions:]%
  Weak-measurements alone are not sufficient to solve the problem for general Hamiltonians.
  The fixed point of the process will now depend on the balance between the probability that a weak-measurement of a local term, performed on the global ground state, succeeds; and the probability that it fails and the corresponding qubits get resampled.
  I.e.\ it depends on how frustrated the Hamiltonian is.
  The more frustrated the Hamiltonian, the further the global ground space is from the ground spaces of the local terms, the further the fixed point will be from the true global ground state.

  However, the original fixed-point dissipative state engineering algorithm makes no use of the information accrued from the measurement outcomes whilst running the process; it simply discards this information after each iteration of the CPT map.
  By making use of the information provided by the measurement outcomes to conditionally stop the process, depending on the sequence of measurement outcomes observed up to that point, we can ensure the process stops on a state that is not the fixed point, but rather one guaranteed to be close to the ground space.\footnote{This is somewhat reminiscent of Sch\"oning's algorithm for SAT~\cite{Schoning}, which stops after a carefully chosen number of steps that is too small to have converged to the steady state.
    Though the stopping conditions for DQE are not the same as those for Sch\"oning's algorithm.
    Sch\"oning's algorithm has also been generalised to the quantum case for 2QSAT~\cite{quantum_Schoening}.}
\end{description}

First, we need a general definition of the DQE family of algorithms.

\begin{algorithm}[DQE]\label{DQE}
  Let $K = \sum_i k_i$ be an approximate ground state projector (AGSP) for a Hamiltonian $H$ (see \cref{sec:AGSP}), where $k_i$ are Hermitian operators.
  Let $\{\cE_{i,0}^{(t)},\cE_{i,1}^{(t)}\}$ be the quantum instrument defined by
  \begin{equation}
    \cE_{i,0}^{(t)}(\rho) = E_i^{(t)} \rho (E_i^{(t)})^\dg, \quad
    \cE_{i,1}^{(t)}(\rho) = \left(1-\tr\bigl(E_i^{(t)} \rho (E_i^{(t)})^\dg \bigr)\right) \cR_i^{(t)}(\rho),
  \end{equation}
  where $E_i^{(t)} := (1-\epsilon_t)\1 + \epsilon_t k_i$, the resampling map $\cR_i^{(t)}$ is CPT, and $0 < \epsilon_t < 1$ is a sequence of real numbers.
  Let $\tau$ be a stopping rule.

  The DQE algorithm consists of successively applying the weak-measurements $\cE_i^{(t)} := \{\cE_{i,0}^{(t)},\cE_{i,1}^{(t)}\}$ in a prescribed order (which could be deterministic or stochastic), and stopping according to the stopping rule $\tau$.
\end{algorithm}

The main result is that this algorithm converges to the ground state for any Hamiltonian, without any spectral or other assumptions on the Hamiltonian.

\begin{theorem}[informal, see \cref{decaying_CPT_map} for precise version]
  We can give an explicit choice of parameters in the DQE \cref{DQE} such that, for any $k$-local Hamiltonian, \cref{DQE} can be implemented using $k$-local generalised measurements and, as the allowed total run-time $t$ increases, the state at the stopping time $\tau$ converges to a ground state of $H$:
  \begin{equation}
    \lim_{t\to\infty} \tr(\Pi_0\rho_\tau) = 1,
  \end{equation}
  where $\Pi_0$ is the projector onto the ground space of $H$.
\end{theorem}

I derive exact expressions for the expected ground state overlap and run-time of \cref{DQE}, both for global resampling $\cR(\rho) = \1/D$ (see \cref{sec:stopped}) and for general resampling maps (see \cref{sec:resampling}).
The total expected run-time scales exponentially in the system size (number of qudits), as it must on complexity-theoretic grounds discussed above.
However, for non-unitary processes, there is an important distinction between total run-time and circuit depth -- the time over which coherence of the quantum state must be maintained in order for the algorithm to succeed.
I show that the circuit depth scales linearly in the system size:

\begin{corollary}[See \cref{stopping_time,stopped_circuit_depth} for precise statements.]
  For an $n$-qubit system, to obtain an output state $\rho$ satisfying
  \begin{equation}
    \tr(\Pi_0\rho) \leq 1 - \epsilon
  \end{equation}
  the DQE algorithm requires a total run time that scales exponentially in the system size $n$, but a circuit depth that only scales as $O(n)$.
\end{corollary}

Moreover, I show that the DQE \cref{DQE} is inherently noise- and fault-resilient if the noise is below a threshold.
(A similar noise-resilience result also applies to the fixed-point process; see \cref{fixed-point_fault-resilience}.)

\begin{theorem}[informal; see \cref{stopped_fault-resilience} for precise version]
  Let $\{\cE_0,\cE_1\}$ be the quantum instrument of \cref{DQE}, and $\{\cE'_0,\cE'_1\}$ be the quantum instrument describing a faulty implementation of this process with
  \begin{equation}
    \norm{\cE'_0-\cE_0}_1 \leq O(\delta).
  \end{equation}
  We can give an explicit choice of parameters in \cref{DQE} such that, under the faulty process $\{\cE'_0,\cE'_1\}$, the state $\rho'_\tau$ at the stopping time satisfies
  \begin{equation}
    \lim_{t\to\infty}\tr(\Pi_0\rho'_\tau) = 1 - O(\epsilon) - O(\delta),
  \end{equation}
  where $\Pi_0$ is the projector onto the ground space of $H$.
\end{theorem}

As an interesting footnote, I also show that for any Hamiltonian there \emph{exists} a CPT map of the form of \cref{DQE} whose fixed point is as close as desired to the ground state.
However, the parameters required for this depend on the spectrum of the Hamiltonian, so cannot be found efficiently.
It is therefore not useful for constructing a general, practical ground state preparation algorithm, but may be of interest theoretically.

\begin{theorem}[informal, see \cref{Chebyshev_fixed-point} for precise version]
  There exists a choice of parameters in the DQE \cref{DQE} such that, for any $k$-local Hamiltonian and any $\epsilon$, \cref{DQE} can be implemented using $k'$-local generalised measurements with $k'=k\log(1/\epsilon)$ and the fixed point of $\cE$ satisfies:
  \begin{equation}
    \tr(\Pi_0\rho_\infty) \geq 1 - O(\epsilon D),
  \end{equation}
  where $\Pi_0$ is the projector onto the ground space of $H$.
\end{theorem}

The DQE algorithm has some distinct advantages over previous ground state preparation methods (see \cref{tbl:alg_comparison} for a summary):

\begin{itemize}
\item%
  Unlike dissipative state engineering, the DQE algorithm is not restricted to the special class of frustration-free Hamiltonians, but works for any Hamiltonian (\cref{CPTP_map,stopped_CPT_map}).
\item%
  Unlike the heuristic VQE algorithm, DQE is provably guaranteed to find the correct ground state in circuit-depth scaling linearly with system size (\cref{stopped_CPT_map,stopped_circuit_depth}). (But note that the total run-time remains exponential, as it must on complexity-theoretic grounds.)
\item%
  Unlike other non-heuristic algorithms, DQE is guaranteed to converge to the ground state without any assumptions on or prior information about the Hamiltonian (\cref{decaying_CPT_map}).
\item%
  Unlike VQE, DQE does not require any costly and heuristic optimisation over parameters (\cref{secretary_CPT_map,expectation_CPT_map}).
\item%
  Unlike adiabatic state preparation and phase estimation, DQE often finds the ground state quickly and in low circuit depth in practice.
\item%
  Unlike adiabatic state preparation, phase estimation, imaginary-time evolution and VQE (and indeed any unitary algorithm), DQE is inherently noise- and fault-resilient, without any additional computational overhead (\cref{stopped_fault-resilience}).\footnote{\label{fn:fault-resilience}This form of fault-resilience has not been proven for dissipative state engineering, but the results of this work imply it also holds there too.}
\item%
  DQE lends itself particularly well to implementation on certain types of quantum hardware, such as photonic quantum computers.
\end{itemize}

It also has some disadvantages compared to VQE.
To estimate the expectation value of any observable on the ground state requires repeatedly preparing the state, measuring the observable, and averaging over measurement outcomes.
To generate a new copy of the ground state, in principle the whole DQE algorithm must be run again.
Contrast this with VQE where, once the variational parameters in the VQE circuit have been trained by the (expensive) outer optimisation step, in the absence of noise the VQE ground state can be generated deterministically by running the VQE circuit with those parameters, without rerunning the optimisation part of the VQE algorithm.
In practice, in the presence of noise, even the trained VQE circuit will not deterministically produce the correct state, and one may need to rerun the VQE circuit and even the optimisation step more than once to get good results.
And in practice, continuing to run the DQE algorithm on the post-measurement state after measuring an observable, may converge back to the ground state relatively quickly, allowing observables to be estimated more efficiently than by re-running the whole process from scratch after each measurement.
Nonetheless, the classical overhead for generating multiple copies of the ground state is likely to be higher in the case of DQE.
(On the other hand, VQE gives no guarantee it is producing the correct state in the first place.)

\begin{savenotes}
\begin{table}[!htbp]
  \center
  \makebox[\textwidth][c]{
    \begin{tabular}{lcccccc}
      \hline
      Algorithm
      & Adiabatic    & QPE          & $i$-time     & VQE          & Dissipative & DQE \\
      \hline
      No conditions on $H$
      &              &              &              & $\checkmark$ &              & $\checkmark$ \\
      Provably succeeds
      &              & $\checkmark$ & $\checkmark$ &              & $\checkmark$ & $\checkmark$ \\
      Low-depth in practice
      &              &              &              & $\checkmark$ &              & $\checkmark$ \\
      No parameter optimisation
      & $\checkmark$ & $\checkmark$ & $\checkmark$ &              & $\checkmark$ & $\checkmark$ \\
      Efficient to generate multiple copies
      &              &              &              & $\checkmark$ &              &              \\
      Convenient implementation
      &              &              &              & $\checkmark$ & $\checkmark$ & $\checkmark$ \\
      Insensitive to initial state
      &              &              &              &              & $\checkmark$ & $\checkmark$ \\
      Noise- and fault-resilient
      &              &              &              &              & $\checkmark$\cref{fn:fault-resilience} & $\checkmark$ \\
      \hline
    \end{tabular}
  }
  \caption{Overview of advantages and disadvantages of the adiabatic state preparation, quantum phase estimation, imaginary-time dynamics, VQE, dissipative state engineering and DQE algorithms for preparing ground states of Hamiltonians.}
  \label{tbl:alg_comparison}
\end{table}
\end{savenotes}

In \cref{sec:background}, I review previous ground state preparation methods and results.
In the remainder of the paper, I give a detailed analysis of the DQE \cref{DQE} and rigorous proofs of the main (and other) results.
In \cref{sec:AGSP} I introduce and give constructions of approximate ground state projectors (AGSPs), which form the basis in \cref{sec:stopping} for constructing dissipative dynamics that converge to the ground state subspace of a Hamiltonian and proving rigorous bounds on their convergence.
These dynamics in turn form the basis, in \cref{sec:stopping,sec:epsilon,sec:resampling}, for the Dissipative Quantum Eigensolver (DQE) of \cref{DQE}, where convergence to the ground state subspace is proven.
\Cref{sec:fault-resilience} contains the discussion and proof of noise- and fault-resilience of DQE, and concludes the analytical results in this work.

\Cref{sec:implementation} shows how the DQE algorithm can be implemented on quantum computing hardware.
\Cref{sec:conclusions} concludes with a discussion of various extensions of the DQE algorithm, and possible future directions.

\clearpage

\section{Background and Previous Work}\label{sec:background}

\subsection{Phase estimation}
For classical Hamiltonians on finite-dimensional systems and for combinatorial constraint optimisation problems, the naive brute-force search algorithm is guaranteed to find the ground state (respectively, optimal assignment): search over all possible states and compute their energies (cost).
The energy can be computed efficiently, but the state space to search over is exponentially large.
So the run-time is exponential.

In the quantum case, it's less clear what the quantum generalisation of naive brute-force search is.
Arguably, the closest analogue (though it's no longer ``naive'') is to use quantum phase estimation.
The quantum phase estimation algorithm~\cite[\S 5.2]{Nielsen+Chuang} applied to the unitary time-evolution operator generated by the Hamiltonian collapses the state into a random eigenstate of the Hamiltonian, along with an estimate of its corresponding energy eigenvalue (modulo $2\pi$).
If the Hamiltonian is suitably rescaled to ensure all eigenvalues are in the interval $[-\pi,\pi]$, and assuming the eigenvalues are estimated to sufficient precision to resolve the ground state energy, this allows one to rank the energy of the eigenstates produced, and with sufficiently (exponentially in the number of qudits) many repetitions, generate and identify the ground state with high probability.

However, knowing what precision in the energy estimate suffices requires prior knowledge of (a lower bound on) the spectral gap.
Moreover, implementing each phase estimation repetition requires a complex (albeit polynomial-time) quantum computation, involving both the quantum Fourier transform and the Hamiltonian simulation algorithm.

\subsection{Adiabatic state preparation}
Adiabatic state preparation, like its close cousin adiabatic quantum computation~\cite{AQC}, is based on the adiabatic theorem in quantum mechanics.
In adiabatic state preparation, the quantum system is initialised in the ground state of some simple Hamiltonian $H_0$ whose ground state is easy to prepare, e.g.\ the all~$\ket{0}$ ground state of the Hamiltonian $H_0 = \sum_i \proj{1}_i$.
By slowly transforming the Hamiltonian from $H_0$ to the Hamiltonian of interest, $H_1$ (e.g.\ by varying the parameter $s$ linearly from 0 to 1 in $H(s) = (1-s)H_0 + s H_1$; more complex and sophisticated paths are also well-studied), the adiabatic theorem tells us that the state of the system will track the instantaneous ground state of $H(s)$.
Thus, at the end of the process, this will have prepared the ground state of $H_1$.

The adiabatic theorem itself follows from solving the time-dependent Schr\"odinger equation and taking the adiabatic approximation, where the rate of change of the Hamiltonian is much slower than the instantaneous spectral gap of $H(t)$.
This gives the conditions under which adiabatic state preparation will succeed: the spectral gap along the entire adiabatic path must remain large relative to the speed at which the Hamiltonian is varied.
In particular, this means that if the final Hamiltonian $H_1$ is in a different phase to the initial Hamiltonian $H_0$, so that the spectral gap vanishes at some point along the path, the adiabatic algorithm may fail.
Moreover, the number of states within energy~$E$ of the ground state (the density of states) must scale at most polynomially in~$E$.

Whilst adiabatic state preparation provably succeeds if the conditions of the adiabatic theorem are fulfilled, in certain cases the spectral gap becomes exponentially small along the given path, and the algorithm is necessarily exponential-time.
In other cases, the density of state scales exponentially (e.g.\ for NP-hard classical Hamiltonians).
In many cases of interest, no bound is known on how fast the spectral gap closes, or even any proof that the spectral gap doesn't vanish along the path.
(Indeed, in general, the spectral gap question is not just hard, it is provably unsolvable~\cite{spectral-gap_long,spectral-gap_short,spectral-gap_1D}.
Though to date this has only been shown for extremely artificial Hamiltonians.)

Ref.~\cite{BoixoKnillSomma2010} develops a discrete-time, digital version of adiabatic state preparation, which relaxes some of the requirements of the standard adiabatic theorem.
But it still requires conditions on the eigenvalues and state overlaps along the whole discrete sequence of Hamiltonians.

Due to these limitations, in practice adiabatic state preparation is generally proposed as a heuristic algorithm.
By trying different rates of adiabatic variation, if the final state is consistent when the rate of variation is low enough, it will hopefully be a good estimate of the ground state.
Of course, due to the complexity theory results discussed above, this must fail for at least some cases.
Moreover, determining which cases it will succeed and fail on is at least as hard as solving the original ground state problem itself.

Even where it works, implementing adiabatic state preparation on a (digital) quantum computer requires significant overhead, since the continuous-time dynamics under the time-varying Hamiltonian $H(s)$ must be approximated by a discretised time evolution (e.g.\ by Trotterisation, or one of the other techniques for quantum time-dynamics simulation).
This run-time overhead can be large, as the discretisation errors must be sufficiently small relative to the minimum spectral gap along the adiabatic path.

Furthermore, as in any unitary quantum computation, in the absence of error correction, any gate errors during the computation accumulate linearly with the run-time of the algorithm.
Meanwhile, any dissipative errors and noise accumulate exponentially with run-time, unless fault-tolerant computation is used.

\subsection{Variational Quantum Eigensolver (VQE)}
The VQE algorithm takes its inspiration partly from classical variational algorithms, partly from adiabatic quantum state preparation.
Variational algorithms for finding ground states are based on the variational characterisation of the minimum eigenvalue of the Hamiltonian $H$: $E_{\mathrm{min}} = \min_{\ket{\psi}}\braket{\psi|H|\psi}$.
Variational algorithms in general work by choosing some suitable class of states $\{\ket{\psi_i}\}$ that hopefully includes (a good approximation to) the ground state, then minimising the energy over this class of states using some minimisation algorithm.
DMRG can be viewed as variational minimisation over the class of matrix product states, which are efficiently representable classically, using a minimisation algorithm that is efficiently (in practice~\cite{DMRG} or even provably~\cite{rigorous_DMRG}) computable classically.
DFT can be viewed as implicit variational minimisation over a class of states parameterised by the density functional, using an iterative self-consistent optimisation loop.

In VQE, rather than choosing states that are efficiently computable classically, one chooses a class of states that can be prepared efficiently \emph{quantumly}, i.e.\ states generated by some class of polynomial-sized quantum circuits.
The energy, which can be computed efficiently on a quantum computer, is then minimised over this class of circuits, which hopefully converges to a good approximation to the ground state.
Various different classes of VQE circuits and various optimisation algorithms have been proposed and studied.

One promising class is the ``Hamiltonian variational ansatz'',\cite{VHA} where the class of circuits consists of layers of unitaries generated by the local interactions of the Hamiltonian of interest.
The free parameters are the scalar coefficients---or ``angles''---multiplying the terms in each layer.
For sufficiently small angles and sufficiently many layers, this class of the circuits includes the circuit implementing Trotterised adiabatic state preparation for the Hamiltonian.
Thus, under the same spectral gap and density of states assumptions as the adiabatic algorithm, there at least \emph{exists} a VQE solution that achieves the correct ground state.

However, the idea in VQE is rather to choose a much smaller number of layers, and hope that the variational minimisation will find more efficient, lower-depth circuits with larger angles that reach the ground state much faster than the full adiabatic circuit.
As the circuit depth is significantly lower, and varying over the parameters may compensate for some portion of the errors in the computation---at least coherent gate errors---VQE may be somewhat less affected by errors and noise during the algorithm.
Thanks to these properties, low-depth VQE can already be run on current and near-term quantum computers~\cite{PhasecraftVQE}.

On the other hand, VQE is an entirely heuristic algorithm.
Whether it will work for any given Hamiltonian is unknowable without running it, and even then there is no way to be sure it has produced the correct state or energy at the output.
Even if the variational class does include a good approximation to the ground state (which is rarely guaranteed or even knowable), the optimisation algorithm may fail to find it.

There is a difficult trade-off: the more free parameters in the variational class, the larger the class of states and the higher the chances it contains a good ground state approximation.
However, generally, the larger the variational class of states, the more computationally expensive it is to optimise over them, and the lower the chance of the optimiser finding a good solution even if one in principle exists.

\subsection{Imaginary-time evolution}
Imaginary-time evolution is a successful classical computational method for certain many-body quantum problems.
As it is a non-unitary evolution, it cannot be implemented directly on a quantum computer.
But it can be approximated unitarily~\cite{Brandao_QITE,Yeter_QITE,Nishi_QITE,Benjamin_QITE,Sun_QITE} or probabilistically~\cite{Lin_PITE,Liu_PITE,Kosugi_PITE}.
There are too many recent results to do justice to all the variations of imaginary-time evolution here, but detailed discussions can be found in these works and the references therein.

However, implementing imaginary-time evolution via these methods is broadly similar to implementing (and sometimes uses as a subroutine) real-time Hamiltonian dynamics simulation, which, although it can be implemented in polynomial time on a quantum computer~\cite{Feynman,Lloyd,lincomb_unitaries,signal_processing,Haah_simulation}, is a more complex quantum algorithm than the simple local measurements needed for the dissipative state engineering algorithm of~\cite{VWC}.

\subsection{Dissipative state engineering}
In 2008, \cite{VWC,Barbara} proposed a new algorithm for preparing ground states of quantum many-body systems based on local dissipative dynamics rather than unitary evolution.
(\cite{VWC} also showed how this dissipative process could be used to implement quantum computation, whilst avoiding some of the ``DiVincenzo'' criteria~\cite{diVincenzo} generally considered at that time to be necessary for quantum computation.)

In its simplest form, the dissipative state engineering algorithm of~\cite{VWC} consists in measuring each local term of the Hamiltonian in turn, and if the result of a measurement is not the minimum eigenstate of that local term, ``resampling'' the qudits involved in that term by randomising their state (or otherwise replacing the state of those qubits according to some procedure).

They prove that this algorithm succeeds in preparing the ground state of any frustration-free Hamiltonian: a Hamiltonian whose ground state simultaneously minimises the energy of all of its local terms.
Another variant of the algorithm can efficiently prepare certain tensor network states, such as stabiliser states, matrix product states (MPS) that satisfy certain generic properties, and their higher-dimensional analogues: projected entangled pair states (PEPS).

In the paper, \cite{VWC} also claim that a version of the dissipative state engineering algorithm with a slightly more sophisticated resampling procedure can prepare any frustration-free ground state \emph{efficiently}.
However, there is an important subtlety to this claim which means that the efficient version of the algorithm is not useful for solving ground-state problems, even for frustration-free Hamiltonians.
Whilst it is true \cite{VWC} shows that, for any frustration-free ground state, there exists a local dissipative process which efficiently prepares that state, one can only construct this process if one already knows in advance the state to be prepared.
(Contrast this with their exponential-time process, which \emph{can} prepare the ground state given only a description of the frustration-free Hamiltonian.)

To understand the crucial difference between \emph{preparing a known} ground state and \emph{finding an unknown} ground state, it is illustrative to consider the case of classical Hamiltonians.
Since ground states of classical Hamiltonians are always product states, it is trivially true that there is an efficient algorithm for preparing any known ground state: one can simply flip the state of each particle to match the ground state.
However, as discussed above, finding the ground state of frustration-free classical Hamiltonians is already NP-hard.
So, whilst any classical state can be prepared efficiently, finding the ground state in the first place is in general intractable.

Since \cite{VWC,Barbara}, there have been many follow-up papers inspired by their results.
These can be divided into two classes:
\begin{enumerate}
\item Algorithms for finding ground states~\cite{demon-cooling,state_restoration}.
\item Algorithms for constructing tensor network states~\cite{Christandl2021,Kim2017,Kim+Swingle2017,Sewell+Jordan2021}.
\end{enumerate}

The latter build on the efficient dissipative state engineering algorithm of~\cite{VWC} for preparing states with a known tensor-network description.
These are important classes of states, and lead to interesting state-preparation algorithms.
But this class of state-preparation results are not applicable to the Hamiltonian ground state problem, for the reasons already discussed above.

The demon-like cooling algorithm of \cite{demon-cooling} overcomes the restriction to frustration-free Hamiltonians of \cite{VWC}, by combining Hamiltonian time-dynamics simulation with projective measurements to produce a dynamics closely related to imaginary-time evolution.
(They also implement their algorithm experimentally using photons, on a toy single-qubit Hamiltonian.)
This achieves some of the benefits of dissipative state engineering, such as iteratively applying the same quantum circuit repeatedly, leading to a simpler implementation than algorithms such as adiabatic state preparation, and immunity to errors in the initial state.
(There is no analysis in~\cite{demon-cooling}, analytical or numerical, of whether the algorithm is resilient to noise and faulty implementation, though as with \cite{VWC} this may follow from results shown in this work.)
However, as with imaginary-time evolution, this comes at the cost of requiring (coherently controlled) Hamiltonian time-dynamics simulation.
To guarantee it succeeds, the algorithm also still requires knowledge of (a lower-bound on) the spectral gap of the Hamiltonian.

The quantum state restoration results of~\cite{state_restoration} are motivated by a different but related problem: of restoring the ground state when part of it is traced out.
They assume access to projective measurements onto the global ground state, and show that this together with access to a single copy of the ground state suffices to compute local observables and other local ground state properties.
As it requires access to the projector onto the global ground state, this is not (nor were the results intended to be) practicably applicable to ground state preparation of many-body Hamiltonians.
Rather, practical ground state preparation algorithms could provide a useful building block on which to implement the algorithms developed in~\cite{state_restoration} for estimating local ground state properties.

\subsection{Preliminaries}
Throughout this paper, $\norm{\;\cdot\;}$ denotes the standard operator norm, whilst $\norm{\;\cdot\;}_p$ denotes the Schatten $p$-norm (so $\norm{\;\cdot\;} = \norm{\;\cdot\;}_\infty$).
When there is no ambiguity, and in particular where $f(X)$ and $g(X)$ are polynomials in the same operator $X$ (hence commute), I will sometimes write $\frac{f(X)}{g(X)} = f(X) g(X)^{-1} = g(X)^{-1} f(X)$ where the former notation increases legibility.

I will sometimes make use of the fact that the action of a CP map $\cE(\rho) = \sum_i A_i\rho A_i^\dg$ as a linear map on the $D$\nbd-dimensional vector space of operators can be represented by matrix multiplication.
For a density matrix $\rho$, I write $\kett{\rho}$ to denote the vectorised density matrix, which concretely can be taken to be the vector formed by stacking the columns of $\rho$.
The action of $\cE$ can then be represented by the $D^2$-dimensional transfer matrix $E = \sum_i A_i\ox\bar{A_i}$ where $A_i$ are the Kraus operators of $\cE$, so that $\kett{\cE(\rho)} = E\kett{\rho}$.

I will frequently refer to ``stopping rules'' and ``stopping times'' defined on the stochastic process generated by the sequence of measurement outcomes produced by \cref{DQE}.
The formal definitions of stopping rules and stopping times for stochastic processes (see e.g.~\cite{Williams_book}) have a simple intuitive interpretation: the decision to stop (or not) can only depend on past outcomes, not future ones.
All stopping rules discussed in this work will manifestly satisfy this requirement.

\clearpage

\section{Approximate Ground State Projectors}\label{sec:AGSP}

Approximate ground state projectors (AGSP) are operators that approximate in some suitable sense the projector onto the ground state of a (typically many-body) quantum Hamiltonian.
AGSPs were first defined in the context of proving area-laws for many-body Hamiltonians~\cite{AGSP}, where the rate at which the operator increases the entanglement of a state plays an important role, but actual implementation of the AGSP is not relevant.
Here, I repurpose and generalise AGSPs for the different context of algorithmic ground state preparation.
The key results are \cref{AGSP_product,AGSP_mixture}, which show how AGSPs can be implemented by local quantum instruments.

\begin{definition}[AGSP]\label{def:AGSP}
  Let $\Pi_0$ be the projector onto the ground space of a Hamiltonian.
  A Hermitian operator $K$ is a \emph{$(\Delta,\Gamma,\epsilon)$-approximate ground state projector} (AGSP) for $\Pi_0$ if there exists a projector $\Pi$ such that:
  \begin{enumerate}
  \item $[K,\Pi]=0$, \label{def:AGSP:commutator}
  \item $K \Pi \geq \sqrt{\Gamma}\Pi$, \label{def:AGSP:Gamma}
  \item $\norm{(\1-\Pi)K(\1-\Pi)} \leq \sqrt{\Delta}$, \label{def:AGSP:Delta}
  \item $\norm{\Pi-\Pi_0} \leq \epsilon$. \label{def:AGSP:epsilon}
  \end{enumerate}
\end{definition}

\begin{remark}
  The AGSPs of \cite[Definition~2.1]{AGSP} correspond to the special case of $(\Delta,1,0)$-AGSPs in \cref{def:AGSP}.
  However, in \cite{AGSP} they also impose an additional condition on the Schmidt-rank blow-up of the AGSP, which is critical to their proof of the 1D area law but is not relevant to this work.

  Note that, since $[K,\Pi]=0$, \labelcref{def:AGSP:Delta} could equivalently be written $K(\1-\Pi) \leq \sqrt{\Delta}(\1-\Pi)$.
  I have chosen to write it as in \labelcref{def:AGSP:Delta}, to match the form of the AGSP definition of~\cite[Definition~2.1]{AGSP}.
\end{remark}

\begin{lemma}[{\cite[\S 4]{AGSP}}]\label{AGSP_H}
  Let $H=\sum_i h_i$ be a local Hamiltonian, with minimum eigenvalue $\lambda_0$ and spectral gap $\delta$.
  Let $\Pi_0$ be the spectral projector corresponding to $\lambda_0$.
  Let $C_\ell(x)$ be the rescaled Chebyshev polynomial of degree $\ell$ from \cite[Lemma~4.1]{AGSP}.
  Then $C_\ell(H)$ is a $(\Delta,1,0)$-AGSP with $\sqrt{\Delta} = 2e^{-2\ell\sqrt{\delta/(\norm{H}-\lambda_0)}}$.
\end{lemma}

\begin{lemma}\label{Chebyshev_H}
  Let $H=\sum_i h_{i=1}^m$ be a $k$-local Hamiltonian, and $C_\ell(x)$ be a degree $\ell$ rescaled Chebyshev polynomial satisfying $\abs{C_\ell(x)} \leq 1$ for $\lambda_{\min}(H) \leq x \leq \lambda_{\max}(H)$.
  Then $C_\ell(H) = \sum_{i=1}^{m'} h'_i$ is a $k\ell$-local Hamiltonian, with $m'\leq \left(\frac{e}{\ell}\right)^\ell m^\ell = \poly(m)$ terms.
\end{lemma}

\begin{proof}
  Immediate from $C_\ell$ being a degree $\ell$ polynomial and standard binomial coefficient bounds (see e.g.\ \cite{partial_binom_sum}).
\end{proof}

The AGSPs of \cref{AGSP_H,Chebyshev_H} are examples of AGSPs with arbitrarily good parameters that can (via the results of \cref{sec:local_AGSPs}) be implemented using local operations.
However, the rescaled Chebyshev polynomials $C_\ell$ of \cite[Lemma~4.1]{AGSP} can only be constructed given knowledge of the ground state energy and spectral gap of $H$, which may not be available.
In \cref{sec:local_AGSPs}, we will see examples of local AGSPs which can be constructed without detailed spectral information.

\subsection{AGSP perturbation}
We now prove some results that allow us to construct new AGSPs by perturbing an existing AGSP.
For this, we will need the Davis-Kahan $\sin\Theta$ Theorem, the following variant of which is due to~\cite{Yu+2015}.

\begin{theorem}\label{Davis-Kahan}
  Let $A$, $A'$ be Hermitian matrices.
  Let $P$, $P'$ be the orthogonal projectors onto the eigenspaces associated with the $k$ largest eigenvalues of $A$ and $A'$, respectively.
  Assume that the spectrum of $A$ satisfies $\lambda_k-\lambda_{k+1} \geq \delta$, where $\lambda_i$ is the $i$'th largest eigenvalue of $A$.
  Then
  \begin{equation}
    \norm{P-P'} \leq \frac{2\sqrt{k}\norm{A-A'}}{\delta}.
  \end{equation}
\end{theorem}

\begin{proof}
  This is a special case of \cite[Theorem~2]{Yu+2015}, together with the fact that\linebreak $\norm{P-P'} = \norm{P^\perp P'} = \norm{\sin\Theta(V,V')}$ (see e.g.\ \cite[Exercise~VII.1.10--11]{Bhatia}) and the standard relation $\norm{X}\leq\norm{X}_2$.
\end{proof}

\begin{lemma}[AGSP perturbation]\label{AGSP_perturbation}
  Let $K$ be a $(\Delta,\Gamma,0)$-AGSP for $\Pi_0$ with ground state degeneracy $N:=\tr\Pi_0$.
  If $K'$ is a Hermitian operator such that $\delta := \norm{K-K'} < \abs{\sqrt{\Gamma}-\sqrt{\Delta}}$, then $K'$ is a $(\Delta+\delta, \Gamma-\delta, \epsilon)$-AGSP, with
  \begin{equation}
    \epsilon = \frac{2\sqrt{N}\delta}{\sqrt{\Gamma}-\sqrt{\Delta}}.
  \end{equation}
\end{lemma}

\begin{proof}
  Since $[K,\Pi_0]=0$ by \cref{def:AGSP}\labelcref{def:AGSP:commutator}, it decomposes as $K=K_0\oplus K_1$ w.r.t.\ the $\Pi_0,(\1-\Pi_0)$ decomposition of the Hilbert space.
  By \cref{def:AGSP}\labelcref{def:AGSP:Gamma,def:AGSP:Delta}, all eigenvalues of $K_0$ are $\geq \sqrt{\Gamma}$ and all eigenvalues of $K_1$ are $\leq \sqrt{\Delta}$.
  Thus by \cref{eigenvalue_perturbation}, noting that $K$ is Hermitian hence diagonalised by a unitary, $K'$ has some eigenvalues $\geq\sqrt{\Gamma}-\delta$ and the rest $\leq\sqrt{\Delta}+\delta$.

  Applying \cref{Davis-Kahan}, we have
  \begin{equation}
    \norm{\Pi'_0-\Pi_0} \leq \frac{2\sqrt{N}\norm{K-K'}}{\sqrt{\Gamma}-\sqrt{\Delta}} = \epsilon,
  \end{equation}
  where $\Pi'_0$ is the projector onto the $\geq\sqrt{\Gamma}-\delta$ eigenspace of $K'$.
  Thus $\Pi'_0$ fulfils \cref{def:AGSP}\labelcref{def:AGSP:epsilon}.

  Let $K'=K'_0\oplus K'_1$ be the decomposition of $K'$ with respect to $\Pi'_0,(\1-\Pi'_0)$.
  By definition of $\Pi'_0,K'_0,K'_1$, we have
  \begin{align}
    [K',\Pi'_0] &= 0, \\
    K'\Pi'_0 &= K'_0\Pi'_0 \geq (\sqrt{\Gamma}-\delta)\Pi'_0, \\
    \norm{(\1-\Pi'_0)K'(\1-\Pi'_0)} &= \norm{K'_1(\1-\Pi'_0)} \leq \sqrt{\Delta}+\delta.
  \end{align}
  Thus $K'$ fulfils \cref{def:AGSP}\labelcref{def:AGSP:Gamma,def:AGSP:Delta,def:AGSP:commutator}.
\end{proof}

Note that \cref{AGSP_perturbation} can readily be generalised to the case where $K$ is an $(\Delta,\Gamma,\epsilon)$-AGSP, rather than assuming $\epsilon=0$.
However, we will not need this generalisation here.

\subsection{Local AGSPs}\label{sec:local_AGSPs}
The following results allow us to construct AGSPs from local operations.

\begin{lemma}\label{AGSP_sum}
  If $H = \sum_i h_i$ is a Hamiltonian, then
  \begin{equation}
    \begin{gathered}
      K = \frac{\1 - H/\kappa}{2} = \sum_i \kappa_i k_i\\
      \text{where}\quad
      k_i = \frac{\1-h_i/\norm{h_i}}{2}, \quad
      \kappa = \sum_i\norm{h_i}, \quad
      \kappa_i = \frac{\norm{h_i}}{\kappa}
    \end{gathered}
  \end{equation}
  is a $(\Delta,\Gamma,0)$-AGSP for the projector $\Pi_0$ onto the ground space of $H$, with
  \begin{equation}
    \sqrt{\Gamma} = \frac{1-\lambda_0/\sum_i\norm{h_i}}{2},\qquad
    \sqrt{\Delta} = \frac{1-\lambda_1/\sum_i\norm{h_i}}{2},
  \end{equation}
  where $\lambda_0$ and $\lambda_1$ are, respectively, the minimum and next-minimum eigenvalues of $H$ (not counting degeneracies).
\end{lemma}

\begin{proof}
  Clearly $[K,\Pi_0] = 0$, so $K$ satisfies \cref{def:AGSP}\labelcref{def:AGSP:epsilon} with $\Pi=\Pi_0$ and $\epsilon=0$.

  Furthermore,
  \begin{equation}
    K\Pi_0 = \frac{\Pi_0 - \lambda_0\Pi_0/\kappa}{2} = \frac{1-\lambda_0/\kappa}{2}
  \end{equation}
  and
  \begin{equation}
    K(\1-\Pi_0)
    \leq \frac{\1-\Pi_0 - (\1-\Pi_0)\lambda_1/\kappa}{2}
    = \frac{1-\lambda_1/\kappa}{2} (\1-\Pi_0).
  \end{equation}
  So $K$ satisfies \cref{def:AGSP}\labelcref{def:AGSP:commutator,def:AGSP:Delta,def:AGSP:Gamma} with $\Gamma$ and $\Delta$ as claimed.
\end{proof}

The AGSP of \cref{AGSP_sum} cannot be implemented locally.
But we can approximate it locally.
\Cref{AGSP-to-AGSP,AGSP_mixture} provide two ways of doing so.

\begin{proposition}\label{AGSP-to-AGSP}
  Let $K = \sum_{i=1}^m k_i$ be a $(\Delta,\Gamma,0)$-AGSP for $\Pi_0$.
  Then
  \begin{align}
    K' = \prod_{i=1}^m\Bigl((1-\epsilon)\1 + \epsilon k_i\Bigr)
         \prod_{i=m}^1\Bigl((1-\epsilon)\1 + \epsilon k_i\Bigr)
  \end{align}
  is a $(\Delta',\Gamma',\epsilon')$-AGSP with
  \begin{align}
    \Delta'   &= (1-\epsilon)^{2m-1}\Bigl(1 - (1-2\sqrt{\Delta})\epsilon\Bigr) + O(\epsilon^2),
      \label{eq:Deltap}\\
    \Gamma'   &= (1-\epsilon)^{2m-1}\Bigl(1-(1-2\sqrt{\Gamma})\epsilon\Bigr) - O(\epsilon^2),
      \label{eq:Gammap}\\
    \epsilon' &= O(\epsilon^2).
      \label{eq:epsilonp}
  \end{align}
\end{proposition}

\begin{proof}
  Let $\tilde{K} := (1-\epsilon)^{2m}\1 + 2\epsilon(1-\epsilon)^{2m-1} K$.
  We have
  \begin{align}
    \tilde{K}\Pi_0
    &= (1-\epsilon)^{2m}\Pi_0 + 2\epsilon(1-\epsilon)^{2m-1} K\Pi_0 \\
    &\geq (1-\epsilon)^{2m-1}\Bigl(1-(1-2\sqrt{\Gamma})\epsilon\Bigr) \Pi_0.
  \end{align}
  Similarly,
  \begin{align}
    \norm{(&1-\Pi_0)\tilde{K}(1-\Pi_0)} \\
           &\leq (1-\epsilon)^{2m} \norm{(1-\Pi_0)\1(1-\Pi_0)} + 2\epsilon(1-\epsilon)^{2m-1}
             \norm{(1-\Pi_0)K(1-\Pi_0)} \\
           &\leq (1-\epsilon)^{2m-1}\Bigl(1 - (1-2\sqrt{\Delta})\epsilon\Bigr).
  \end{align}
  Thus $\tilde{K}$ is a $(\tilde{\Delta},\tilde{\Gamma},0)$-AGSP with
  \begin{align}
    \tilde{\Delta} &= (1-\epsilon)^{2m-1}\Bigl(1 - (1-2\sqrt{\Delta})\epsilon\Bigr), \\
    \tilde{\Gamma} &= (1-\epsilon)^{2m-1}\Bigl(1-(1-2\sqrt{\Gamma})\epsilon\Bigr).
  \end{align}

  Now, $K'$ is Hermitian and
  \begin{equation}
    K' = (1-\epsilon)^{2m}\1 + 2\epsilon(1-\epsilon)^{2m-1} K + O(\epsilon^2),
  \end{equation}
  so that
  \begin{equation}
    \norm{K'-\tilde{K}} = O(\epsilon^2).
  \end{equation}
  Thus by \cref{AGSP_perturbation}, $K'$ is an $(\tilde{\Delta}+O(\epsilon^2),\tilde{\Gamma}-O(\epsilon^2),O(\epsilon^2))$-AGSP, as claimed.
\end{proof}

Applying \cref{AGSP-to-AGSP} to the AGSP from \cref{AGSP_sum} immediately gives the following.

\begin{corollary}\label{AGSP_product}
  If $H = \sum_i h_i$ is a Hamiltonian whose smallest and second-smallest eigenvalues (not counting degeneracies) are $\lambda_0$ and $\lambda_1$, respectively, then
  \begin{equation}
    \begin{gathered}
      K' = \prod_{i=1}^m\Bigl((1-\epsilon)\1 + \epsilon\kappa_i k_i\Bigr)
      \prod_{i=m}^1\Bigl((1-\epsilon)\1 + \epsilon\kappa_i k_i\Bigr) \\
      \text{where}\quad
      k_i = \frac{\1-h_i/\norm{h_i}}{2}, \quad
      \kappa = \sum_i\norm{h_i}, \quad
      \kappa_i = \frac{\norm{h_i}}{\kappa}
    \end{gathered}
  \end{equation}
  is a $(\Delta',\Gamma',\epsilon')$-AGSP with
  \begin{align}
    \Delta'   &= (1-\epsilon)^{2m-1}\left(1 - \frac{\epsilon\lambda_1}{\kappa}\right) + O(\epsilon^2), \\
    \Gamma'   &= (1-\epsilon)^{2m-1}\left(1 - \frac{\epsilon\lambda_0}{\kappa}\right) - O(\epsilon^2), \\
    \epsilon' &= O(\epsilon^2).
  \end{align}
\end{corollary}

\Cref{AGSP_product} gives a way of approximating the AGSP of \cref{AGSP_sum} locally, by weak-measuring local terms in the Hamiltonian in sequence.
The following result gives an alternative way, by repeatedly sampling terms to weak-measure uniformly at random.

\begin{proposition}\label{AGSP-to-CP}
  Let $K = \sum_{i=1}^m k_i$ be an AGSP, and let $K'$ be the AGSP from \cref{AGSP-to-AGSP}.
  Define the (trace-non-increasing) CP map
  \begin{equation}
    \cE(\rho) := \sum_i E_i\rho E_i^\dg, \qquad
    E_i = \frac{1}{\sqrt{m}} \left((1-\epsilon)\1 + \epsilon k_i\right)
  \end{equation}
  Then
  \begin{equation}
    \cE^{2m}(\rho) = K'\rho {K'}^\dg + O(\epsilon^2).
  \end{equation}
\end{proposition}

\begin{proof}
  From \cref{AGSP-to-AGSP}, we have
  \begin{equation}
    K'\rho {K'}^\dg = (1-\epsilon)^{4m} \rho + 2\epsilon(1-\epsilon)^{4m-1} (K\rho + \rho K^\dg) + O(\epsilon^2).
  \end{equation}
  Meanwhile,
  \begin{align}
    \cE(\rho)
    &= \sum_i E_i\rho E_i^\dg \\
    &= \sum_i \frac{1}{m} \Bigl((1-\epsilon)\1 + \epsilon k_i\Bigr) \rho \Bigl((1-\epsilon)\1 + \epsilon k_i^\dg\Bigr) \\
    &= (1-\epsilon)^2\rho + \frac{1}{m}\epsilon(1-\epsilon) (K\rho + \rho K^\dg) + O(\epsilon^2),
  \end{align}
  so that
  \begin{align}
    \cE^{2m}
    &= (1-\epsilon)^{4m}\rho + 2m\frac{1}{m}\epsilon(1-\epsilon)(1-\epsilon)^{2(2m-1)}[K,\rho] + O(\epsilon^2) \\
    &= (1-\epsilon)^{4m}\rho + 2\epsilon(1-\epsilon)^{4m-1}(K\rho + \rho K^\dg) + O(\epsilon^2) \\
    &= K'\rho{K'}^\dg + O(\epsilon^2),
  \end{align}
  as claimed.
\end{proof}

Once again, applying \cref{AGSP-to-CP} to the AGSP from \cref{AGSP_sum} gives a way of approximating the AGSP locally.

\begin{corollary}\label{AGSP_mixture}
  If $H = \sum_i h_i$ is a Hamiltonian whose smallest and second-smallest eigenvalues (not counting degeneracies) are $\lambda_0$ and $\lambda_1$, respectively, then
  \begin{equation}
    \begin{gathered}
      \cE(\rho) := \sum_i E_i\rho E_i^\dg, \qquad
      E_i = \frac{1}{\sqrt{m}} \left((1-\delta)\1 + \epsilon \kappa_i k_i\right) \\
      \text{where}\quad
      k_i = \frac{\1-h_i/\norm{h_i}}{2}, \quad
      \kappa = \sum_i\norm{h_i}, \quad
      \kappa_i = \frac{\norm{h_i}}{\kappa}
    \end{gathered}
  \end{equation}
  is a CPT map which acts as
  \begin{equation}
    \cE^{2m}(\rho) = K'\rho {K'}^\dg + O(\epsilon^2)
  \end{equation}
  where $K'$ is the AGSP from \cref{AGSP_product}.
\end{corollary}

However, note that \cref{AGSP-to-CP,AGSP_mixture} (unlike \cref{AGSP-to-AGSP,AGSP_product}) do \emph{not} construct a local AGSP (as in \cref{def:AGSP}) that approximates $K$.
Rather, they construct a CP map whose action on a density matrix approximates the action of the CP map with Kraus operator $K$.
Thus the results in this work that apply to AGSPs do not apply directly to the maps of \cref{AGSP-to-CP,AGSP_mixture}.
Instead, the results of \cref{sec:fault-resilience} show that these maps approximately recover the same output as the ideal AGSP.

Finally, we note that applying \cref{AGSP_product} or \cref{AGSP_mixture} to \cref{Chebyshev_H,AGSP_H} gives a locally-implementable approximation to the Chebyshev AGSP:

\begin{corollary}\label{AGSP_Chebyshev}
  For the Chebyshev Hamiltonian of \cref{AGSP_H}, \cref{AGSP_product,AGSP_mixture} give a $\bigl(\1-O(\epsilon),\, 4e^{-4\ell\sqrt{\delta/(\norm{H}-\lambda_0)}},\, O(\epsilon^2)\bigr)$-AGSP which can be implemented using $k\ell$-local measurements.
\end{corollary}

\clearpage

\section{Stopping conditions}\label{sec:stopping}

Iterating an AGSP will approximate to increasing accuracy the projection onto the ground state subspace, up to normalisation.
The AGSP can be viewed as the Kraus operator of a completely-positive (CP) map.
However, AGSPs are not trace-preserving, so this does not constitute a full description of a quantum dynamics.
Instead, the AGSP describes one outcome of a quantum instrument (generalised quantum measurement).
That is, one component of a set of CP maps which sum to a completely-positive and trace preserving (CPT) map.
Physically, this corresponds to the fact that an AGSP can be applied probabilistically by quantum measurement, but not necessarily deterministically.
To obtain a full description of a dissipative dynamics, we must extend the AGSP to a full quantum instrument that also describes the other possible measurement outcome(s), so that the overall dynamics is trace-preserving.

In \cref{sec:fixed-point}, I will extend the AGSP to a CPT map, analogous to the dissipative state engineering map of \cite{VWC}.
For a good AGSP, the fixed point of this CPT map is close to the ground state subspace, so the dynamics converges inexorably to the ground state subspace (albeit taking exponentially long to get there in general).
This generalises the results of \cite{VWC}, which were restricted to frustration-free Hamiltonians, to arbitrary Hamiltonians.
However, constructing a good AGSP for this CPT map is itself QMA-hard, as it requires knowing the ground state energy of the Hamiltonian.
So, except in special cases, this fixed-point construction cannot serve as the foundation for a general, practical ground state preparation algorithm.

\Cref{sec:stopped} remedies this by constructing a quantum instrument, rather than a single CPT map.
By keeping track of the measurement outcomes of this instrument, I show that the state can be driven as close as desired to the ground state subspace, without requiring any prior knowledge of the spectral properties of the Hamiltonian.

\subsection{Stopped process ground state preparation}
\label{sec:stopped}
A fixed-point process constructed from the AGSP can be interpreted (and would in practice be implemented) as repeatedly performing a binary generalised quantum measurement, and resampling qudits when the 1~outcome is obtained.
But such a process blindly iterates the CPT map, ignoring all the information potentially provided by the measurement outcomes obtained along the way.

Instead of blindly running until convergence, we can make use of the measurement outcomes to decide when to stop the process.
I will show that suitable choices of stopping rule can guarantee that the state at the stopping time is a good approximation to the ground state.

\begin{theorem}\label{stopped_CPT_map}
  Let $K$ be a $(\Delta,\Gamma,\epsilon)$-AGSP for $\Pi_0$ on a Hilbert space of dimension~$D$, and let $N:=\tr\Pi_0$ be the ground state degeneracy.
  Let $\{\cE_0,\cE_1\}$ be the quantum instrument defined by
  \begin{equation}
    \cE_0(\rho) = K\rho K^\dg, \quad \cE_1(\rho) = \left(1-\tr(K\rho K^\dg)\right) \frac{\1}{D}.
  \end{equation}
  Consider the stopped process whereby we iterate $\{\cE_0,\cE_1\}$, starting from the maximally mixed state $\rho_0=\tfrac{\1}{D}$, until we obtain a sequence of $n$ 0's.

  The state $\rho_n$ at the stopping time satisfies
  \begin{equation}\label{eq:stopped_overlap}
    \tr(\Pi_0\rho_n) \geq 1 - \frac{1-\tr\Pi\rho_0}{\tr\Pi\rho_0} \left(\frac{\Delta}{\Gamma}\right)^n - \epsilon
                     \geq 1 - \epsilon - \frac{D}{N} \left(\frac{\Delta}{\Gamma}\right)^n.
  \end{equation}
\end{theorem}

\begin{proof}
  Let $\Pi$ be as in \cref{def:AGSP}, so $[K,\Pi]=0$.
  Thus
  \begin{equation}
    \tr\left(K^n \rho (K^\dg)^n \Pi\right)
      = \tr\left((K\Pi)^n \rho (\Pi K^\dg)^n\right)
      \geq \Gamma^n \tr\Pi\rho
  \end{equation}
  by \cref{def:AGSP}, and similarly
  \begin{equation}
    \tr\left(K^n \rho (K^\dg)^n (\1-\Pi)\right) \leq \Delta^n (1-\tr\Pi\rho).
  \end{equation}

  For any state $\rho$, after a sequence of $n$ 0's we have
  \begin{equation}\label{eq:stopped_state}
    \rho_n = \frac{\cE_0^n(\rho)}{\tr\cE_0^n(\rho)} = \frac{K^n\rho(K^\dg)^n}{\tr(K^n\rho(K^\dg)^n)}.
  \end{equation}
  So
  \begin{align}
    \tr\Pi\rho_n
    &= \frac{\tr\bigl(K^n\rho (K^\dg)^n \Pi\bigr)}
            {\tr\bigl(K^n\rho(K^\dg)^n \Pi\bigr) + \tr\bigl(K^n\rho (K^\dg)^n (\1-\Pi)\bigr)}\\
    &= 1 - \frac{\tr\bigl(K^n\rho (K^\dg)^n (\1-\Pi)\bigr)}
                {\tr\bigl(K^n\rho(K^\dg)^n \Pi\bigr) + \tr\bigl(K^n\rho (K^\dg)^n (\1-\Pi)\bigr)}\\
    &\geq 1 - \frac{\Delta^n(1-\tr\Pi\rho)}{\Gamma^n \tr\Pi\rho}\\
    &= 1 - \frac{1-\tr\Pi\rho}{\tr\Pi\rho} \left(\frac{\Delta}{\Gamma}\right)^n.
  \end{align}

  Since the state immediately before any run of 0's is $\tfrac{\1}{D}$, we have
  \begin{equation}
    \tr\Pi\rho_n \geq 1 - \frac{1-N/D}{N/D} \left(\frac{\Delta}{\Gamma}\right)^n
      \geq 1 - \frac{D}{N} \left(\frac{\Delta}{\Gamma}\right)^n.
  \end{equation}
  Finally, by \cref{def:AGSP} we have $\norm{\Pi-\Pi_0}\leq\epsilon$, so
  \begin{equation}
    \abs{\tr\Pi_0\rho_n - \tr\Pi\rho_n} \leq \epsilon,
  \end{equation}
  and the \namecref{stopped_CPT_map} follows.
\end{proof}

We also note the following simple result for later reference.

\begin{lemma}\label{expected_stopped_state}
  The expected state at the stopping time in the process of \cref{stopped_CPT_map} is given by
  \begin{equation}
    \E(\rho_n) = \frac{K^{2n}}{\tr(K^{2n})},
  \end{equation}
  hence satisfies
  \begin{equation}
    \E\tr(\Pi_0\rho) = \frac{\tr(\Pi_0K^{2n})}{\tr(K^{2n}}.
  \end{equation}
\end{lemma}

\begin{proof}
  Immediate by substituting $\rho = \1/D$ in \cref{eq:stopped_state}.
\end{proof}

Remarkably, we can also derive an exact analytic expression for the (necessarily exponentially growing with system size) expected stopping time of this process.

\begin{theorem}\label{stopping_time}
  For the process of \cref{stopped_CPT_map}, the expected stopping time---i.e.\ the time until the process first produces a sequence of $n$ 0's---is given by
  \begin{equation}
    \E(\tau_n) = \frac{1}{\tr(K^{2n})} \tr\left(\frac{1-K^{2n}}{1-K^2}\right)
    \leq \frac{1}{\Gamma^n} \left(n + \frac{1-\Delta^n}{1-\Delta\phantom{{}^n}} \left(\frac{D}{N}-1\right)\right).
  \end{equation}
\end{theorem}

\begin{proof}
  Since the state is reset to the maximally mixed state whenever the outcome 1 is obtained, the probability of obtaining another 0 after a run of $k$ 0's since the last 1 is given by
  \begin{equation}
    \Pr(0|0^k1\dots)
    = \tr\left(K \frac{(K^{2k}/D)}{\tr(K^{2k}/D)} K\right)
    = \frac{\tr(K^{2(k+1)})}{\tr(K^{2k})}.
  \end{equation}

  Let $X_t$ denote the $t$'th outcome of the process.
  Define the stochastic process\footnote{%
    The process $M$ can be understood intuitively as the net winnings from the following betting strategy on the sequence of random events $X_t$.
    At time-step $t$, stake all your winnings so far plus an additional $t$ on the outcome $X_t=0$, at fair odds.
    If you indeed get the outcome $0$ then, because the odds are fair, you get back your stake multiplied by a factor determined by the probability of that outcome given everything that's happened so far.
    So your net winnings in this case are the amount you get back, minus the additional $t$ you put in.
    If you get the outcome $1$, you lose your whole stake.
    So your net winnings are~0 minus the additional~$t$ you put in.
    The fact that you're betting at fair odds means that you expect on average to break even, and overall to walk away with~0 net winnings; i.e.\ the process is a Martingale.}
  \begin{equation}\label{eq:Martingale}
    M_0 = 0, \qquad
    M_t = \begin{cases}
              \frac{M_{t-1} + t}{\Pr(X_t=0 | X_{t-1},X_{t-2},\dots,X_1)} - t & X_t = 0 \\
             -t & X_t = 1.
           \end{cases}
  \end{equation}
  Now,
  \begin{align}
    \E(&M_t|X_{t-1},\dots,X_1) \\
    \begin{split}
      &= \Pr(X_t=0|X_{t-1},\dots,X_1) \left(\frac{M_{t-1} + t}{\Pr(X_t=0 | X_{t-1},X_{t-2},\dots,X_1)} - t\right) \\
      &\mspace{20mu} - t Pr(X_t=1|X_{t-1},\dots,X_1)
    \end{split}\\
    &= M_{t-1},
  \end{align}
  so $M$ is a Martingale with respect to $X_t$.
  By Doob's optional stopping theorem, $M_{\tau_n}$ is also a Martingale.
  Therefore, $\E(M_{\tau_n}) = \E(M_0) = 0$.

  But at $\tau_n$, we have just had a run of $n$ 0's since the last 1.
  Thus, by \cref{eq:Martingale},
  \begin{align}
    M_{\tau_n-n} &= -\tau_n+n, \\
    M_{\tau_n-k}   &= \frac{1}{\Pr(0|0^{n-1}1\dots)} \left( M_{\tau_n-k-1} + \tau_n-k \right) - \tau_n+k.
  \end{align}
  Solving this recurrence relation, we obtain
  \begin{align}
    M_{\tau_n}
    &= \frac{1}{\Pr(0|0^{n-1}1\dots)} \Biggl(
         \frac{1}{\Pr(0|0^{n-2}1\dots)} \Biggl(
           \cdots \notag \\
    &\mspace{80mu}
           \cdots \left(
             \frac{1}{\Pr(0|01\dots)} \left(
               \frac{1}{\Pr(0|1\dots)} + 1
             \right) + 1
           \right) \cdots + 1
           \Biggr) + 1
      \Biggr) - \tau_n\\
    &= \frac{\tr(K^{2(n-1)})}{\tr(K^{2n})} \Biggl(
         \frac{\tr(K^{2(n-2)})}{\tr(K^{2(n-1)})} \Biggl(
           \cdots \notag \\
    &\mspace{80mu}
           \cdots \left(
             \frac{\tr(K^2)}{\tr(K^{2\cdot 2})} \left(
               \frac{1}{\tr(K^2)} + 1
             \right) + 1
           \right) \cdots + 1
           \Biggr) + 1
      \Biggr) - \tau_n\\
    &= \frac{\tr\left(\sum_{k=0}^{n-1}K^{2k}\right)}{\tr(K^{2n})} - \tau_n\\
    &= \frac{1}{\tr(K^{2n})} \tr\left(\frac{1-K^{2n}}{1-K^2}\right) - \tau_n.
  \end{align}
  Since $\E(M_{\tau_n}) = 0$, we have
  \begin{equation}
    \E(\tau_n) = \frac{1}{\tr(K^{2n})} \tr\left(\frac{1-K^{2n}}{1-K^2}\right)
  \end{equation}
  as claimed.

  Using $K^2 \geq \Gamma\Pi$ and $K^2 \leq \Pi + \Delta(\1-\Pi)$ from \cref{def:AGSP}, we can bound this as
  \begin{align}
    \E(\tau_n) &= \frac{\tr\left(\sum_{k=0}^{n-1}K^{2k}\right)}{\tr(K^{2n})}
    \leq \frac{\sum_{k=0}^{n-1}\tr\left(\Pi + \Delta^k(\1-\Pi)\right)}{N\Gamma^n} \\
    &\leq \frac{n N + \frac{1-\Delta^n}{1-\Delta\phantom{{}^n}}(D-N)}{N\Gamma^n}
    = \frac{1}{\Gamma^n} \left(n + \frac{1-\Delta^n}{1-\Delta\phantom{{}^n}} \left(\frac{D}{N}-1\right)\right),
  \end{align}
  again as claimed.
\end{proof}

\begin{remark}
  When $K=\Pi$, we have $\Gamma=1$, $\Delta=0$, and \cref{stopping_time} gives $\E(\tau_n) \leq \frac{D}{N}+n-1$.
  I.e.\ when the AGSP is exactly a projector, the expected time to obtain a sequence of $n$ 0's is just the expected number of attempts to project the maximally mixed state onto $\Pi$, plus another $n-1$ steps to deterministically obtain a further $n-1$ 0's, as expected.
\end{remark}

\begin{figure}[!htbp]
  \centering
  \includegraphics[width=\textwidth]{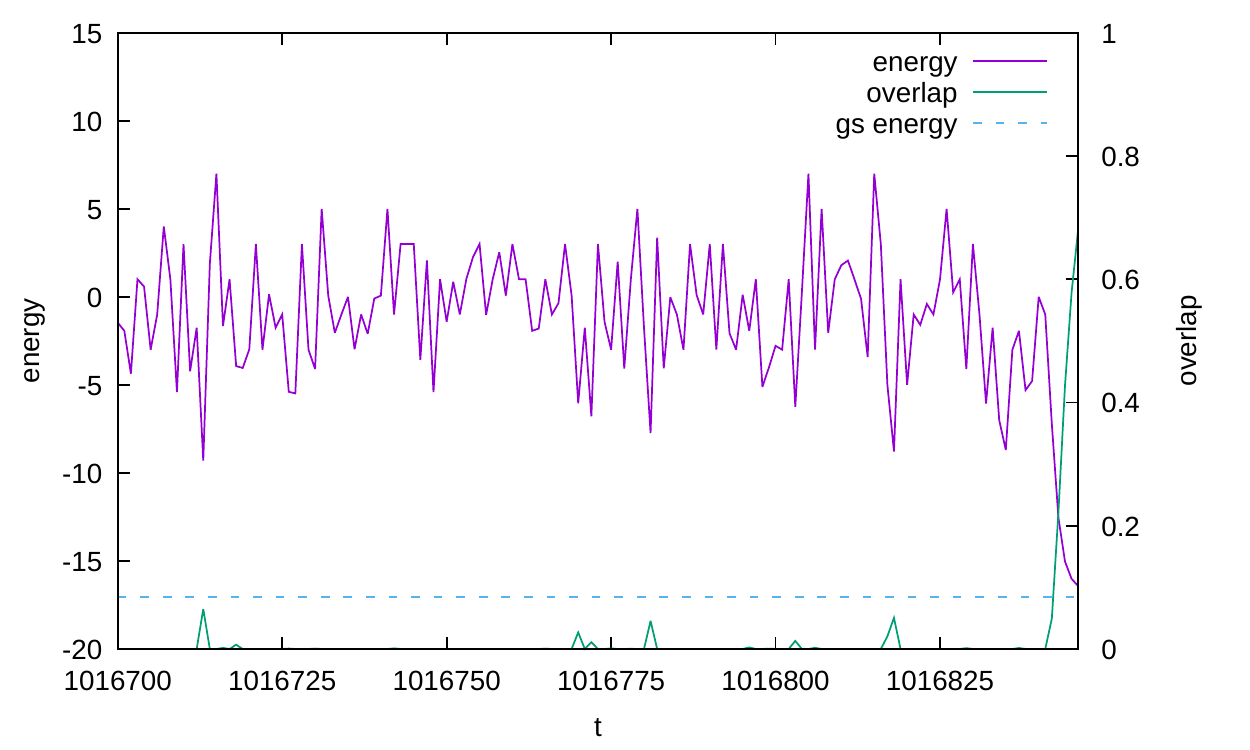}
  \stepcounter{footnote}{\makeatletter\protected@xdef\@thefnmark{\thefootnote}\makeatother}
  \caption{%
    The evolution of the energy (upper purple solid line) and ground state overlap (lower green solid line) in (part of) a typical run of the DQE \cref{DQE} for the 1D Heisenberg chain of length 10, with $\epsilon=0.2$, stopping on the first run of 4~zeros, and using global resampling.\makeatletter\@footnotemark\makeatother
    Although this run of the algorithm took a total of 1016846 iterations before it stopped, the quantum circuit depth required to achieve this is significantly lower.
    Because global resampling discards the entire state whenever the measurement outcome~1 is obtained, coherence of the quantum state need only be maintained for a maximum of 4~iterations to attain an energy close to the ground state energy (dashed blue line).}
  \label{fig:1D_Heisenberg_fixed_global}
\end{figure}
\footnotetext{Numerical simulations performed with the help of the QuEST package~\cite{QuEST}.}

\Cref{stopping_time} proves that the expected run-time scales exponentially, as it must on complexity-theoretic grounds.
However, there is an important distinction here between the total run-time and required quantum circuit depth, which is different from a quantum algorithm that only uses unitary quantum gates.
Since the entire state is discarded and replaced with the maximally mixed state when we obtain the ``1''~outcome (i.e.\ the $\cE_1$ outcome of the quantum instrument of \cref{DQE}), this means that the required \emph{circuit depth} -- i.e.\ the time for which the quantum computer must maintain coherence in order for the algorithm to succeed -- only scales with the number of applications of $\cE_0$ required to get close to the ground state.
I.e.\ it scales as $O(n)$ where $n$ is the number of 0's we stop on in \cref{DQE}, leading directly to the following result.

\begin{corollary}\label{stopped_circuit_depth}
  To obtain an output state $\rho_n$ satisfying
  \begin{equation}
    \tr(\Pi_0\rho_n) \leq 1 - \epsilon
  \end{equation}
  the algorithm of \cref{stopped_CPT_map} requires a circuit depth that scales as
  \begin{equation}
    O\left(\frac{\log\left(\frac{D}{N}\cdot\frac{1}{\epsilon}\right)}
                {\log\left({\Gamma/\Delta}\right)}\right),
  \end{equation}
  i.e.\ linearly in the system size (total number of qudits).\footnote{Recall that $D$ is the total Hilbert space dimension, so $\log D$ scales linearly in the total number of qudits.}
\end{corollary}

\begin{proof}
  Immediate from \cref{stopped_CPT_map}.
\end{proof}

In \cref{sec:resampling}, we will see how the total run-time can be improved by using better resampling strategies than the simplistic global resampling in \cref{stopped_CPT_map,stopping_time} (which amounts to simply discarding the entire state and restarting the process scratch upon failure, and naively repeating until success).
However, more sophisticated resampling strategies do not change the required circuit depth, which is still given by \cref{stopped_circuit_depth} and therefore scales linearly in the system size.

\subsection{Optimally stopped process}
The process of \cref{stopped_CPT_map} is guaranteed to stop on the first sufficiently long run of successful applications of the AGSP.
But this comes at the cost of a non-deterministic run-time, whose expectation scales exponentially in the system size.
In practice, it may be more useful to set a fixed maximum run-time $T$, and stop on the longest run of AGSPs that occurs within time $t<T$.
However, we cannot know what the longest run is until the process finishes, at which point the state produced by that run has already been destroyed.
To address this, we need some results from optimal stopping theory.

In principle, the optimal stopping strategy given knowledge of the measurement outcomes so far can be derived from the AGSP $K$.\footnote{For simplicity, we restrict the exposition to time-independent $K$ here, but the argument generalises straightforwardly.}
Let $X = X_0,X_1,\dots$ denote the sequence of measurement outcomes in \cref{DQE}, starting from initial state $\rho_0$, and let $\rho_X$ denote the state after the sequence of measurement outcomes $X$.
If our goal is to maximise the overlap with the ground state, then the Bellmann equation for the value $v_s$ of the process with maximum run-time $s$ is
\begin{align}
  v_0(\rho_0) &= \tr\Pi_0\rho_0, \\
  \begin{split}
    v_{s+1}(\rho_X)
    &= \max \biggl\{\tr\Pi_0\rho_X,\;
                   \tr(K\rho_X K^\dg)\; v_{t+1}(\rho_{X0})\\
    &\mspace{180mu} + (1-\tr(K\rho_XK^\dg))\; v_{t+1}(\cR(\rho_X)) \biggr\}
  \end{split}
\end{align}

The optimal stopping condition is then given by
\begin{equation}
  \tau = \min\{t\geq 0: v_{T-t}(X) \geq \tr\Pi_0\rho_X\}.
\end{equation}
This follows directly from the standard finite-horizon optimal stopping argument: that the optimal policy at time $t$ is to stop iff the current ground state overlap is at least as large as the expected ground state overlap if we continue for one more step and employ the optimal policy from then on.
This is a dynamic programming problem which can, in principle, be solved numerically for any choice of parameters in \cref{DQE} by back-substitution.

However, this is not particularly useful to us.
It requires computing the full density matrix $\rho_X$ conditioned on any sequence of measurement outcomes, which is at least as hard as computing the ground state itself.
What we require is a stopping policy that only takes into account the information that is readily accessible to us during \cref{DQE}: the measurement outcomes themselves.
In the absence of any knowledge of the spectral properties of the Hamiltonian or AGSP, the only property we know from \cref{stopped_CPT_map} is that the expected ground state overlap is monotonically increasing in the string of 0~outcomes obtained.
We can therefore focus our attention on stopping policies that optimise the length of the string of 0's we stop on.
This is (closely related to) classic problems in optimal stopping theory, and we can derive suitable stopping rules by applying these to the case of \cref{DQE}.

There is a beautiful theory of optimal stopping for iid Bernoulli processes (i.e.\ repeated tosses)~\cite{success_runs,success_runs_note}.
However, in our case, the probabilities of the 0 and 1~outcomes are not iid, but rather depend on the quantum state which itself depends on the history of the measurement outcomes up to now; in particular, with the global resampling strategy considered so far, it depends on the length of the current run of 0's.
Whilst it may be feasible to generalise these results to our case, by focusing on the length of the 0-runs we can instead reduce the problem to simpler, standard optimal stopping problems and still obtain the correct asymptotic scaling.

\subsubsection{Secretary stopping policy}
The secretary problem is perhaps the best-known example of an optimal stopping problem~\cite{secretary}.
In this classic problem, we are tasked with hiring a secretary under the following (somewhat artificial) conditions: the total number of candidates is known in advance; candidates are interviewed one by one, and must either be hired or rejected immediately; only the relative rank of the current candidate can be determined, relative to the candidates seen so far; there are no ties.
We are tasked with hiring the best candidate; selecting any other candidate is counted as a failure.
The optimal strategy for the secretary problem is well known~\cite{secretary}:

\begin{lemma}\label{secretary}
  Let $n$ be the total number of candidates in the secretary problem.
  For any $n$, the optimal stopping policy is a threshold policy: reject the first $\gamma n$ candidates, where $0<\gamma<1$, then accept the first candidate that is better than any seen previously.
  In the limit $n\to\infty$, $\gamma\to e^{-1}$ and the probability of selecting the best candidate also $\to e^{-1}$.
\end{lemma}

We can use this to prove a bounded-time version of the dissipative ground state preparation process.
We will make use of the following standard result, where the notation $f(n) \sim g(n)$ to mean $\lim_{n\to\infty} f(n)/g(n) \to 1$.

\begin{lemma}[{see e.g.\ \cite[Example~V.4]{Analytic_Combinatorics}}]\label{longest_run}
  Consider a sequence of $t$ two-outcome Bernoulli trials with probability $\Pr(0) = p$.
  Let $L_t$ denote the longest run of 0's.
  Then $\mathbbm{E}(L_t) \sim \log_{1/p} t$.
\end{lemma}

\begin{theorem}\label{secretary_CPT_map}
  Fix $t\in\N$ and let $\{\cE_0,\cE_1\}$ be as in \cref{stopped_CPT_map}.
  Consider the stopped process in which we repeatedly measure $\{\cE_0,\cE_1\}$ for $t/e$ steps without stopping, then stop at the first run of 0's that is longer than any previous run, breaking ties uniformly at random.

  In the limit $t\to\infty$, with probability~$\to e^{-1}$ we will stop on a run of zeros of length $n = \Omega(log_{1/\Gamma}(tN/D))$.
  The resulting state $\rho_n$ satisfies \cref{eq:stopped_overlap}.
\end{theorem}

\begin{proof}
  A given run of the process will have some sequence of~0 runs (possibly including some runs of zero length), each terminated by a~1.
  Since the state at the start of any 0 run is identical (namely the maximally mixed state), the probability of a 0-run depends only on its length.
  Thus the sequence of 0-run-lengths is a sequence of iid random variables, each chosen according to the same distribution determined by the AGSP.
  If we rank longer runs as better than shorter ones, breaking ties at random, we can uniquely rank each run relative to the others.
  If we were given a sequence of run-lengths one-by-one, and restrict ourselves to only making use of the relative rank of a run (not the length of the run itself), then selecting the longest run reduces to the secretary problem.

  Since the length of each run is iid, the expected number of 0-runs occurring in any time interval depends only on the length of that interval.
  Thus the expected proportion of 0-runs that occur within the first $t/e$ time steps therefore tends to $1/e$ as $t\to\infty$.

  Therefore, under the stopping policy of \cref{secretary_CPT_map}, we expect to discard the first $e^{-1}$ fraction of the total number of runs of 0's, and then select the first run longer than any seen previously.
  I.e.\ as $t\to\infty$, the stopping policy converges in expectation to the stopping policy of \cref{secretary}, and $\Pr(\text{stop on longest run}) \to e^{-1}$.

  However, unlike the situation in \cref{secretary}, we are not presented with the 0-runs in their entirety, but sequentially, one~0 at a time.
  We cannot know if the latest~0 will be the end of the current 0-run until we see the next 1~outcome, at which point it is too late to stop on that 0-run.
  Instead, we stop as soon as we encounter a run of 0's longer than any seen previously.
  Thus we still stop within the same 0-run as we would have, had we been able to see the whole run before deciding to stop.
  However, the length of the 0-run when we stop will not necessarily be as long as it could have been, and the payoff is slightly altered.
  Rather than stopping on the longest 0-run within $t$ steps, we instead end up stopping on a run 1~longer than the longest run observed within the first $t/2$ steps.

  Now, the probability of a run of $n$ 0's is $\tr(K^{2n})/D \geq \Gamma^n N/D$ by \cref{def:AGSP}.
  The latter is equal to the probability of first getting heads on a biased coin with $\Pr(\text{heads})=N/D$, and if one gets heads then running a sequence of Bernoulli trials with $p=\Gamma$ and obtaining a run of $n$ successes.
  The expected length of time required to obtain the initial head before starting a Bernoulli trial sequence is $D/N$.
  Thus the expected length of the longest 0-run within $t/e$ steps is $\geq \log_{1/\Gamma}(tN/eD)$ and the expected length of the run we stop on is $n \geq \log_{1/\Gamma}(tN/eD) + 1 = \Omega(\log_{1/\Gamma}(tN/D))$.

  The claim on $\rho_n$ follows by \cref{stopped_CPT_map}.
\end{proof}

\begin{remark}
  \Cref{secretary_CPT_map} tells us that we can obtain a sequence of $n$ 0's with constant probability in run-time $t = O(D/N\Gamma^n)$.
  Which is consistent with the expected time to obtain a sequence of $n$ 0's from \cref{stopping_time}, which goes as $D/N\Gamma^n$ to leading order.
\end{remark}

\subsubsection{Expected run-length stopping policy}
Although the stopping policy in \cref{secretary_CPT_map} is optimal for the secretary problem we reduced it to, we can still do somewhat better for the ground state preparation task (even without any additional knowledge of the AGSP).
Rather than optimising the probability of stopping on the longest run of 0's, what we would really like to do is maximise the expected length of the 0-run we stop at.
The optimal stopping policy for the minimum-expected-rank variant of the secretary problem is also known~\cite{Chow1964}:

\begin{lemma}[\cite{Chow1964}]
  Let $n$ be the total number of candidates in the secretary problem.
  Define recursively
  \begin{align}
    s_i = \left\lfloor\frac{i+1}{n+1} c_i\right\rfloor, \qquad
    c_{i-1} = \frac{1}{i}\left(\frac{t+1}{i+1}\cdot\frac{s_i(s_i+1)}{2} + (i-s_i)c_i\right), \qquad s_n = n.
  \end{align}
  For any $n$, the following stopping policy minimises the expected rank: stop at the first candidate $k\geq 1$ such that $y_k\leq s_k$, where $y_k$ is the relative rank of the $k$'th candidate relative to the all the candidates seen so far.

  The expected rank is given by $c_0$, and $\lim_{n\to\infty} c_0 = \prod_{k=1}^\infty\left(\frac{k+2}{k}\right)^{1/k+1} \mbox{$< 3.8695$}$.
\end{lemma}

Employing this stopping policy gives, by similar arguments to \cref{secretary_CPT_map}, the following variant.

\begin{theorem}\label{expectation_CPT_map}
  Fix $t\in\N$ and let $\{\cE_0,\cE_1\}$ be as in \cref{stopped_CPT_map}.
  Consider the stopped process in which we repeatedly measure $\{\cE_0,\cE_1\}$ and stop on the $i$'th run of 0's, where $i$ is the first 0-run such that $y_i\leq s_i$, up to a maximum of $t$ runs.
  Here, $y_i$ is the relative rank of the $i$'th 0-run relative to all previous 0-runs, ranking the runs by length (longer is better) and breaking ties randomly.

  In the limit $t\to\infty$, we will stop on a run of zeros of length $n = \Omega(\log_{1/\Gamma}(tN/D))$.
  The resulting state $\rho_n$ satisfies \cref{eq:stopped_overlap}.
\end{theorem}

\subsection{Fixed-point ground state preparation}
\label{sec:fixed-point}

In the case of frustration-free Hamiltonians, \cite{VWC} showed how to construct a CPT map whose fixed point lies in the ground state subspace.
The rate of convergence to the fixed point depends on spectral properties of the Hamiltonian, and is in general exponential in the system size even for Hamiltonians with a constant spectral gap.
But the fixed point itself is guaranteed to be a ground state as long as the Hamiltonian is frustration-free.

Here, I will construct something similar, but for arbitrary (i.e.\ possibly frustrated) Hamiltonians.
As in \cite{VWC}, the rate of convergence to the fixed point will depend on spectral properties of the Hamiltonian, via the AGSP used to construct the CPT map.
However, in contrast to \cite{VWC}, the proximity of the fixed point itself to the ground state subspace now also depends on these parameters.
In particular, it depends on how close the parameter $\Gamma$ in \cref{def:AGSP} is to~1.

For frustration-free Hamiltonians, a rescaled version of the Hamiltonian itself is an AGSP with $\Gamma=1$, recovering results analogous to those of~ \cite{VWC}.
For frustrated Hamiltonians, the more complex transformation of \cref{Chebyshev_H} can give AGSPs with $\Gamma$ as close to~1 as desired, at a cost of slower convergence time and requiring increasingly complex local measurements.

However, note that this transformation of the Hamiltonian can only be constructed if the ground state energy is known in advance.
The ground state energy of frustration-free Hamiltonians can always be taken to be~0 without loss of generality by rescaling and shifting, so this is consistent with the frustration-free results.
But for general Hamiltonians, approximating the ground state energy is itself QMA-hard; it is still NP-hard even for gapped, commuting Hamiltonians.
So although \cref{Chebyshev_H} together with the results of this section imply that, for arbitrary Hamiltonians, there exists a local CPT map whose fixed point is the ground state for arbitrary Hamiltonians, this CPT map cannot be constructed efficiently in general.
We will remedy this in \cref{sec:stopped}.

\begin{theorem}\label{CPTP_map}
  Let $K$ be a $(\Delta,\Gamma,\epsilon)$-AGSP for $\Pi_0$ on a Hilbert space of dimension~$D$, and let $N:=\tr\Pi_0$ be the ground state degeneracy.
  Define the CPT map
  \begin{equation}
    \cE(\rho) = K\rho K^\dg + \left(\tr\rho-\tr(K\rho K^\dg)\right) \frac{\1}{D}.
  \end{equation}
  The fixed points $\rho_\infty$ of $\cE$ satisfy:
  \begin{align}
    \tr(\Pi_0\rho_\infty)
    &\geq \frac{1-\Delta}{(\Gamma-\Delta) + \frac{D}{N}(1-\Gamma)} - \epsilon \\
    &= 1 - \frac{1-\Gamma}{1-\Delta} \left(\frac{D}{N}-1\right)
             \frac{1}{1 + \frac{1-\Gamma}{1-\Delta} \left(\frac{D}{N}-1\right)} - \epsilon.
  \end{align}
\end{theorem}

\begin{proof}
  Let $\Pi$ be as in \cref{def:AGSP}, so $[K,\Pi]=0$.
  Then,
  \begin{align}
    \tr(K\rho K^\dg)
    &= \tr\left(K(\Pi+\1-\Pi)\rho(\Pi+\1-\Pi)K^\dg\right) \\
    &= \tr(K\Pi\rho\Pi K^\dg) + \tr\left((\1-\Pi) K (\1-\Pi) \rho (\1-\Pi) K^\dg (\1-\Pi) \right) \\
    &\leq \tr(K\Pi\rho\Pi K^\dg) + \Delta(1-\tr(\Pi\rho)).
  \end{align}
  Noting that $\tr\Pi = \tr\Pi_0 = N$ by \cref{AGSP_perturbation}, for any state $\rho$ we have
  \begin{align}
    \tr(\Pi\cE(\rho))
    &= \tr(\Pi K\rho K^\dg)
      + \left(1-\tr(K\rho K^\dg)\right) \tr\left(\frac{\Pi}{D}\right) \\
    &\geq \tr(K\Pi\rho\Pi K^\dg)
      + \frac{N}{D}\Bigl(1 - \tr(K\Pi\rho\Pi K^\dg) - \Delta\bigl(1-\tr(\Pi\rho)\bigr)\Bigr) \\
    &= \left(1-\frac{N}{D}\right) \tr(K\Pi\rho\Pi K^\dg)
      + \frac{N}{D}\Bigl(1 - \Delta\bigl(1-\tr(\Pi\rho)\bigr)\Bigr) \\
    &\geq \left(1-\frac{N}{D}\right) \Gamma\tr(\Pi\rho)
      + \frac{N}{D}\Bigl(1 - \Delta\bigl(1-\tr(\Pi\rho)\bigr)\Bigr) \\
    &= \left(\Gamma-\frac{N}{D}(\Gamma-\Delta)\right)\tr(\Pi\rho)
      + \frac{N}{D}(1-\Delta).
  \end{align}
  For a fixed point $\rho_\infty$, we obtain
  \begin{equation}
    \tr(\Pi\rho_\infty) = \tr(\Pi\cE(\rho_\infty))
    \geq \left(\Gamma-\frac{N}{D}(\Gamma-\Delta)\right)\tr(\Pi\rho_\infty)
      + \frac{N}{D}(1-\Delta),
  \end{equation}
  thus
  \begin{equation}
    \tr(\Pi\rho_\infty) \geq \frac{1-\Delta}{(\Gamma-\Delta) + \frac{D}{N}(1-\Gamma)}.
  \end{equation}

  Since $\norm{\Pi-\Pi_0}\leq \epsilon$ by \cref{def:AGSP}\labelcref{def:AGSP:epsilon}, we have
  \begin{align}
    \tr(\Pi_0\rho_\infty) &= \tr(\Pi\rho_\infty) + \tr\left((\Pi-\Pi_0)\rho_\infty\right) \\
    &\geq \frac{1-\Delta}{(\Gamma-\Delta) + \frac{D}{N}(1-\Gamma)} - \epsilon \\
    &= 1 - \frac{1-\Gamma}{1-\Delta} \left(\frac{D}{N}-1\right)
             \frac{1}{1 + \frac{1-\Gamma}{1-\Delta} \left(\frac{D}{N}-1\right)} - \epsilon,
  \end{align}
  as required.
\end{proof}

\begin{remark}
  Taking $\epsilon=0$ for simplicity:
  \begin{itemize}
  \item If $\Gamma=1$ and $\Delta<1$, then $\tr(\Pi_0\rho_\infty)=1$, i.e.\ the fixed points are exact ground states.
  \item If $\Gamma = \Delta$, then $\tr(\Pi_0\rho_\infty)=\frac{N}{D}$, i.e.\ the overlap of the fixed points with the ground state subspace is no better than that of the maximally mixed state, i.e.\ no better than guessing a state uniformly at random.
  \item If $\Gamma = 1-\delta$, then $\tr(\Pi_0\rho) \geq 1 - \frac{D/N}{1-\Delta}\delta$, i.e.\ if $\Delta=1-O(1)$ then the fixed point is $O(\delta)$-close to a ground state, albeit with a constant prefactor of order the Hilbert space dimension.
  \end{itemize}
\end{remark}

In fact, we can derive an exact expression for the fixed point of this process in terms of its AGSP.

\begin{theorem}\label{fixed-point}
  Let $K$ be a $(\Delta,\Gamma,\epsilon)$-AGSP for $\Pi_0$ with $0\leq \Delta \leq \Gamma < 1$, and
  \begin{equation}
    \cE(\rho) = K\rho K^\dg + \left(\tr\rho-\tr(K\rho K^\dg)\right) \frac{\1}{D}
  \end{equation}
  the CPT map from \cref{CPTP_map}.
  Then $\cE$ has a unique fixed point given by
  \begin{equation}
    \rho_\infty = \frac{(\1-K^2)^{-1}}{\tr\left((\1-K^2)^{-1}\right)}.
  \end{equation}
\end{theorem}

\begin{proof}
  Fixed points $\rho_\infty$ of $\cE$ satisfy
  \begin{equation}
    \rho_\infty = \cE(\rho_\infty) = K\rho_\infty K^\dg +  \frac{\tr\rho_\infty - \tr(K\rho_\infty K^\dg)}{D} \1
  \end{equation}
  or, equivalently,
  \begin{equation}\label{eq:Lyapunov}
    K\rho_\infty K^\dg - \rho_\infty + c\1 = 0, \qquad c = \frac{\tr\rho_\infty - \tr(K\rho_\infty K^\dg)}{D}.
  \end{equation}

  This is a discrete Lyapunov equation, so has a unique solution given by
  \begin{equation}\label{eq:rho_infty}
    \rho_\infty = \sum_{n=0}^\infty K^n\, c\1\, (K^\dg)^n = c \sum_{n=0}^\infty K^{2n} = c\,(\1-K^2)^{-1},
  \end{equation}
  recalling that $K$ is Hermitian by \cref{def:AGSP}.

  Since we want $\tr\rho_\infty=1$, we must have that $c=1/\tr(\1-K^2)^{-1}$.
  To see this explicitly, note that
  \begin{equation}
    K\rho_\infty K^\dg = c\,K \Bigl(\sum_{n=0}^\infty K^{2n}\Bigr) K
      = c\sum_{n=1}^\infty K^{2n}
      = c\Bigl(\sum_{n=0}^\infty K^{2n} - \1\Bigr)
      = c(\1-K^2)^{-1} - c\1,
  \end{equation}
  so, from the expression for $c$ in \cref{eq:Lyapunov},
  \begin{align}
    cD = \tr\rho_\infty - \tr\left(K\rho_\infty K^\dg\right)
       = 1 - c\tr(\1-K^2)^{-1} + cD.
  \end{align}
  Hence
  \begin{equation}
    c = \frac{1}{\tr\left((\1-K^2)^{-1}\right)}.
  \end{equation}
\end{proof}

The following result is immediate from \cref{fixed-point} and the properties of the AGSP $K$ in \cref{def:AGSP}.

\begin{corollary}\label{fixed-point_expr}
  Let $\cE$ be as in \cref{fixed-point}, with $\Pi$ the projector for $K$ from \cref{def:AGSP}.
  The fixed point $\rho_\infty$ of $\cE$ has the form
  \begin{equation}
    \rho_\infty = \frac{X_0\oplus X_\perp}{\tr X_0 + \tr X_\perp}
    \quad \text{where} \quad
    X_0 \geq \frac{\1}{1-\Gamma}, \;
    X_\perp \leq \frac{\1}{1-\Delta}
  \end{equation}
  where $X_0$ is $N$-dimensions, $X_\perp$ $D-N$-dimensional, and the direct sum is with respect to the $\Pi,\1-\Pi$ partition of the Hilbert space.
\end{corollary}

Using the exact expression for the fixed point from \cref{fixed-point,fixed-point_expr}, we have
\begin{align}
  \tr(\Pi\rho_\infty)
  &= 1 - \tr\bigl((\1-\Pi)\rho_\infty\bigr)
  = 1 - \frac{\tr(X_\perp)}{\tr X_0 + \tr X_\perp}
  = 1 - \frac{1}{\frac{\tr X_0}{\tr X_\perp} + 1} \\
  &\geq 1 - \frac{1}{\frac{N/(1-\Delta)}{(D-N)/(1-\Gamma)} + 1}
  = 1 - \frac{1-\Gamma}{1-\Delta} \left(\frac{D}{N}-1\right)
          \frac{1}{1 + \frac{1-\Gamma}{1-\Delta} \left(\frac{D}{N}-1\right)},
\end{align}
i.e.\ we recover exactly the bound from \cref{CPTP_map}.
Thus this bound is the tightest possible in terms of the AGSP parameters.

Using the AGSP of \cref{AGSP_Chebyshev} in \cref{fixed-point}, we immediately obtain the following:

\begin{corollary}\label{Chebyshev_fixed-point}
  Let $H=\sum_i h_i$ be a $k$-local Hamiltonian on a Hilbert space $\C^D$ with ground state projector $\Pi_0$, ground state energy $\lambda_0$ and spectral gap $\delta$.
  In \cref{fixed-point}, choose $K$ to be the AGSP from \cref{AGSP_Chebyshev} and choose $\epsilon = N/D$.
  Then the CPT map $\cE$ from \cref{fixed-point} can be implemented by $k\ell$-local generalised measurements, and its fixed points $\rho_\infty$:
  \begin{equation}
    \tr(\Pi_0\rho_\infty)
    \geq 1 - \frac{\epsilon}{1-4e^{-4\ell\sqrt{\delta(\norm{H}-\lambda_0)}}} \left(\frac{D}{N}-1\right).
  \end{equation}
\end{corollary}

Since $\epsilon$ can be chosen as small as desired and $\ell$ can also be chosen freely, this means that for any local Hamiltonian $H$, including frustrated ones, there exists a local CPT map whose fixed point is as close as desired to the ground state of $H$.
Since we take $\epsilon = O(D/N)$, the convergence time will scale with the total Hilbert space dimension $D$, i.e.\ exponentially in the number of qudits, consistent with the complexity-theoretic considerations discussed in \cref{sec:introduction}.
However, even though \cref{Chebyshev_fixed-point} shows that this CPT map \emph{exists}, we cannot construct it efficiently just from the description of the local Hamiltonian $H = \sum_i h_i$.
Because we need to know the ground state energy and spectral gap of $H$ before we can construct the AGSP of \cref{AGSP_H}.

Contrast this with the results of \cite{VWC}, which show that there exists a CPT map that efficiently prepares ground states of frustration-free Hamiltonians, but to construct that CPT map requires knowing an MPS or PEPS description of the ground state.
In the case of \cref{Chebyshev_fixed-point}, even were we given the CPT map, it does not construct the ground state efficiently.
But this is perhaps not surprising, as the AGSP of \cref{AGSP_H} only uses information about the spectrum of $H$, not a full description of the ground state itself as in \cite{VWC}.
Even knowing the full spectrum of $H$ does not necessarily make finding a description of its ground state easy.
Indeed, for general Hamiltonians (unlike in the case of MPS or PEPS ground states), an efficient, explicit description of the ground state may not even exist.

\clearpage

\section{Epsilon schedules}\label{sec:epsilon}
For \cref{stopped_CPT_map} to succeed, $\epsilon$ must be chosen sufficiently small that $\Gamma'>\Delta'$ in \cref{AGSP_product,AGSP_mixture}.
However, $\Gamma'$ and $\Delta'$ ultimately depend on the spectrum of the Hamiltonian, which one typically does not know.
So it is not clear what $\epsilon$ one should choose when applying the algorithm to a given Hamiltonian, unless one has prior knowledge of its spectral properties.

Similar issues occur in most ground state preparation algorithms.
The standard approach is to make some assumption on the spectrum of $H$, e.g.\ a lower-bound on the spectral gap and/or the density of states~\cite{AQC}, and show that the algorithm succeeds on Hamiltonians satisfying those assumptions.
As $\Gamma-\Delta$ is directly related to the spectral gap of $H$ via \cref{AGSP_sum}, in our case any lower-bound on the spectral gap allows a sufficiently small $\epsilon$ to be chosen to guarantee success.

However, determining the spectral gap of a Hamiltonian is in general harder than finding the ground state in the first place! (The spectral gap problem is $P^{UQMA[\log n]}$-hard~\cite{Ambainis} for Hamiltonians on $n$ particles, and even becomes undecidable in the asymptotic limit~\cite{spectral-gap_short,spectral-gap_long,spectral-gap_1D}.)
Furthermore, as there is in general no way to determine whether a given quantum state has good overlap with the ground state having some information about the spectrum of the Hamiltonian, it is not possible to run such algorithms with some guess at suitable parameters, and then verify at the end whether they succeeded or not.

For the dissipative quantum eigensolver of \cref{DQE}, we can do better.
We will prove that if, instead of a constant $\epsilon$, we choose it according to a suitable decreasing sequence, then \cref{DQE} will always converge to the ground state of the Hamiltonian, without any assumptions on or prior knowledge about the Hamiltonian.

\begin{lemma}\label{Kn_bounds}
  Let $K_t$ be a family of $(\Gamma_t,\Delta_t,\epsilon_t)$-AGSPs for the same ground state projector $\Pi_0$.
  Let $K=\prod_{t=n}^1 K_t$.
  Then
  \begin{align}
    \tr\left(K \rho K^\dg\Pi_0\right)
      &\geq \tr(\rho\Pi_0)\prod_{t=1}^n \Gamma_t - 2\sum_{t=1}^n\epsilon_t\prod_{s=t+1}^n\Gamma_s\\
    \tr\left(K\rho K^\dg(1-\Pi_0)\right)
      &\leq \tr(\rho(\1-\Pi_0))\prod_{t=1}^n \Delta_t + 2\sum_{t=1}^n\epsilon_t\prod_{s=t+1}^n\Delta_s.
  \end{align}
\end{lemma}

\begin{proof}
  Denote $K^{(k)}=\prod_{t=k}^1 K_t$.
  Noting that $\norm{\Pi_t-\Pi_s} \leq \norm{\Pi_t-\Pi_0} + \mathmbox{\norm{\Pi_s-\Pi_0}} \leq \epsilon_t + \epsilon_s$, we have
  \begin{align}
    \tr\left(K^{(k)} \rho K^{(k)}\Pi_0\right)
    &= \tr\left(K_k K^{(k-1)} \rho K^{(k-1)} K_k\Pi_0\right)\\
    &\geq \tr\left(K_k K^{(k-1)} \rho K^{(k-1)} K_k\Pi_k\right) - \epsilon_k\\
    &= \tr\left((K_k\Pi_k) K^{(k-1)} \rho K^{(k-1)} (K_k\Pi_k)\right) - \epsilon_k\\
    &\geq \Gamma_k \tr\left(K^{(k-1)} \rho K^{(k-1)} \Pi_k\right) - \epsilon_k \label{eq:Kn_inequality}\\
    &\geq \Gamma_k \tr\left(K^{(k-1)} \rho K^{(k-1)} \Pi_0\right) - 2\epsilon_k,
  \end{align}
  where we have used \cref{def:AGSP} in \cref{eq:Kn_inequality} and $\Gamma_k \leq 1$ in the final line.
  Applying this recursively to $\tr(K\rho K^\dg\Pi_0) = \tr(K^{(n)}\rho K^{(n)}\Pi_0)$ gives the first inequality in the \namecref{Kn_bounds}.
  The second inequality follows by a very similar argument.
\end{proof}

We will need the following result.

\begin{lemma}\label{sum_prod}
  For any $0 < c < 1/2$,
  \begin{equation}
    \sum_{t=1}^\infty \frac{1}{t^2}\prod_{s=1}^t\frac{1}{1-\frac{c}{s}} \leq e^{2c}\zeta\left(2(1-c)\right),
  \end{equation}
  where $\zeta(s)$ is the Riemann zeta-function.
\end{lemma}

\begin{proof}
  For  $c < 1/2$, we have
  \begin{align}
    \prod_{s=1}^t\frac{1}{1-\frac{c}{s}}
    &\leq \prod_s\left(1+\frac{2c}{s}\right)
    \leq \prod_s e^{2c/s}
    = e^{2c H_t}
    \leq e^{2c(\ln t + 1)}
    = e^{2c} t^{2c},
  \end{align}
  where $H_t$ is the $t$'th Harmonic number and we have used $H_t < \ln t + 1$.
  Thus
  \begin{equation}
    \sum_{t=1}^\infty \frac{1}{t^2}\prod_{s=t+1}^n\frac{1}{1-\frac{c}{s}}
    \leq e^{2c}\sum_{t=1}^\infty \frac{1}{t^{2(1-c)}}
    = e^{2c}\zeta\left(2(1-c)\right),
\end{equation}
  as claimed.
\end{proof}

We will also need the following standard fact from analysis (see e.g.~\cite{Leonard_lecture_notes}).
\begin{theorem}\label{convergent_product}
  If $\epsilon_t\geq 0$, then $\prod_{t=1}^\infty (1-\epsilon_t) = \text{const} > 0$ converges iff $\epsilon_t$ is summable (i.e.\ $\lim_{n\to\infty}\sum_{t=1}^n\epsilon_t$ converges).
  Conversely, if $\lim_{n\to\infty}\sum_{t=1}^n\epsilon_t$ diverges, then $\prod_{t=1}^\infty(1-\epsilon_t) = 0$.
\end{theorem}

We are now in a position to prove that the DQE \cref{DQE} converges to the ground state unconditionally, for suitable choices of the parameters.

\begin{theorem}\label{decaying_CPT_map}
  Let $H=\sum_i h_i$ be a $k$-local Hamiltonian on a Hilbert space $\C^D$ with ground state projector $\Pi_0$.
  In \cref{DQE}, choose $K$ to be the AGSP from \cref{AGSP_product,AGSP_mixture}, and
  $\epsilon_t = \epsilon/(t-t_1)$, where $t_1$ is the time step at which the last outcome~1 was obtained.\footnote{Here, when we say a 0~outcome was obtained in a given time step $t$, we mean that the $\cE_{i,0}^{(t)}$ outcome was obtained \emph{for all} $i$ at that time step in \cref{DQE}. Conversely, a 1~outcome is obtained if the outcome $\cE_{i,1}^{(t)}$ is obtained for \emph{any} $i$.}

  Choose any initial state $\rho_0$ such that $\tr\Pi_0\rho_0>0$ (e.g.\ $\rho_0=\1/D$), and choose any $\cR$ such that $\cR(\rho)$ has full support if $\rho$ does (e.g.\ $\cR(\rho) = \1/D$).
  Let $\tau$ be one of the stopping rules from \cref{stopped_CPT_map,secretary_CPT_map,expectation_CPT_map} (setting some maximum allowed run-time $t$ in the \cref{stopped_CPT_map} case, and choosing $n$ increasing with $t$).

  Then there is some sufficiently small constant $\epsilon$ such that the state $\rho_\tau$ at the stopping time after running \cref{DQE} for at most $t$ steps satisfies
  \begin{equation}
    \lim_{t\to\infty}\tr(\Pi_0\rho_\tau) = 1.
  \end{equation}
\end{theorem}

\begin{proof}
  The argument is reminiscent of \cref{stopped_CPT_map}.

  Denote $K^{(k)} = \prod_{t=k}^1 K_t$ as before.
  For any state $\rho_0$ at the start of a run of $n$ 0's, the state at the end of the run is
  \begin{equation}
    \rho_n = \frac{K^{(n)}\rho_0{K^{(n)}}^\dg}{\tr(K^{(n)}\rho_0 {K^{(n)}}^\dg)}.
  \end{equation}

  Recall from \cref{AGSP_product} that
  \begin{align}
    \Delta'   &= (1-\epsilon)^{2m-1}\Bigl(1 - (1-2\sqrt{\Delta})\epsilon\Bigr) + O(\epsilon^2),\\
    \Gamma'   &= (1-\epsilon)^{2m-1}\Bigl(1-(1-2\sqrt{\Gamma})\epsilon\Bigr) - O(\epsilon^2),\\
    \epsilon' &= O(\epsilon^2).
  \end{align}
  Using \cref{Kn_bounds},
  \begin{align}
    \tr\Pi_0\rho_n
    &= 1 - \frac{\tr\bigl(K^{(n)}\rho_0{K^{(n)}}^\dg (\1-\Pi_0)\bigr)}
                {\tr\bigl(K^{(n)}\rho_0{K^{(n)}}^\dg \Pi_0\bigr)
                 + \tr\bigl(K^{(n)}\rho_0{K^{(n)}}^\dg (\1-\Pi_0)\bigr)} \\
    &\geq 1 - \frac{(1-\tr\Pi_0\rho_0) \prod_{t=1}^n\Delta'_t + 2\sum_{t=1}^n \epsilon'_t \prod_{s=t+1}^n\Delta'_s}
                   {\tr\Pi_0\rho_0 \prod_{t=1}^n\Gamma'_t - 2\sum_{t=1}^n\epsilon'_t \prod_{s=1}^n\Gamma'_s} \\
    &= 1 - \frac{(1-\tr\Pi_0\rho_0) + 2\sum_{t=1}^n \epsilon'_t \prod_{s=1}^t\frac{1}{\Delta'_s}}
                {\tr\Pi_0\rho_0 - 2\sum_{t=1}^n\epsilon'_t \prod_{s=1}^t\frac{1}{\Gamma'_s}}
           \left(\prod_{t=1}^n\frac{\Delta'_t}{\Gamma'_t}\right) \\
    &= 1 - O\left(
             \frac{(1-\tr\Pi_0\rho_0)
                   + 2 \sum_{t=1}^\infty \frac{\epsilon^2}{t^2}
                       \prod_{s=1}^t\left(1-\frac{\epsilon}{t}\right)^{1-2m}
                                   \left(1-\frac{\lambda_0\epsilon}{\kappa t}\right)^{-1}}
                  {\tr\Pi_0\rho_0
                   + 2 \sum_{t=1}^\infty \frac{\epsilon^2}{t^2}
                       \prod_{s=1}^t\left(1-\frac{\epsilon}{t}\right)^{1-2m}
                                   \left(1-\frac{\lambda_1\epsilon}{\kappa t}\right)^{-1}}
             \prod_{t=1}^n\frac{1-\frac{\epsilon\lambda_1}{\kappa t}}
                              {1-\frac{\epsilon\lambda_0}{\kappa t}}
           \right) \\
    &\geq 1 - O\left(
                \frac{(1-\tr\Pi_0\rho_0)
                      + \sum_{t=1}^\infty \frac{\epsilon^2}{t^2}
                        \prod_{s=1}^t\left(1-\frac{(2m-1+\lambda_0/\kappa)\epsilon}{t}\right)^{-1}}
                     {\tr\Pi_0\rho_0
                      + \sum_{t=1}^\infty \frac{\epsilon^2}{t^2}
                        \prod_{s=1}^t\left(1-\frac{(2m-1+\lambda_1/\kappa)\epsilon}{t}\right)^{-1}}
                \prod_{t=1}^n\left(1 - \frac{(\lambda_1-\lambda_0)\epsilon}
                                           {\kappa t - \lambda_0\epsilon} \right)
              \right) \\
    &\geq 1 - O\left(\frac{(1-\tr\Pi_0\rho_0) + \epsilon^2 e^{2c_0}\zeta(2(1-c_0))}
                          {\tr\Pi_0\rho_0 - \epsilon^2 e^{2c_1}\zeta(2(1-c_1))}
                     \prod_{t=1}^n\left(1 - \frac{(\lambda_1-\lambda_0)\epsilon}
                                                {\kappa t - \lambda_0\epsilon} \right)
              \right)
  \end{align}
  where we have used \cref{sum_prod} in the final line with $c_{0/1} := (2m+1+\lambda_{0/1}/\kappa)\epsilon$ and have assumed $c_{0/1} < 1/2$, i.e.
  \begin{equation}\label{eq:epsilon_bound1}
    \epsilon < (4m+2+2\lambda_1/\kappa)^{-1}.
  \end{equation}

  Now, for
  \begin{equation}\label{eq:epsilon_bound2}
    \epsilon < \sqrt{\frac{\tr(\Pi_0\rho_0)}{2^{2c_1}\zeta(2(1-c_1))}}
  \end{equation}
  the first denominator is lower-bounded by some positive constant, so that
  \begin{equation}\label{eq:decaying_overlap}
    \tr\Pi_0\rho_n
    \geq 1 - O\left(\prod_{t=1}^n\left(1 - \frac{(\lambda_1-\lambda_0)\epsilon}
        {\kappa t - \lambda_0\epsilon} \right) \right).
  \end{equation}
  This holds for $\epsilon$ a sufficiently small constant such that both the conditions on $\epsilon$ from \cref{eq:epsilon_bound1,eq:epsilon_bound2} are satisfied.

  All of the stopping rules from \cref{stopped_CPT_map,secretary_CPT_map,expectation_CPT_map} stop after a run of some number $n$ of 0's, with $n\to\infty$ as the total run-time $t\to\infty$.
  Since $\lim_{n\to\infty}\sum_{k=1}^n \frac{a}{k-b}$ diverges for any constants $a,b>0$, and we have $\lambda_1-\lambda_0 > 0$ by definition (see \cref{AGSP_sum}), \cref{convergent_product} implies that the product in \cref{eq:decaying_overlap} tends to~0 in the limit $n\to\infty$.
  Thus, for sufficiently small constant $\epsilon$, we have $\lim_{t\to\infty}\tr\Pi_0\rho_\tau = \lim_{n\to\infty}\tr\Pi_0\rho_n = 1$, as claimed.
\end{proof}

\clearpage

\section{Resampling strategies}\label{sec:resampling}
In \cref{sec:stopping}, we derived expressions and bounds for the expected ground space overlap and run-time of the DQE \cref{DQE} in the simplest case of global resampling $\cR(\rho) = \1/D$, where the entire state is discarded and replaced by the maximally mixed state after each wrong measurement outcome.
Global resampling allows clean analytical results.
It also has some practical benefits in implementations on near-term quantum computers: it only requires coherence of the quantum state to be maintained during each individual run of consecutive 0's, but not across multiple such runs.

However, discarding the entire state and starting from scratch each time the wrong measurement outcome is obtained is inefficient.
Resampling just the subsystem upon which that measurement acted, is likely to be a more efficient strategy in terms of the overall run-time.
Or just resampling one of the qubits that measurement acted on, or not doing any resampling at all,\footnote{For frustrated Hamiltonians, resampling is not necessarily required to avoid getting stuck in a higher energy state.} or other variations on the resampling strategy.

In this section, I generalise \cref{stopped_CPT_map,stopping_time} to arbitrary resampling strategies $\cR$ in \cref{DQE}.
In the general case, analytical results are expressed in terms of more complex expressions involving the transfer matrices of the CP maps involved, which cannot readily be reduced to upper-bounds on the run-time in terms of the parameters of the AGSP.
However, these analytic results still admit exact numerical calculations of the expected run-time.
\Cref{fig:1D_Heisenberg_resampling} gives an example of these calculations, where it is clearly seen that local resampling performs more efficiently than global resampling.

\subsection{Expected ground space overlap}\label{sec:expected_rho}

\begin{theorem}\label{local_resampling}
  Let $K$ be a $(\Delta,\Gamma,\epsilon)$-AGSP for $\Pi_0$ on a Hilbert space of dimension~$D$, and let $N:=\tr\Pi_0$ be the ground state degeneracy.
  Let $\{\cE_0,\cE_1\}$ be the quantum instrument defined by
  \begin{equation}
    \cE_0(\rho) = K\rho K^\dg, \quad \cE_1(\rho) = \left(1-\tr(K\rho K^\dg)\right) \cR(\rho)
  \end{equation}
  where the resampling map $\cR$ is any CPT map satisfying the condition \mbox{$\tr(\cE_0\circ\cR(\rho)) > 0$}.

  Consider the stopped process whereby we iterate $\{\cE_0,\cE_1\}$, starting from the state $\rho_0$, until we obtain a sequence of $n$ 0's.
  The expected state at the stopping time is
  \begin{equation}\label{eq:local_expected_state}
    \E\kett{\rho_n} = E_0^n W^{-1} \kett{\rho_0}
    \quad\text{where}\quad
    W := \1 - E_1 \frac{\1 - E_0^n}{\1-E_0},
  \end{equation}
  hence satisfies
  \begin{equation}\label{eq:local_expected_overlap}
    \E\tr(\Pi_0\rho_n)
    = \braa{\1} (\Pi_0\ox\Pi_0) E_0^n W^{-1} \kett{\rho_0}
  \end{equation}
\end{theorem}

\begin{proof}
  Since the process stops on the first run of $n$ 0's, the sequence of measurement outcomes that led to $\rho_n$ could be any sequence of 0's and 1's that does not contain a run of $n$ or more zeros, followed by a 1 and a final run of $n$ 0's.
  Let $\mathcal{W}^{\langle n\rangle}$ denote the collection of binary strings starting with a 1 that do \emph{not} contain $n$ consecutive 0's.
  Note that, as a formal sum,
  \begin{equation}
    \sum_{\mathclap{s\in\mathcal{W}^{\langle k\rangle}}} s = \sum_{k=0}^\infty\biggl(1\sum_{j=0}^{n-1}0^j\biggr)^k.
  \end{equation}
  Therefore, letting the $i$'th bit $s_i$ of bit-string $s\in\{0,1\}^*$ denote the outcome of the $i$'th measurement of the quantum instrument, we have
  \begin{align}\label{eq:expected_rho}
    \E\rho_n
    &= \sum_{\mathclap{s\in \mathcal{W}^{\langle n\rangle}}} \Pr(0^n s) \,
       \frac{\cE_0^n\circ\cE_{s_1}\circ\cE_{s_2}\circ\cdots\circ\cE_{s_{\abs{s}}}(\rho_0)}
       {\tr(\cE_0^n\circ\cE_{s_1}\circ\cE_{s_2}\circ\cdots\circ\cE_{s_{\abs{s}}}(\rho_0))}\\
    &= \sum_{\mathclap{s\in \mathcal{W}^{\langle n\rangle}}}
       \cE_0^n\circ\cE_{s_1}\circ\cE_{s_2}\circ\cdots\circ\cE_{s_{\abs{s}}}(\rho_0) \\
    &= \cE_0^n \circ \sum_{k=0}^\infty \biggl(\cE_1 \circ\sum_{j=0}^{n-1}\cE_0^j\biggr)^k (\rho_0).
    \intertext{So}
    \E\kett{\rho_n}
    &= E_0^n \sum_{k=0}^\infty \biggl(E_1 \sum_{j=0}^{n-1}E_0^j\biggr)^k \kett{\rho_0}
    = E_0^n \left(\1 - E_1 \frac{\1-E_0^n}{\1-E_0} \right)^{-1} \kett{\rho_0},
  \end{align}
  where $E_{0,1}$ are the transfer matrices corresponding to $\cE_{0,1}$.

  By linearity of expectation, $\E\tr(\Pi_0\rho_n) = \tr(\Pi_0\,\E\rho_n)$ and the \namecref{local_resampling} follows.
\end{proof}

We also can rewrite \cref{eq:local_expected_overlap} in a form that is more convenient for numerical calculations (as it doesn't require explicitly computing matrix inverses, and prioritises matrix-vector multiplications over matrix-matrix multiplications).

\begin{corollary}
  The state $\ket{\rho_n}$ at the stopping time in \cref{local_resampling} equivalently satisfies
  \begin{equation}
    \E\tr(\Pi_0\rho_n) = \braa{\1} (\Pi_0\ox\Pi_0) E_0^n (\1-E_0)(\1-E_0-E_1+E_1E_0^n)^{-1} \kett{\rho_0}.
  \end{equation}
\end{corollary}

\begin{lemma}\label{expected_normalisation}
  The expected density matrix $\E\rho_n$ from \cref{local_resampling} is normalised,\linebreak i.e.\ $\tr(\E\rho_n) = \braakett{\1|\E\rho_n} = 1$.
\end{lemma}

\begin{proof}
  To see this, note that in \cref{eq:expected_rho} we are summing over measurement outcomes which eventually end in a run of $n$ 0's.
  The probability of this is:
  \begin{equation}
    \Pr(\text{eventually obtain a run of } n \text{ 0's})
    = 1-\Pr(\text{never obtaining a run of } n \text{ 0's}).
  \end{equation}
  Thus, denoting $\rho_s := \cE_{s_{\abs{s}}}\circ\cdots\circ\cE_{s_2}\circ\cE_{s_1}(\rho_0)$ where $s$ is a binary string, we have
  \begin{align}
    \tr(\E\rho_n)
    &= \tr\left[\sum_{s\in \mathcal{W}^{\langle n\rangle}} \rho_{0^n s}\right]
     = 1 - \tr\left[\sum_{s\in(\mathcal{W}^{\langle n\rangle})^\infty}\mspace{-10mu} \rho_s\right] \\
    &= 1 - \tr\biggl[\biggl(\cE_1\circ\sum_{j=0}^{n-1}\cE_0^j\biggr)^\infty (\rho_0)\biggr] \\
    &=: 1 - \tr\left(\cE^\infty(\rho_0)\right) \label{eq:trace_E(rho_n)}
  \end{align}
  where we denote the CP trace-non-increasing map $\cE := \cE_1\circ\sum_{j=0}^{n-1}\cE_0^j$.

  Note that, for any state $\rho$, we have
  \begin{align}
    \tr\left(\cE_1\circ\cE_0^k(\rho)\right)
    &= \tr\left(\cE_1\circ\cE_0^k(\rho) + \cE_0\circ\cE_0^k(\rho) - \cE_0^{k+1}(\rho)\right) \\
    &= \tr\left((\cE_0+\cE_1)\circ\cE_0^k(\rho) - \cE_0^{k+1}\right) \\
    &= \tr\left(\cE_0^k(\rho) - \cE_0^{k+1}\right), \label{eq:E1_E0}
  \end{align}
  where the final equality follows from the fact that $\cE_0+\cE_1$ is trace-preserving.
  Applying \cref{eq:E1_E0} repeatedly, we obtain
  \begin{align}
    \tr(\cE(\rho))
    &= \tr\biggl(\cE_1\circ\sum_{j=0}^{n-1}\cE_0^j(\rho)\biggr) \\
    &= \tr\biggl(\cE_1\circ\sum_{j=0}^{n-2}\cE_0^j(\rho) + \cE_0^{n-1}(\rho) - \cE_0^n(\rho)\biggr) \\
    &= \tr\biggl(\cE_1\circ\sum_{j=0}^{n-3}\cE_0^j(\rho) + \cE_0^{n-2}(\rho) - \cE_0^n(\rho)\biggr) \\
    &= \tr\biggl(\cE_1(\rho) - \cE_0^n(\rho)\biggr) \\
    &= 1 - \tr\left(\cE_0^n(\rho)\right).
  \end{align}

  If $\tr(\cE_0^n(\rho)) = 1$ (hence $\cE_0(\rho) = 1$), then the process of \cref{local_resampling} will immediately produce a string of $n$ 0's with probability~1, and the claim in the \namecref{expected_normalisation} is trivially satisfied. Therefore, assume that $\tr(\cE_0(\rho)) < 1$.

  If $\tr(\cE_0(\rho)) = 0$, then the process must produce a 1~outcome initially:
  \begin{equation}
    \cE(\rho) = \cE_1\circ\sum_{j=0}^{n-1}\cE_0^j(\rho) = \cE_1(\rho).
  \end{equation}
  In this case, $\cE_0\circ\cE(\rho) = \tr(\cE_0\circ\cE_1(\rho)) > 0$ by the condition on $\cR$ in \cref{local_resampling}.
  Otherwise,
  \begin{equation}
    \tr(\cE(\rho)) = 1-\tr\left(\cE_0^n(\rho)\right) < 1.
  \end{equation}

  In either case, we have that $\tr(\cE^n(\rho)) < 1$ (for $n\geq 2$ in the case $\tr(\cE_0(\rho)) = 0$).
  If the spectral radius of $\cE$ were $\geq 1$, this would contradict \cref{spectral_radius_eigenvalue}.
  Thus $\cE$ must have spectral radius~$<1$, hence $\cE^\infty = 0$.
  Therefore, from \cref{eq:trace_E(rho_n)}, we have
  \begin{equation}
    \tr(\E\rho_n) = 1 - \tr\left(\cE^\infty(\rho_0)\right) = 1,
  \end{equation}
  as claimed.
\end{proof}

\begin{corollary}\label{run-time_local-global}
  For global resampling $\cR(\rho) = \1/D$ with initial state $\rho_0=\1/D$, \cref{local_resampling} reduces (as it should) to \cref{expected_stopped_state}.
\end{corollary}

\begin{proof}
  Recall that for $\cR(\rho)=\1/D$ we have $\cE_1(\rho) = \bigl(1-\tr\cE_0(\rho)\bigr) \1/D$, or equivalently
  \begin{equation}
    E_1 = \kettbraa{\tfrac{\1}{D}}{\1}(\1-E_0).
  \end{equation}
  Thus $W^{-1}$ becomes
  \begin{align}
    W^{-1}\kett{\tfrac{\1}{D}}
    &:= \left(\1 - E_1\frac{\1-E_0^n}{\1-E_0}\right)^{-1} \kett{\tfrac{\1}{D}} \\
    &= \sum_{k=0}^\infty\left(E_1\frac{\1-E_0^n}{\1-E_0}\right)^k \kett{\tfrac{\1}{D}} \\
    &= \sum_{k=0}^\infty\left(\kettbraa{\tfrac{\1}{D}}{\1}(\1-E_0)\frac{\1-E_0^n}{\1-E_0}\right)^k
       \kett{\tfrac{\1}{D}} \\
    &= \sum_{k=0}^\infty\Bigl(\kettbraa{\tfrac{\1}{D}}{\1}(\1-E_0^n)\Bigr)^k \kett{\tfrac{\1}{D}} \\
    &= \kett{\tfrac{\1}{D}} \sum_{k=0}^\infty\Bigl(\braa{\1}(\1-E_0^n)\kett{\tfrac{\1}{D}} \Bigr)^k \\
    &= \kett{\tfrac{\1}{D}} \sum_{k=0}^\infty\Bigl(1 - \tfrac{1}{D}\tr(K^{2n})\Bigr)^k \\
    &= \frac{1}{\tr(K^{2n})} \kett{\1}. \label{eq:Winv_global}
  \end{align}
  The \namecref{run-time_local-global} follows by substituting this identity in \cref{eq:local_expected_state} of \cref{local_resampling} and recalling that $\cE_0^n(\1) = K^{2n}$.
\end{proof}

\subsection{Expected run-time}
To derive an expression for the expected stopping time, we will need a few preliminary results.
Since \cref{expected_normalisation} holds for any initial state $\rho_0$, it can equivalently be restated as:
\begin{corollary}\label{E0Wi_TP}
  The map
  \begin{equation}
    \cN := \cE_0^n \circ \sum_{k=0}^\infty \biggl(\cE_1 \circ\sum_{j=0}^{n-1}\cE_0^j\biggr)^k,
  \end{equation}
  whose corresponding transfer matrix is $E_0^n W^{-1}$, is trace-preserving.
\end{corollary}
Intuitively, this is obvious: it simply says that if we allow arbitrarily many attempts, we are guaranteed to eventually get a run of $n$ 0's.

Similarly, the following result expresses that we are guaranteed to eventually get a 1~outcome:
\begin{lemma}\label{E1E0i_TP}
  The map $\cO := \cE_1\circ\sum_{k=0}^\infty\cE_0^k$, whose corresponding transfer matrix is $E_1(\1-E_0)^{-1}$, is trace-preserving.
\end{lemma}

\begin{proof}
  For any state $\rho_0$,
  \begin{align}
    \tr\left(\cO(\rho_0)\right)
    &= \Pr(\text{eventually obtaining a 1}) \\
    &= 1 - \Pr(\text{never obtaining a 1}) \\
    &= 1 - \tr\left(\cE_0^\infty(\rho_0)\right)
     = 1,
  \end{align}
  where the final equality follows since $\cE_0$ is trace-decreasing, hence the spectral radius of $\cE_0$ is $<1$ by \cref{spectral_radius_eigenvalue}.
\end{proof}

This allows us to prove that more complex sequences of 0's and 1's are eventually guaranteed to occur.
\begin{lemma}\label{sequence_probability}
  \begin{equation}
    \braa{\1} E_0^n W^{-1} E_1 (\1-E_0)^{-1} E_0^n W^{-1} \kett{\rho_0} = 1.
  \end{equation}
\end{lemma}

\begin{proof}
  With $\cN$, $\cO$ the maps from \cref{E0Wi_TP,E1E0i_TP}, we have
  \begin{equation}
    \braa{\1} E_0^n W^{-1} E_1 (\1-E_0)^{-1} E_0^n W^{-1} \kett{\rho_0}
    = \tr\left(\cN\circ\cO\circ\cN(\rho_0)\right) = 1.
  \end{equation}
\end{proof}

With this and the results of \cref{sec:expected_rho}, we are in a position to derive an analytic expression for the expected stopping time of \cref{local_resampling} under arbitrary resampling maps $\cR$ in $\cE_1$.

\begin{theorem}\label{local_resampling_run-time}
  For the process of \cref{local_resampling}, the expected stopping time is given by
  \begin{align}\label{eq:local_expected_run-time}
    \E(\tau_n) &= \braa{\1} E_0^n W^{-1} E_1 \left(\frac{\1-E_0^n}{(\1-E_0)^2}\right) W^{-1} \kett{\rho_0}
  \end{align}
  where
  \begin{equation}
    W := \1 - E_1 \frac{\1 - E_0^n}{\1-E_0}.
  \end{equation}
\end{theorem}

\begin{proof}
  We use a generating function approach this time.\footnote{Cf.\ the Martingale approach used to prove \cref{stopping_time}, which is more elegant but doesn't generalise so easily to arbitrary resampling maps $\cR$.}
  Expressing the expected run-time as a derivative of a generating function in the standard way (with $\mathcal{W}^{\langle n\rangle}$ as in \cref{local_resampling}):
  \begin{align}
    \E(\tau_n)
    &=\sum_{\mathclap{s\in \mathcal{W}^{\langle n\rangle}}} \left(n+\abs{s}\right) \, \Pr(0^n s) \\
    &=\sum_{\mathclap{s\in \mathcal{W}^{\langle n\rangle}}}
      \tr\left[
        \left(n+\abs{s}\right)
        \cE_0^n\circ\cE_{s_1}\circ\cE_{s_2}\circ\cdots\circ\cE_{s_\abs{s}}(\rho_0)
      \right] \\
    &=\tr\left[
        \cE_0^n \circ \sum_{k=0}^\infty
        \biggl(\cE_1 \circ\sum_{j=0}^{n-1}\left(n+1+j\right)\cE_0^j\biggr)^k (\rho_0)
      \right] \\
    &=\frac{\dd}{\dd x}\left.
      \tr\left[
        x^n\cE_0^n \circ \sum_{k=0}^\infty
        \biggl(x\cE_1 \circ\sum_{j=0}^{n-1} x^j \cE_0^j\biggr)^k (\rho_0)
      \right]\right|_{x=1} \\
    &=\braa{\1} \left.\frac{\dd}{\dd x}\left[
        x^n E_0^n \sum_{k=0}^\infty \biggl(x E_1 \sum_{j=0}^{n-1} x^j E_0^j\biggr)^k
      \right] \right|_{x=1} \!\! \kett{\rho_0}
      \\
    &=\braa{\1} \frac{\dd}{\dd x}\left.\left[
        x^n E_0^n \left(\1 - x E_1 \frac{\1 - x^n E_0^n}{\1 - x E_0} \right)^{-1}
      \right]\right|_{x=1}\!\! \kett{\rho_0}.
  \end{align}

  Making use of the fact from matrix calculus that $\frac{\dd}{\dd x} A^{-1} = -A^{-1}\frac{\dd A}{\dd x}A^{-1}$, we have
  \begin{align}
    \E(\tau_n)
    &= \braa{\1} \frac{\dd}{\dd x}\left.\left[
      x^n E_0^n \left(\1 - x E_1 \frac{\1 - x^n E_0^n}{\1 - x E_0} \right)^{-1}
    \right]\right|_{x=1}\!\! \kett{\rho_0} \\
    &=\braa{\1}
        n E_0^n W^{-1} + E_0^n W^{-1} \left(E_1 \frac{\1 - E_0^n - n (\1-E_0)E_0^n}{(\1-E_0)^2}\right) W^{-1}
      \kett{\rho_0} \\
    &=\braa{\1} E_0^n W^{-1} E_1 \left(\frac{\1-E_0^n}{\1-E_0}\right)
        \Bigl(\1 - (E_0+E_1) + E_1 E_0^n\Bigr)^{-1}
      \kett{\rho_0} \\
    \begin{split}
      &= n \braakett{\1|\E\rho_n} - n \braa{\1} E_0^n W^{-1} E_1 (\1-E_0)^{-1} E_0^n W^{-1} \kett{\rho_0} \\
      &\phantom{=}\;\mspace{145mu}
        + \braa{\1} E_0^n W^{-1} E_1 \left(\frac{\1-E_0^n}{(\1-E_0)^2}\right) W^{-1} \kett{\rho_0}
    \end{split}\label{eq:reduce_to_Erho}\\
    &=\braa{\1} E_0^n W^{-1} E_1 \left(\frac{\1-E_0^n}{(\1-E_0)^2}\right) W^{-1} \kett{\rho_0},
      \label{eq:normalised_cancellation}
  \end{align}
  where \cref{eq:reduce_to_Erho} follows from \cref{local_resampling}, and \cref{eq:normalised_cancellation} from
  \cref{expected_normalisation,sequence_probability}.
\end{proof}

As in the case of \cref{local_resampling}, we can also rewrite \cref{eq:local_expected_run-time} in a form that is more convenient for numerical calculations.

\begin{corollary}
  The expected runtime $\E(\tau_n)$ in \cref{eq:local_expected_run-time} can equivalently be written as
  \begin{equation}
    \begin{split}
      \E(\tau_n)
      &= \braa{\1} E_0^n (\1-E_0)(\1-E_0-E_1+E_1E_0^n)^{-1} E_1 (\1-E_0^n)\cdot \\
      &\mspace{170mu}
         (\1-E_0)^{-1} (\1-E_0-E_1+E_1E_0^n)^{-1} \kett{\rho_0}.
    \end{split}
  \end{equation}
\end{corollary}

\begin{corollary}
  For global resampling $\cR(\rho) = \1/D$ with initial state $\rho_0 = \1/D$, \cref{local_resampling_run-time} reduces (as it should) to \cref{stopping_time}.
\end{corollary}

\begin{proof}
  Using the identity shown in \cref{eq:Winv_global}, namely that for global resampling
  \begin{equation}
    W^{-1}\kett{\tfrac{\1}{D}} = \frac{1}{\tr(K^{2n})} \kett{\1},
  \end{equation}
  together with \cref{expected_normalisation} in \cref{local_resampling_run-time}, and recalling that $\cE_0(\1) = K^2$, we obtain
  \begin{align}
    \E(\tau_n)
    &= \braa{\1} E_0^n W^{-1} E_1 \left(\frac{\1-E_0^n}{(\1-E_0)^2}\right) W^{-1} \kett{\rho_0} \\
    &= \braa{\1} E_0^n W^{-1} \kett{\tfrac{\1}{D}}
       \braa{\1}(\1-E_0) \frac{\1-E_0^n}{(\1-E_0)^2} W^{-1} \kett{\tfrac{\1}{D}} \\
    &= \frac{1}{\tr(K^{2n})} \braakett{\1}{\E\rho_n}
       \braa{\1}\frac{\1-E_0^n}{\1-E_0} \kett{\1} \\
    &= \frac{1}{\tr(K^{2n})} \braa{\1}\sum_{k=0}^\infty (E_0^k - E_0^{n+k})\kett{\1} \\
    &= \frac{1}{\tr(K^{2n})} \tr\left(\sum_{k=0}^\infty K^{2k} - K^{2n+2k}\right) \\
    &= \frac{1}{\tr(K^{2n})} \tr\left((\1-K^{2n})\sum_{k=0}^\infty K^{2k}\right) \\
    &= \frac{1}{\tr(K^{2n})} \tr\left(\frac{\1-K^{2n}}{\1-K^2}\right),
  \end{align}
  which exactly matches the result of \cref{stopping_time}.
\end{proof}

\begin{figure}[!htbp]
  \centering
  \begin{subfigure}{.9\textwidth}
    \includegraphics[width=\textwidth]{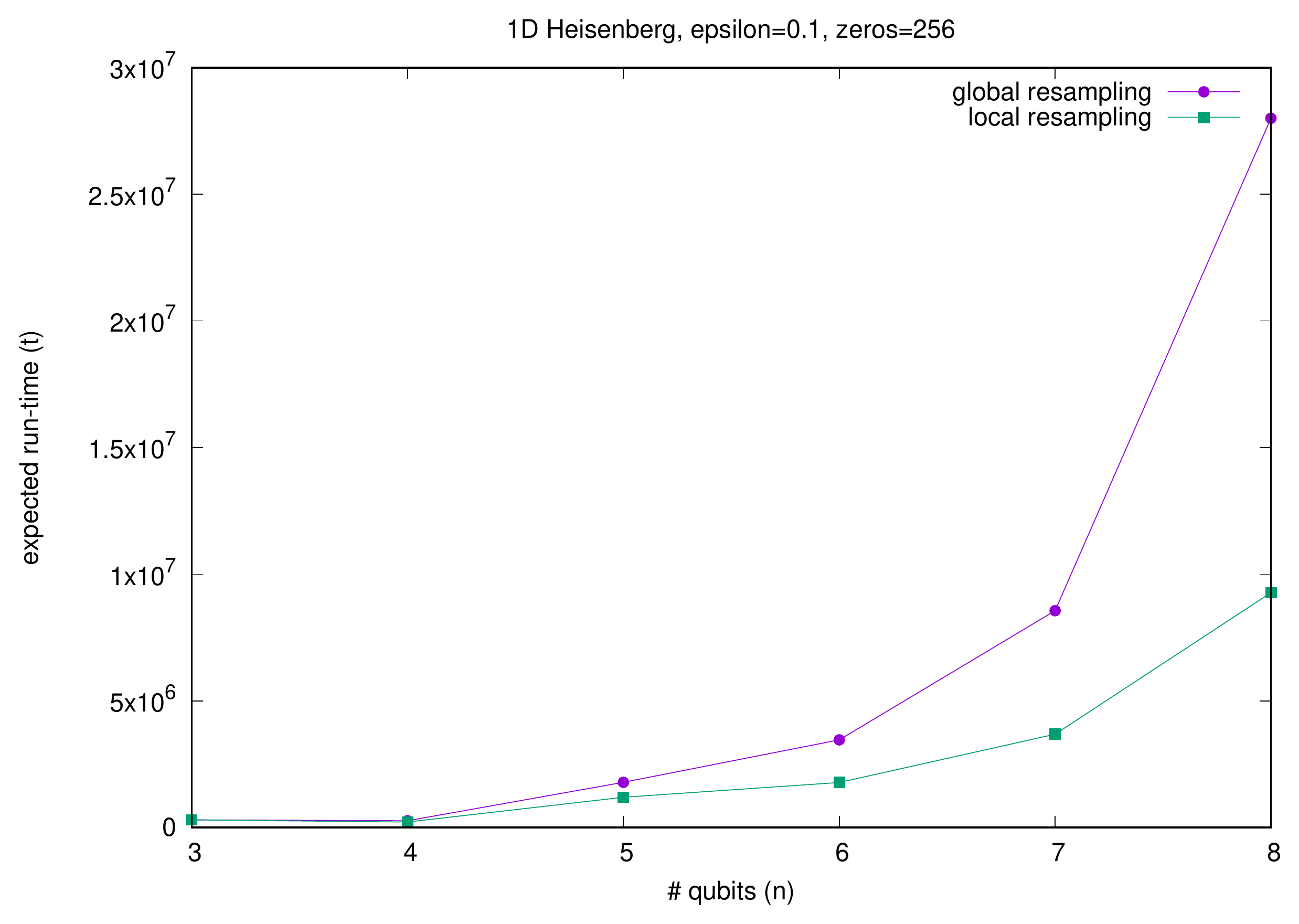}
    \caption{%
      Expected run-time ($t$) as a function of system-size ($n$) with global (purple circles) and local (green squares) resampling for the length $n$ 1D Heisenberg chain, stopping on the first run of 256 consecutive 0's in \cref{stopped_CPT_map} and using \cref{AGSP_mixture} with $\epsilon=0.1$ as the AGSP.
      The plotted values are calculated exactly using \cref{stopping_time} and \cref{local_resampling_run-time}, respectively, with the exception of the $n=8$ points.
      For $n=8$, the approximation $E_0^{256} \approx \Pi_0$ was used, as exact computation of the matrix power becomes too computationally expensive at this system size.
      (Note that $n=256$ is chosen significantly larger here than required in practice (see \cref{fig:1D_Heisenberg_fixed_global}, where $n=4$ sufficed) in order to illustrate the scaling more clearly.
      Thus the magnitudes of the run-times shown here are not indicative of actual run-times.)
    }
    \label{fig:1D_Heisenberg_resampling:lin}
  \end{subfigure}\\[1em]
  \begin{subfigure}{.9\textwidth}
    \includegraphics[width=\textwidth]{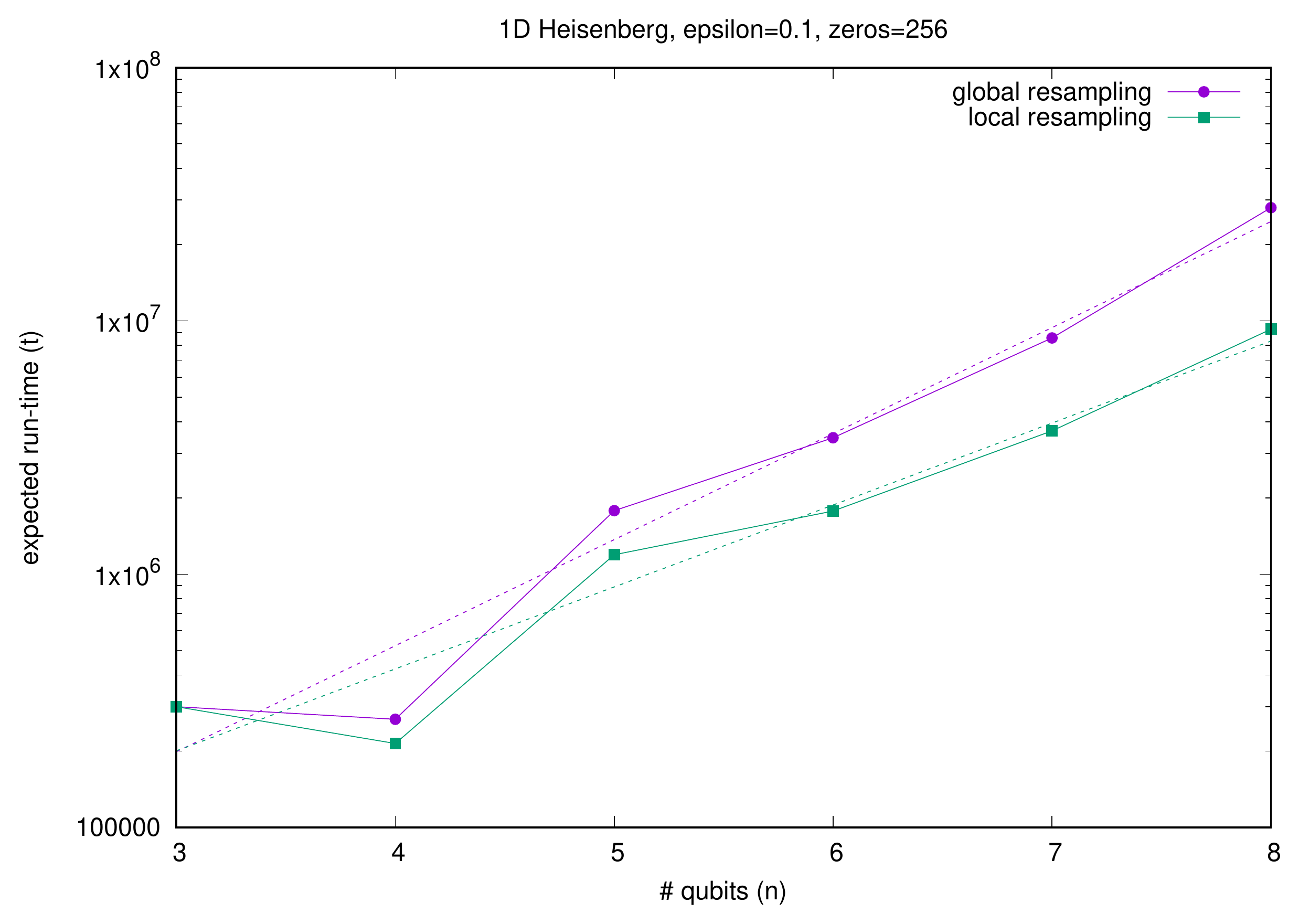}
    \caption{%
      Log-linear fits (dashed lines) to the run-time scaling.
      The run-time scales roughly exponentially in both cases, as expected, with some anomalies at small system sizes.
      But local resampling achieves a lower exponent than global resampling.
      Best fit of the form $a e^{m t}$ for global resampling: %
      $m = 0.96$, 
      $a = 9.31$; 
      local resampling: %
      $m = 0.75$, 
      $a = 9.97$. 
    }
    \label{fig:1D_Heisenberg_resampling:log}
  \end{subfigure}
  \label{fig:1D_Heisenberg_resampling}
\end{figure}

Global resampling discards the entire state on a 1~outcome, so it is clear that the required quantum circuit depth (as opposed to the total run-time) scales favourably as shown in \cref{stopped_circuit_depth}.
The results of \cref{fig:1D_Heisenberg_resampling} indicate that more sophisticated resampling strategies, such as local resampling, can reduce the total run-time of DQE.
However, as the state is no longer completely discarded on a 1~outcome, the required circuit depth would appear to now be the same as the total run-time, i.e.\ exponential rather than polynomial in the system size.
However, it is again misleading to equate total run-time of a dissipative process with depth of a unitary quantum circuit.
Just as with global resampling, for the DQE \cref{DQE} to succeed, coherence of the quantum state need only be maintained for the number of applications of $\cE_0$ required to get close to the ground state.
Local resampling -- or any other resampling strategy -- only changes the expected number of iterations required to see that outcome.
Thus the required circuit depth for the DQE \cref{DQE} is given by \cref{stopped_circuit_depth}, independent of the resampling strategy.

\clearpage

\section{Error- and Fault-Resilience}\label{sec:fault-resilience}

In the presence of noise and without any error correction, the probability of obtaining the correct output of a computation decreases exponentially with time, limiting the maximum run-time of computations to some constant determined by the error rate.
Error correction enables a success probability arbitrarily close to~1 in the presence external noise, independent of the length of the computation or the error rate, assuming the computational steps themselves can be carried out perfectly.
Fault-tolerance enables the success probability to be maintained even when the computational steps themselves are faulty and themselves introduce errors, as long the error rate is below some threshold.
However, quantum error correction and fault-tolerance come at a large cost in resource overhead (qubits and gates).

I will show that the dissipative ground state algorithms of \cref{CPTP_map,stopped_CPT_map,secretary_CPT_map,expectation_CPT_map,decaying_CPT_map} achieve a limited form of error correction and fault-tolerance inherently, without any additional overhead.
Specifically, as long as the error rate $\delta$ is below a threshold related to $\Gamma-\Delta$, the output of the algorithm from \cref{stopped_CPT_map,secretary_CPT_map,expectation_CPT_map,decaying_CPT_map} is $O(\delta)$-close to the correct output, \emph{independent of the run-time of the algorithm}.
However, unlike full error correction and fault-tolerance, the error in the output cannot be made arbitrarily small independent of the error rate $\delta$.
I will call this weaker version of fault-tolerance, where the output error is independent of the run-time but not of the error rate, \keyword{fault-resilience}.
\Cref{fig:1D_Heisenberg_noise} gives a numerical illustration of these analytic results, showing how the DQE algorithm still produces a state close to the ground state even in the presence of a constant noise rate.

As in fault-tolerant quantum computation, to give a rigorous proof of fault-resilience, we must fix an error model.
As in the standard fault-tolerance threshold theorems, we will assume errors occur at each step of the algorithm, and that these errors are iid in time.
Also, as in standard fault-tolerance proofs, we will need to assume that we are able to perform some form of noisy state initialisation (or reset).
However, the state initialisation requirements are very weak: we will only require that we are able to (probabilistically) initialise qudits to a density matrix with full support.
I.e.\ it will suffice for our purposes to e.g.\ be able to initialise the qudits in a random, noisy computational basis state, as long as this process gives a non-zero probability of measuring any bit string on the resulting state.

We denote the induced 1-norm on a map $\cE$ by $\norm{\cE}_1 := \max_{\rho\neq0} \norm{\cE(\rho)}_1/\norm{\rho}_1$, and the completely-bounded version thereof---the diamond norm---by $\norm{\cE}_\diamond := \max_d\norm{\cE\ox\cI_d}$ where $\cI_d$ is the identity map on dimension $d$.
Note that $\norm{\cE}_1\leq\norm{\cE}_\diamond$ by definition, so upper-bounding the noise rate in terms of the 1-norm is a less stringent requirement than bounding it in terms of the diamond-norm.

\begin{definition}[Noise model]\label{def:noise_model}
  Let $\{\cE_0,\cE_1\}$ be the quantum instrument of \cref{stopped_CPT_map}, and $\cE = \cE_0+\cE_1$ the quantum process of \cref{CPTP_map}.
  In each iteration of the algorithm, the maps $\{\cE'_0,\cE'_1\}$ or $\cE'=\cE'_0+\cE'_1$ are applied instead.
  We say that the \emph{error rate} is $\norm{\cE' - \cE}_1$.

  We also assume we have access to a state initialisation procedure whose output has non-zero support on the full Hilbert space, so that $\cE'_1(\rho) \geq \mu\1$.
\end{definition}

Note that this noise model encompasses external noise, as well as faulty implementations of $\cE$.
If an additional, external noise map $\cN$ is applied at each step of the algorithm, we can take $\cE' = \cN\circ\cE$, and $\norm{\cE'-\cE}_\diamond = \norm{\cN\circ\cE - \cE}_\diamond \leq \norm{\cN - \cI}_\diamond\norm{\cE}_\diamond = \norm{\cN - \cI}_\diamond$.
Also note that this error model does not assume that the same error occurs at each step of the process, only that the probability distribution over the different errors that can occur is iid: if the noise is given by an ensemble $\{p_i\cE^{(i)}\}$ where $p_i$ is the probability that the CP(T) map ${\cE}'^{(i)}$ is applied, by linearity this is equivalent to applying the CP(T) map $\cE' = \sum_i p_i{\cE'}^{(i)}$.
Note also that \cref{def:noise_model} encompasses global noise, where the errors or noise can act globally across all the qudits.
We will see, below, that restricting to local noise improves the dimension dependence of our results.

In the following, we will make use of the fact that the action of a CP map $\cE(\rho) = \sum_i A_i\rho A_i^\dg$ as a linear map on a $D$\nbd-dimensional density matrix $\rho$ can be represented as multiplication of the vectorised density matrix $\kett{\rho} = \vec\rho$ by the $D^2$-dimensional transfer matrix $E = \sum_i A_i\ox\bar{A_i}$, where $A_i$ are the Kraus operators of $\cE$.

Perturbations of the transfer matrix $E$ are related to perturbations of the map $\cE$ up to an overall dimension factor, by standard norm relations:

\begin{lemma}\label{transfer_matrix_perturbation}
  If $\cE$ and $\cE'$ are CP maps on $\cB(\cH^D)$ satisfying $\norm{\cE-\cE'}_1\leq\delta$ or $\norm{\cE-\cE'}_\diamond\leq\delta$, then their transfer matrices $E$ and $E'$ satisfy $\norm{E-E'} \leq \sqrt{D}\delta$.
\end{lemma}

\begin{proof}
  This follows from the standard relation $\norm{X}_2 \leq \norm{X}_1 \leq \sqrt{D}\norm{X}_2$ on Schatten $p$-norms of $D\times D$ matrices:
  \begin{align}
    \norm{E-E'}
    &=\max_{X\neq 0} \frac{\bigNorm{(\cE-\cE')(X)}_2}{\norm{X}_2}
    \leq \sqrt{D} \max_{\rho\neq 0} \frac{\bigNorm{(\cE-\cE')(\rho)}_1}{\norm{\rho}_1} \\
    &= \sqrt{D}\norm{\cE-\cE'}_1 \leq \sqrt{D} \norm{\cE-\cE'}_\diamond.
  \end{align}
\end{proof}

If the noise is local, the dimension dependence reduces from a global dimension factor to a local one, by standard arguments:

\begin{lemma}\label{product_transfer_matrix_perturbation}
  If $\cE = \cE_m\circ\cE_{m-1}\circ\cdots\circ\cE_1$ and $\cE' = \cE'_m\circ\cE'_{m-1}\circ\cdots\circ\cE'_1$ are CP maps on qudits, where $\cE_i,\cE'_i$ are trace non-increasing CP maps that act non-trivially on at most $k$ qudits and satisfy $\norm{\cE-\cE'}_1\leq\delta$, then their transfer matrices $E$ and $E'$ satisfy $\norm{E-E'} \leq m d^{k/2}\delta$.
\end{lemma}

\begin{proof}
  Since $\cE_i,\cE'_i$ act non-trivially on at most $k$ qudits, we have $\cE_i = \mathfrak{e}_i\ox\1$ where $\mathfrak{e}_i$ is a CP map on $k$ qudits, and similarly for $\cE'_i$, with $\norm{\mathfrak{e}'_i-\mathfrak{e}_i}_1 \leq \delta$.
  Let $e_i,e'_i$ be the transfer matrices corresponding to $\mathfrak{e}_i,\mathfrak{e}'_i$.
  Using \cref{transfer_matrix_perturbation}, we have
  \begin{equation}\label{eq:local_norm}
    \norm{E'_i-E_i}
    = \norm{(e'_i-e_i)\ox\1}
    = \norm{e'_i-e_i}
    \leq \sqrt{d^k}\norm{\mathfrak{e}'_i-\mathfrak{e}_i}
    \leq d^{k/2}\delta.
  \end{equation}
  We have
  \begin{equation}
    \norm{E'_i E'_j - E_i E_j}
    \leq \norm{E'_i-E_i}\norm{E'_j} + \norm{E_i}\norm{E'_j-E_j}
    \leq \norm{E'_i-E_i} + \norm{E'_j-E_j},
  \end{equation}
  recalling that $\norm{E}\leq 1$ for the transfer matrix of a trace non-increasing CP map.
  The \namecref{product_transfer_matrix_perturbation} follows from applying this repeatedly and using \cref{eq:local_norm}.
\end{proof}

\begin{lemma}\label{sum_transfer_matrix_perturbation}
  If $\cE = \sum_{i=1}^m\cE_i$ and $\cE' = \sum_{i=1}^m\cE'_i$ are CP maps on qudits, where $\cE_i,\cE'_i$ are trace non-increasing CP maps that act non-trivially on at most $k$ qudits and satisfy $\norm{\cE-\cE'}_1\leq\delta$, then their transfer matrices $E$ and $E'$ satisfy $\norm{E-E'} \leq m d^{k/2}\delta$.
\end{lemma}

\begin{proof}
  \begin{equation}
    \norm{E-E'} \leq \sum_{i=1}^m\norm{E'_i-E_i} \leq m d^{k/2}\delta
  \end{equation}
  by \cref{eq:local_norm}.
\end{proof}

\Cref{product_transfer_matrix_perturbation,sum_transfer_matrix_perturbation} imply that, for the AGSPs of \cref{AGSP_product,AGSP_mixture}, if each local operation involved in implementing \cref{DQE} is faulty and/or subject to noise, but still acts locally on at most $k$ qudits (which may still be more qudits than the original operation acted upon), we do not pick up a global dimension dependence in our fault-resilience results.

\subsection{Stopped process resilience}
We will make use of the following Perron-Frobenius-type result from~\cite{Michael_channels}.

\begin{theorem}[{\cite[Theorem~6.5]{Michael_channels}}]\label{spectral_radius_eigenvalue}
  Let $\cE$ be a positive map with spectral\linebreak radius $r$.
  Then $r$ is an eigenvalue of $\cE$ and there is an eigenoperator $X\geq 0$ of $\cE$ such that $\cE(X) = r X$.
\end{theorem}

We will also need some standard perturbation results for eigenvalues and eigenspaces.

\begin{theorem}[{e.g.\ \cite[Theorem~7.2.2]{Golub+vanLoan}}]\label{eigenvalue_perturbation}
  If $\mu$ is an eigenvalue of $A + E \in \cM_n(\C)$ and $X^{-1}AX = \diag(\lambda_1,\dots,\lambda_n)$, then
  \begin{equation}
    \min_{\lambda\in\lambda(A)}\abs{\lambda-\mu} \leq \norm{X}\,\norm{X^{-1}}\,\norm{E}.
  \end{equation}
\end{theorem}

\begin{theorem}[{\cite[Theorem~4.11]{Stewart73}}]\label{eigenspace_perturbation}
  Let
  \begin{equation}
    UAU^\dg = \blockmatrix{A_{11}}{A_{12}}{0}{A_{22}}, \quad U = (V_1 \;|\; V_2)
  \end{equation}
  be a Schur decomposition of $A\in\cM_n{\C}$.
  Let $E\in\cM_n(\C)$ be partitioned as
  \begin{equation}
    UEU^\dg = \blockmatrix{E_{11}}{E_{12}}{E_{21}}{E_{22}}.
  \end{equation}
  Let
  \begin{equation}
    s = \sep(A_{11},A_{22}) - \norm{E_{11}} - \norm{E_{22}}.
  \end{equation}
  If
  \begin{equation}
    \frac{\norm{E_{21}}(\norm{A_{12}}+\norm{E_{12}})}{s^2} \leq \frac{1}{4},
  \end{equation}
  there is a matrix $Q$ satisfying
  \begin{equation}
    \norm{Q} \leq 2\frac{\norm{E_{21}}}{s}
  \end{equation}
  such that the columns of $V'_1 = (V_1+V_2Q)(\1+Q^\dg Q)^{-1/2}$ span an invariant subspace of $A+E$.
\end{theorem}

The following \namecref{Q_to_P} relates $\norm{Q}$ to the difference between the orthogonal projectors onto the unperturbed and perturbed invariant subspaces.

\begin{lemma}[{\cite[Theorems~2.2 and~2.7]{Stewart73}}]\label{Q_to_P}
  Let $U=(V_1|V_2)$ be unitary, and $V'_1 = (V_1+V_2Q)(\1+Q^\dg Q)^{-1/2}$.
  If $P$, $P'$ are the orthogonal projectors onto the subspaces spanned by the columns of $V_1$, $V'_1$, respectively, then
  \begin{equation}
    \norm{P-P'} \leq \norm{Q}.
  \end{equation}
\end{lemma}

From these, we can derive the following useful perturbation bounds.

\begin{corollary}\label{spectral_radius_perturbation}
  Let $\cE$ be a CP map whose transfer matrix $E$ is normal.
  Let $\cE'$ be a CP map satisfying $\norm{E-E'}\leq\delta$, with corresponding transfer matrix $E'$.
  Then the spectral radii $r(E)$ and $r(E')$ satisfy
  \begin{equation}
    \abs{r(E') - r(E)} \leq \delta.
  \end{equation}
\end{corollary}

\begin{proof}
  Follows immediately by combining \cref{transfer_matrix_perturbation,spectral_radius_eigenvalue,eigenvalue_perturbation}, noting that a normal matrix is diagonalisable by a unitary $U$ and $\norm{U}=1$.
\end{proof}

\begin{corollary}\label{spectral_projector_perturbation}
  Let $K$ be a $(\Delta,\Gamma,\epsilon)$-AGSP for $\Pi_0$ on a Hilbert space of dimension~$D$, $N:=\tr\Pi_0$ the ground state degeneracy, and $E_0$ the transfer matrix of the CP map $\cE_0(\rho) = K\rho K^\dg$.
  Let $\cE'_0$ be a CP map with transfer matrix $E'$ such that
  \begin{equation}
    \norm{E'_0-E_0} \leq \delta < \frac{\sqrt{\Gamma}(\sqrt{\Gamma}-\sqrt{\Delta})}{2N}.
  \end{equation}
  Let $P$ be the orthogonal projector onto the invariant subspace spanned by the generalised right-eigenvectors of $E$ associated with eigenvalues of magnitude $\leq\sqrt{\Gamma\Delta}$.

  Then there exists an orthogonal projector $P'$ onto the invariant subspace spanned by the generalised right-eigenvectors of $E'$ associated with eigenvalues of magnitude $\geq\sqrt{\Gamma\Delta}-\delta$, which satisfies
  \begin{equation}
    \norm{P'-P} \leq \frac{2\delta}{\sqrt{\Gamma}(\sqrt{\Gamma}-\sqrt{\Delta})-2\delta}.
  \end{equation}
\end{corollary}

\begin{proof}
  Since $E_0=K\ox\bar{K}$, by \cref{def:AGSP} the eigenvalues of $E_0$ can be partitioned into two sets: those with magnitudes $\geq\Gamma$ and those with magnitudes $\leq\sqrt{\Gamma\Delta}$.

  Let $U=(V_1\;|\;V_2)$ be unitary, with the columns of $V_1$ spanning the\linebreak $N^2$\nbd-dimensional invariant subspace associated with the eigenvalues with magnitude $\leq\sqrt{\Gamma\Delta}$ and $V_2$ spanning the $N^2$-dimensional subspace associated with those of magnitude $\geq\Gamma$, so that $P=V_1V_1^\dg$.
  Since $E_0$ is Hermitian, we have that $U^\dg E_0 U = \diag(A_{11},A_{22})$ is block diagonal, with (see \cite[\S 4.3]{Stewart73})
  \begin{equation}
    \sep(A_{11},A_{22}) \geq \sqrt{\Gamma}(\sqrt{\Gamma}-\sqrt{\Delta}).
  \end{equation}

  Noting that
  \begin{align}
    \norm{E_{12}}
    &= \max_{\ket{\psi}} \norm{E_{12}\ket{\psi}} \\
    &\leq \max_{\ket{\psi}} \Norm{\blockvector{E_{12}\ket{\psi}}{E_{22}\ket{\psi}}} \\
    &= \max_{\ket{\psi}} \Norm{\blockmatrix{E_{11}}{E_{12}}{E_{21}}{E_{22}} \blockvector{0}{\ket{\psi}}} \\
    &\leq \max_{\ket{\varphi}} \Norm{\blockmatrix{E_{11}}{E_{12}}{E_{21}}{E_{22}} \ket{\varphi}} \\
    &= \norm{E},
  \end{align}
  and similarly $\norm{E_{21}},\norm{E_{11}},\norm{E_{22}} \leq \norm{E}$, we see that the conditions of \cref{eigenspace_perturbation} are fulfilled by $A=E_0$ and $E=E'_0-E_0$, implying existence of an isometry
  \begin{equation}
    V'_1 = (V_1+V_2Q)(\1+Q^\dg Q)^{-1/2}
  \end{equation}
  whose columns span an invariant subspace of $E'_0$ corresponding to eigenvalues of magnitude $\geq\sqrt{\Gamma\Delta}-\delta$, with
  \begin{equation}
    \norm{Q} \leq \frac{2\delta}{\sqrt{\Gamma}(\sqrt{\Gamma}-\sqrt{\Delta})-2\delta}.
  \end{equation}
  The \namecref{spectral_projector_perturbation} follows by \cref{Q_to_P}.
\end{proof}

We will also need the following result from \cite{Michael_channels}.

\begin{lemma}[{\cite[Lemma~8.5]{Michael_channels}}]\label{triangular_bound}
  Let $X\in\cM_D(\C)$ have Schur decomposition $UXU^\dg=\Lambda + T$, where $\Lambda$ is diagonal with $\norm{\Lambda}\leq 1$, $T$ is strictly upper-triangular, and $U$ unitary.
  Then, for $n\in\N$,
  \begin{equation}
    \norm{X^n} \leq \norm{\Lambda^n} + C_{D,n}\norm{\Lambda}^{n-D+1}\max(\norm{T},\norm{T}^{D-1}),
  \end{equation}
  with $C_{D,n} = (D-1)n^{D-1}$.
  If in addition $2(D-1)\leq n$ then this holds with $C_{D,n} = (D-1)\binom{n}{D-1}$.
\end{lemma}

With this, we can prove the following contraction upper-bound.

\begin{proposition}\label{contraction_bound}
  Let $E\in\cM_D(\C)$ be a (not necessarily normal) matrix of spectral radius $\leq 1$, with Jordan decomposition $E = VJV^{-1}$.
  Let $P$ be the orthogonal projector onto the $d$-dimensional invariant subspace spanned by the generalised eigenvectors of $E$ associated with eigenvalues of magnitude $\leq\Delta$.

  Then
  \begin{equation}
    \norm{P E^n} \leq \kappa_J^2 \left(1 + (d-1) n^{d-1} \Delta^{-d+1}\right) \Delta^n.
  \end{equation}
  where $\kappa_J = \norm{V}\norm{V^{-1}}$ is the Jordan condition number.
\end{proposition}

\begin{proof}
  Write the Jordan decomposition $VEV^{-1} = J_1 \oplus J_2$ where $J_1$ contains all the Jordan blocks for eigenvalues with magnitude $\leq\Delta$.
  The first $d$ rows of $V$ span the invariant subspace onto which $P$ projects, so $P = V(S_1\oplus 0)V^{-1}$ for some invertible matrix $S_1$.
  Thus
  \begin{equation}
    P E^n = V (S_1\oplus 0) (J_1\oplus J_2)^n V^{-1} = V (S_1 J_1^n\oplus 0) V^{-1}.
  \end{equation}

  Now,
  \begin{equation}
    \norm{S_1} = \norm{S_1\oplus 0} = \norm{V^{-1} P V} \leq \kappa_J\norm{P} = \kappa_J.
  \end{equation}
  So
  \begin{equation}
    \norm{PE^n} \leq \kappa_J\norm{S_1}\norm{J_1^n} \leq \kappa_J^2 \norm{J_1^n}.
  \end{equation}

  Since $J_1 = \Lambda + T$ is a $d$-dimensional upper-triangular matrix with diagonal part $\Lambda\leq\Delta\1$ and $T$ the matrix with 1's along the off-diagonal, we can apply \cref{triangular_bound} to obtain
  \begin{align}
    \norm{PE^n}
    &\leq \kappa_J^2 \left(\norm{\Lambda^n} + C_{d,n}\norm{\Lambda^{n-d+1}}\max(\norm{T},\norm{T^{d-1}})\right)\\
    &\leq \kappa_J^2\left(1 + (d-1)n^{d-1} \Delta^{-d+1}\right)\Delta^n.
  \end{align}
\end{proof}

We will also need a lower-bound on the contraction of the trace under $\cE'_0$, which follows straightforwardly from \cref{spectral_radius_eigenvalue}.

\begin{lemma}\label{contraction_lower_bound}
  Let $\cE$ be a positive map with spectral radius $\geq r$, and $\sigma > 0$ positive definite with minimum eigenvalue $\lambda_0(\sigma) \geq s > 0$.
  Then \mbox{$\tr(\cE^n(\sigma)) \geq s r^n$}.
\end{lemma}

\begin{proof}
  Let $\rho = X/\tr(X)$ with the $X\geq 0$ from \cref{spectral_radius_eigenvalue}.
  Then $\sigma/s-\rho \geq 0$ and
  \begin{align}
    \tr(\cE^n(\sigma))
    &= s\tr(\cE^n(\sigma/s)) = s\tr\bigl(\cE^n(\rho)\bigr) + s\tr\bigl(\cE^n(\sigma/s-\rho)\bigr) \\
    &\geq s \tr(\cE^n(\rho)) \geq s r^n\tr(\rho) = s r^n.
  \end{align}
\end{proof}

We are now in a position to prove fault-resilience of the stopped process ground state preparation algorithms of \cref{stopped_CPT_map,secretary_CPT_map,expectation_CPT_map}.

\begin{theorem}\label{stopped_fault-resilience}
  Let $K$ be a $(\Delta,\Gamma,\epsilon)$-AGSP for $\Pi_0$ on a Hilbert space of dimension~$D$, and let $N:=\tr\Pi_0$ be the ground state degeneracy.
  Let $\{\cE_0,\cE_1\}$ be the quantum instrument of \cref{stopped_CPT_map} for $K$, and let $\{\cE'_0,\cE'_1\}$ be a quantum instrument such that the corresponding transfer matrices $E_0$, $E'_0$ satisfy
  \begin{equation}
    \norm{E'_0-E_0} \leq \delta < \frac{1}{2}\sqrt{\Gamma}(\sqrt{\Gamma}-\sqrt{\Delta}),
  \end{equation}
  and $\cE'_1$ satisfies $\cE'_1(\rho) \geq \mu\1$ for any state $\rho$.

  Consider the stopped process whereby we iterate $\{\cE'_0,\cE'_1\}$, starting from the maximally mixed state $\rho_0 = \tfrac{\1}{D}$, and stop at some point when we have just obtained a run of $n$ 0's.
  The state $\rho'_n$ at the stopping time satisfies
  \begin{equation}
    \begin{split}
    \tr\Pi_0\rho'_n
    \geq 1 &- \epsilon
            - \frac{2\delta}{\sqrt{\frac{\Gamma}{N}}(\sqrt{\Gamma}-\sqrt{\Delta})-2\delta} \\
           &- \left(\frac{\sqrt{\Delta}+\delta/\sqrt{\Gamma}}
                         {\sqrt{\Gamma}-\delta/\sqrt{\Gamma}}\right)^n
              \frac{\sqrt{D}\kappa_J}{\mu} \left(1 + (D^2-N^2-1) n^{D^2-N^2-1} {\Delta'}^{-D^2+N^2+1}\right),
    \end{split}
  \end{equation}
  where $\kappa_J$ is the Jordan condition number of the transfer matrix $E'_0$.

  In particular,
  \begin{equation}
    \lim_{n\to\infty}\tr(\Pi_0\rho'_n)
    \geq 1 - \epsilon
            - \frac{2\delta}{\sqrt{\Gamma}(\sqrt{\Gamma}-\sqrt{\Delta})-2\delta}.
  \end{equation}
\end{theorem}

\begin{proof}
  Let $\Pi$ be the projector for $K$ from \cref{def:AGSP}, and let $E_0$ and $E'_0$ be the transfer matrices for $\cE_0$ and $\cE'_0$, respectively.
  Since $E_0 = K\ox\bar{K}$, by \cref{def:AGSP} $\Pi\ox\Pi$ is the orthogonal projector onto the $N^2$-dimensional eigenspace of $E_0$ with eigenvalues $\geq\Gamma$, and $P = \1-\Pi\ox\Pi$ is the orthogonal projector onto the $D^2-N^2$-dimensional eigenspace with eigenvalues $\leq \sqrt{\Gamma\Delta}$.

  Noting that $\norm{E'_0-E_0} \leq \delta$ by assumption, \cref{spectral_projector_perturbation} implies that there is an orthogonal projector $P'$ with
  \begin{equation}
    \norm{P'-P} \leq \frac{2\delta}{\sqrt{\Gamma}(\sqrt{\Gamma}-\sqrt{\Delta})-2\delta}
  \end{equation}
  such that $P'$ projects onto the $(D^2-N^2)$-dimensional subspace corresponding to eigenvalues of magnitude $\leq \Delta' := \sqrt{\Gamma\Delta}+\delta$.
  Note that, since $\delta < \frac{1}{2}\sqrt{\Gamma}(\sqrt{\Gamma}-\sqrt{\Delta})$, this implies that $\rank P' = \rank P = N$.

  Let $\rho'_0$ be the state after the last application of $\cE'_1$.
  Using Cauchy-Schwartz and \cref{contraction_bound}, and noting that $\norm{\kett{\rho}} = \sqrt{\tr(\rho^2)} \leq 1$ for any state $\rho$, we have
  \begin{align}
    \Abs{\braakett{\1| P' {E'_0}^n |\rho'_0}}
    &\leq \norm{\kett{\1}} \Norm{P'{E'_0}^n} \norm{\kett{\rho'_0}}
    \leq \sqrt{D} \Norm{P'{E'_0}^n} \\
    &\leq \sqrt{D}\kappa_J \left(1 + (D^2-N^2-1) n^{D^2-N^2-1} {\Delta'}^{-D^2+N^2+1}\right) {\Delta'}^n.
      \label{eq:Delta_bound}
  \end{align}

  Meanwhile, $\rho'_0 = \cE'_1(\rho)$ has minimum eigenvalue $\lambda_0(E'_1(\rho)) \geq \mu$ by assumption.
  So by \cref{contraction_lower_bound,spectral_radius_perturbation} we also have
  \begin{equation}
    \Abs{\braakett{\1| {E'_0}^n |\rho'_0}} = \tr\left({\cE'_0}^n(\rho'_0)\right) \geq \mu\Gamma'^n, \label{eq:Gamma_bound}
  \end{equation}
  where $\Gamma' := \Gamma - \delta$.

  Proceeding similarly to \cref{stopped_CPT_map}, using \cref{eq:Delta_bound,eq:Gamma_bound} we obtain
  \begin{align}
    \abs{\braakett{\1| P' |\rho'_n}}
    &=\Abs{\frac{\braakett{\1| P' {E'_0}^n |\rho'_0}}{\braakett{\1| {E'_0}^n |\rho'_0}}}\\
    &\leq \frac{\sqrt{D}}{\mu} \left(\frac{\Delta'}{\Gamma'}\right)^n
      \kappa_J \Bigl(1 + (D^2-N^2-1) n^{D^2-N^2-1} {\Delta'}^{-D^2+N^2+1}\Bigr) \\
    &=\frac{\sqrt{D}}{\mu}
      \left(\frac{\sqrt{\Delta}+\delta/\sqrt{\Gamma}}{\sqrt{\Gamma}-\delta/\sqrt{\Gamma}}\right)^n
      \kappa_J \left(1 + (D^2-N^2-1) n^{D^2-N^2-1} {\Delta'}^{-D^2+N^2+1}\right).
  \end{align}

  \noindent
  But
  \begin{align}
    \tr(\Pi_0\rho'_n) &\geq \tr(\Pi\rho'_n) - \norm{\Pi-\Pi_0}\\
    &= \braakett{\1| \Pi\ox\Pi |\rho'_n} - \norm{\Pi-\Pi_0}\\
    &= 1 - \abs{\braakett{\1|P|\rho'_n}} - \norm{\Pi-\Pi_0}\\
    &\geq 1 - \abs{\braakett{\1| P' |\rho'_n}} - \norm{P-P'} - \norm{\Pi-\Pi_0},
  \end{align}
  and the \namecref{stopped_fault-resilience} follows.
\end{proof}

\begin{figure}[!htbp]
  \centering
  \includegraphics[width=\textwidth]{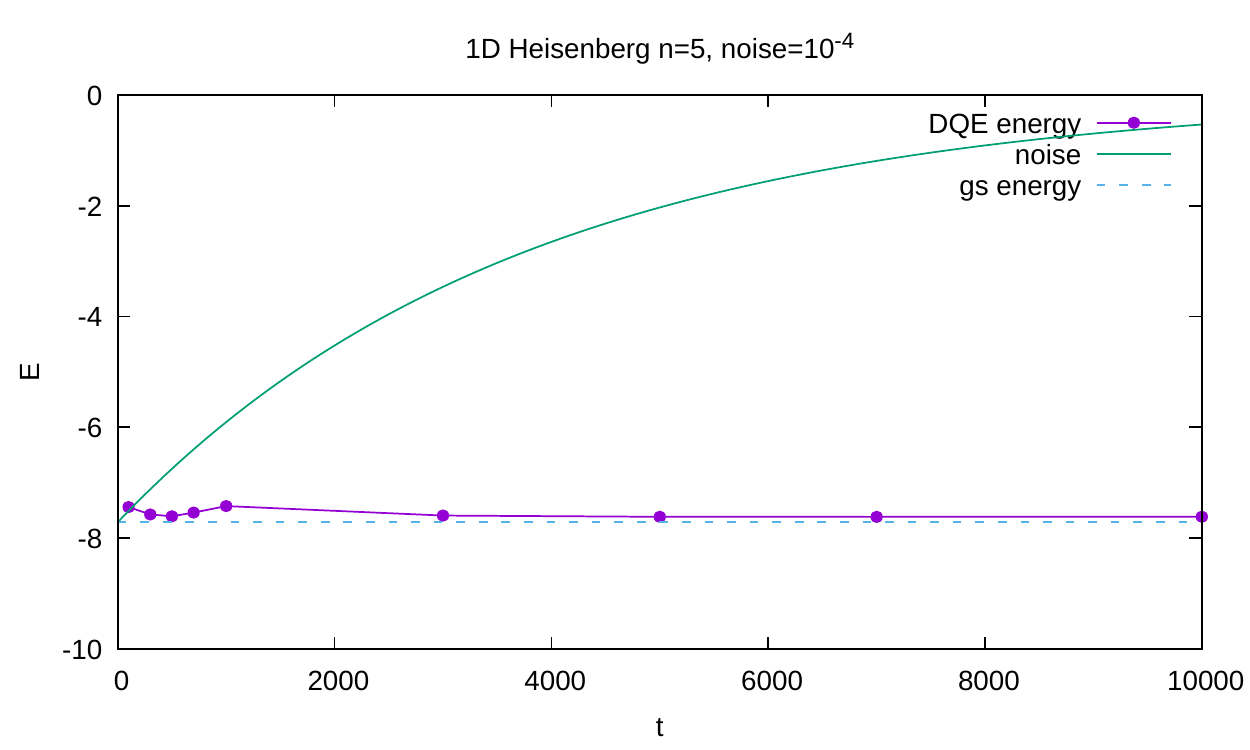}
  \stepcounter{footnote}{\makeatletter\protected@xdef\@thefnmark{\thefootnote}\makeatother}
  \caption{%
    Energy of the DQE output state as a function of maximum run-time $t$ in \cref{secretary_CPT_map} with the standard error model (green circles), for the 1D Heisenberg chain of length~5 using the secretary stopping policy of \cref{secretary_CPT_map}.\makeatletter\@footnotemark\makeatother
    In the standard error model, a 1- or 2-qubit depolarising channel is applied after each 1-qubit or 2-qubit gate, respectively, here with depolarising parameter $10^{-4}$.
    For comparison, the energy of the state that would be obtained by starting from the ground state and applying the same $10^{-4}$ 1-qubit depolarising channel across all qubits $t$ times is shown (blue curve).
    The true ground state energy is also indicated (purple horizontal line).
    The energy of the DQE output state is clearly seen to be independent of the run-time and remains close in energy to the true ground state (green), illustrating the noise- and fault-resilience discussed in \cref{sec:fault-resilience}.
    Whereas without the DQE process, the state decays over the same run-time to the maximally mixed state, which is far from the ground state.}
  \label{fig:1D_Heisenberg_noise}
\end{figure}
\footnotetext{Numerical simulations performed with the help of the QuEST package~\cite{QuEST}.}

\subsection{Fixed-point process resilience}
For completeness, I also prove fault-resilience for the fixed-point process of \cref{CPTP_map}.
However, note that the simple bound derived here is not as strong as that of \cref{stopped_fault-resilience} for the stopped process, and is very likely not tight.

\begin{theorem}\label{fixed-point_fault-resilience}
  Let $K$ be a $(\Delta,\Gamma,\epsilon)$-AGSP for $\Pi_0$ on a Hilbert space of dimension~$D$, and let $N:=\tr\Pi_0$ be the ground state degeneracy.
  Let $\cE$ be the CPT map of \cref{CPTP_map} for $K$, and let $\cE'$ be a CPT map such that $\norm{\cE'-\cE}_1 \leq \delta$.

  The fixed points $\rho'_\infty$ of $\cE'$ satisfy:
  \begin{equation}
    \tr(\Pi_0\rho'_\infty)
    \geq 1 - \left(\frac{1-\Gamma}{1-\Delta}+\delta\right) \left(\frac{D}{N}-1\right) - \epsilon - \delta.
  \end{equation}
\end{theorem}

\begin{proof}
  Let $\Pi$ be as in \cref{def:AGSP}.
  For any state $\rho$, we have
  \begin{equation}
    \tr(\Pi\cE'(\rho))
      = \tr(\Pi\cE(\rho)) - \tr\bigl(\Pi\,(\cE-\cE')(\rho)\bigr)
      \geq \tr(\Pi\cE(\rho)) - \delta.
  \end{equation}
  From the proof of \cref{CPTP_map}, we also have
  \begin{equation}
    \tr(\Pi\cE(\rho))
    \geq \left(\Gamma-\frac{N}{D}(\Gamma-\Delta)\right)\tr(\Pi\rho) + \frac{N}{D}(1-\Delta).
  \end{equation}
  Applying these to the perturbed fixed point $\rho'_\infty$, we obtain
  \begin{align}
    \tr(\Pi\rho'_\infty)
    &= \tr(\Pi\cE'(\rho'_\infty)) \\
    &\geq \tr(\Pi\cE(\rho'_\infty)) - \delta \\
    &\geq \left(\Gamma-\frac{N}{D}(\Gamma-\Delta)\right)\tr(\Pi\rho'_\infty) + \frac{N}{D}(1-\Delta) - \delta.
  \end{align}
  Thus
  \begin{equation}
    \tr(\Pi\rho'_\infty) \geq \frac{1-\Delta-\frac{D}{N}\delta}{(\Gamma-\Delta)+\frac{D}{N}(1-\Gamma)}
    = 1 - \left(\frac{1-\Gamma}{1-\Delta}+\delta\right) \left(\frac{D}{N}-1\right) - \epsilon - \delta,
  \end{equation}
  and the \namecref{fixed-point_fault-resilience} follows from $\norm{\Pi-\Pi_0}\leq\epsilon$ as in \cref{CPTP_map}.
\end{proof}

\clearpage

\section{Circuit implementations}\label{sec:implementation}
I work out explicit implementation details for the DQE \cref{DQE} on Pauli Hamiltonians, i.e.\ Hamiltonians whose local terms are tensor products of Pauli operators.
Since almost all quantum computers act on qubits, it makes sense to focus on this case when working out implementation details.
But it is also without loss of generality, since qudits of any dimension can be embedded into $\log d$ qubits, and all qubit Hamiltonians can be decomposed in terms of Paulis.

Consider, therefore, a general Pauli Hamiltonian
\begin{equation}\label{eq:Pauli_Hamiltonian}
  H = \sum_v \alpha_v h_v, \quad h_v = \bigotimes_{i\in v} \sigma_v^{(i)},
\end{equation}
where $v$ are subsets of qubits, $\alpha_v\in\R$ and $\sigma_v \in \{X,Y,Z\}$.
Note that $h_v^2 = \1$ and $\norm{h_v} = 1$.

For Pauli Hamiltonians, the AGSP of \cref{AGSP_sum} simplifies:
\begin{equation}
  k_v = \frac{\1-s_v h_v}{2} = \Pi_v,
  \qquad
  K = \sum_v\kappa_v\Pi_v
\end{equation}
where $s_v = \mathrm{sign}(\alpha_v)$, $\kappa_v = \frac{\abs{\alpha_v}}{\alpha}$ with $\alpha := \sum_v\abs{\alpha_v}$, and $\Pi_v$ is a projector since we have $h_v^2=\1$.
Thus, to approximate $K$ locally, we need to implement generalised measurements of the form
\begin{equation}
  E_0 = (1-\epsilon)\1 + \epsilon\kappa_v\Pi_v = (1-\epsilon(1-\kappa_v))\Pi_v + (1-\epsilon)(\1-\Pi_v).
\end{equation}
To form a complete set of measurement operators $\{E_0,E_1\}$, we require
\begin{equation}
  E_1^\dg E_1 = \1 - E_0^\dg E_0
  = \epsilon(2-\epsilon)(\1-\Pi_v) + \epsilon(1-\kappa_v)(2-\epsilon(1-\kappa_v))\Pi_v,
\end{equation}
hence
\begin{equation}
  E_1 = \sqrt{\epsilon(1-\kappa_v)(2-\epsilon(1-\kappa_v))}\Pi_v + \sqrt{\epsilon(2-\epsilon)}(\1-\Pi_v).
\end{equation}

The Steinspring dilation of the generalised measurement $\{E_0,E_1\}$ is the isometry
\begin{align}
  V = \ket{0}\ox E_0 + \ket{1}\ox E_1
  =\begin{pNiceMatrix}[last-col]
      (1-\epsilon(1-\kappa_v))\1                            & 0              & \ket{0} \Pi      \\
      0                                                     & (1-\epsilon)\1 & \ket{0} \Pi^\perp \\
      \sqrt{\epsilon(1-\kappa_v)(2-\epsilon(1-\kappa_v))}\1 & 0              & \ket{1} \Pi      \\
      0                                      & \sqrt{\epsilon(2-\epsilon)}\1 & \ket{1} \Pi^\perp
     \end{pNiceMatrix},
\end{align}
followed by a projective measurement of the ancilla qubit in the computational basis.
This can be completed to the unitary
\begin{align}
  U &=
  \begin{pNiceMatrix}[last-col]
    (1-\epsilon(1-\kappa_v))\1                            & 0              & -\sqrt{\epsilon(1-\kappa_v)(2-\epsilon(1-\kappa_v))}\1 & 0             & \ket{0} \Pi       \\
    0                                                     & (1-\epsilon)\1 & 0                                     & -\sqrt{\epsilon(2-\epsilon)}\1 & \ket{0} \Pi^\perp \\
    \sqrt{\epsilon(1-\kappa_v)(2-\epsilon(1-\kappa_v))}\1 & 0              & (1-\epsilon(1-\kappa_v))\1                            & 0              & \ket{1} \Pi       \\
    0                                      & \sqrt{\epsilon(2-\epsilon)}\1 & 0                                                     & (1-\epsilon)\1 & \ket{1} \Pi^\perp
  \end{pNiceMatrix}\\
  &= R_\phi\ox\Pi + R_\theta\ox(\1-\Pi) \\
  &= \bigl(R_\phi\ox\1\bigr) \bigl(\1\ox\Pi + R_{\theta-\phi}\ox(\1-\Pi)\bigr),
\end{align}
where the $R$ operations are the rotations
\begin{equation}
  R_x = \begin{pmatrix}
    \cos x & -\sin x \\
    \sin x & \phantom{-}\cos x
  \end{pmatrix},\qquad
  \theta = \cos^{-1}(1-\epsilon), \qquad
  \phi = \cos^{-1}(1-\epsilon(1-\kappa_v)).
\end{equation}
Note that in the case $\kappa_v = 1$ (equivalent to all Pauli terms in the Hamiltonian~\cref{eq:Pauli_Hamiltonian} having identical coefficients), we have $\phi=0$ and the above simplifies to $R_\phi = \1$, $R_{\theta-\phi} = R_\theta$.

Now, $U$ is made up of a controlled-$R$ operation on the ancilla qubit, controlled by the outcome of a $\{\Pi,\1-\Pi\}$ projective measurement on the main qubit register.
Since $\Pi$ is the projector onto the $+1$ eigenspace of a Pauli string $\bigotimes_{i=1}^k\sigma_i$, this can be implemented by \cite[\S 4.3]{Nielsen+Chuang} a Clifford circuit plus three single-qubit rotations:
\begin{equation}\label{measurement_circuit}
  \Qcircuit @C=1em @R=.7em {
    & \gate{\makebox[\widthof{$R_{\theta-\phi}^{-\frac{1}{2}}$}]{\raisebox{.2em}{$R_\phi$}\rule{0pt}{\heightof{$R_{\theta-\phi}^{-\frac{1}{2}}$}}}} & \targ & \targ & \qw & \cdots && \targ & \gate{R_{\theta-\phi}^{-\frac{1}{2}}} &
      \targ & \qw & \cdots && \targ & \targ & \gate{R_{\theta-\phi}^{\frac{1}{2}}} & \meter \\
    & \Gate{H_{\sigma_1}^\dg} & \ctrl{-1} & \qw & \qw & \cdots && \qw & \qw &
      \qw & \qw & \cdots && \qw & \ctrl{-1} & \Gate{H_{\sigma_1}} & \qw \\
    & \Gate{H_{\sigma_2}^\dg} & \qw & \ctrl{-2} & \qw & \cdots && \qw & \qw &
      \qw &  \qw & \cdots && \ctrl{-2} & \qw & \Gate{H_{\sigma_2}} & \qw \\
    & \vdots & & & & \ddots & & & & & & \iddots & & & & \vdots \\\\
    & \Gate{H_{\sigma_k}^\dg} & \qw & \qw & \qw & \cdots && \ctrl{-5} & \qw &
      \ctrl{-5} & \qw & \cdots && \qw & \qw & \Gate{H_{\sigma_k}} & \qw \\
  }
\end{equation}
where $H_{\sigma}$ are the single-qubit unitaries that rotate into the $\sigma$ eigenbasis, i.e.\ $H_\1 = \1$, $H_x = H$, $H_y = SH$ and $H_z = \1$.

The above circuit implements the required weak measurement of a single local Hamiltonian term $\alpha_v h_v$ in \cref{eq:Pauli_Hamiltonian}.
To implement one complete step of \cref{DQE}, a circuit of this form must be replicated for each local term in \cref{eq:Pauli_Hamiltonian}.
Note that for Hamiltonian terms acting on disjoint sets of qubits, these circuits can be implemented in parallel.
So the for a local Hamiltonian on a lattice, the total circuit depth for a single step of \cref{DQE} is constant, independent of the number of qubits.
(More generally, the optimal circuit depth is proportional to the chromatic number of the interaction graph of the Hamiltonian.)

This whole quantum circuit must then be applied repeatedly, until the stopping condition of \cref{DQE} is satisfied.
After each circuit application, depending on the measurement outcomes in that step, the resampling operation $\cR$ of \cref{DQE} may need to be applied before applying the next iteration of the circuit.
If simple replacement (either global or local) with the maximally mixed state is used as the resampling rule, this can be implemented straightforwardly, e.g.\ by discarding the appropriate qubits and bringing in fresh ones initialised in a random computational basis state.
For global resampling, the entire quantum state is discarded and replaced; for local resampling, only those qubits involved in the specific local term(s) that gave outcome~1 in the weak measurement are discarded and replaced.
In both cases, local and global, the resampling can also be applied immediately after a 1~outcome is obtained for a local term, rather than waiting until all local terms have been measured in the current time step.
If the resampling is applied after each local measurement, then only a single ancilla qubit is required; this ancilla can be reused in every measurement.
The overall structure of \cref{DQE} is illustrated in \cref{fig:DQE_product_schematic,fig:DQE_mixture_schematic}.

\begin{figure}
  \center
  \makebox[\textwidth][c]{
  \Qcircuit @C=1em @R=.7em {
    &&&                &&        & \cwx[4]                                &            & \cwx[4]                                &           &       &&     & \cwx[4]                                &           &        &&     & \cwx[4]                                &           & \\
    \\
    &&&                &&        &                                        & \cctrl{2}  &                                        & \cctrl{2} &       &&     &                                        & \cctrl{2} &        &&     &                                        & \cctrl{2} \\
    \\
    &&&\ustick{\ket{0}}&&        & \multigate{5}{\{\cE_{1,0},\cE_{1,1}\}} & \gate{X}   & \multigate{5}{\{\cE_{2,0},\cE_{2,1}\}} & \gate{X} & \cdots && \qw & \multigate{5}{\{\cE_{i,0},\cE_{i,1}\}} & \gate{X}  & \cdots && \qw & \multigate{5}{\{\cE_{m,0},\cE_{m,1}\}} & \gate{X}  & \qw & \ustick{\ket{0}} \qw & \qw & \qw \ar@{->} `r/.5em[dddddddddddd]`[ddddddddddd]`[dddddddddddlllllllllllllllllllllll]`[lllllllllllllllllll] [llllllllllllllllll]\\
    &&&                &&        & \nghost{\{\cE_{1,0},\cE_{1,1}\}}       &            & \nghost{\{\cE_{2,0},\cE_{2,1}\}}       &          &        &&     & \nghost{\{\cE_{i,0},\cE_{i,1}\}}       &           &        &&     & \nghost{\{\cE_{m,0},\cE_{m,1}\}}       \\
    &&&                &&        & \ghost{\{\cE_{1,0},\cE_{1,1}\}}        & \qw        & \ghost{\{\cE_{2,0},\cE_{2,1}\}}        & \qw      & \cdots && \qw & \ghost{\{\cE_{i,0},\cE_{i,1}\}}        & \qw       & \cdots && \qw & \ghost{\{\cE_{m,0},\cE_{m,1}\}}        & \qw       & \qw & \qw & \qw \ar@{->} `r/.5em[dddddddd]`[dddddddd] `[dddddddlllllllllllllllllllll]`[llllllllllllllllll] [lllllllllllllllll]\\
    &&&                &&        & \ghost{\{\cE_{1,0},\cE_{1,1}\}}        & \qw        & \ghost{\{\cE_{2,0},\cE_{2,1}\}}        & \qw      & \cdots && \qw & \ghost{\{\cE_{i,0},\cE_{i,1}\}}        & \qw       & \cdots && \qw & \ghost{\{\cE_{m,0},\cE_{m,1}\}}        & \qw       & \qw & \qw \ar@{->} `r/.5em[dddddd]`[dddddd]`[ddddddlllllllllllllllllll]`[lllllllllllllllll] [llllllllllllllll]\\
    &&&                && \vdots & \nghost{\{\cE_{1,0},\cE_{1,1}\}}       & \vdots     & \nghost{\{\cE_{2,0},\cE_{2,1}\}}       &          & \vdots &&     & \nghost{\{\cE_{i,0},\cE_{i,1}\}}       &           & \vdots &&     & \nghost{\{\cE_{m,0},\cE_{m,1}\}}       & \vdots    \\
    &&&                &&        & \ghost{\{\cE_{1,0},\cE_{1,1}\}}        & \qw        & \ghost{\{\cE_{2,0},\cE_{2,1}\}}        & \qw      & \cdots && \qw & \ghost{\{\cE_{i,0},\cE_{i,1}\}}        & \qw       & \cdots && \qw & \ghost{\{\cE_{m,0},\cE_{m,1}\}}        & \qw       \ar@{->} `r/.5em[ddd]`[dd]`[ddlllllllllllllll]`[lllllllllllllll] [llllllllllllll]\\
    &&&&&&&&&&&&&&&&&&&&&&&&&\\
    &&&&&&&&&&&&&&&&&&&&&&&&&\\
    &&&&&&&&&&&&&&&&&&&&&&&&&\\
    &&&&&&&&&&&&&&&&&&&&&&&&&\\
    &&&&&&&&&&&&&&&&&&&&&&&&&\\
    &&&&&&&&&&&&&&&&&&&&&&&&&\\
    &&&&&&&&&&&&&&&&&&&&&&&&&
  }}
  \caption{Schematic illustration of the DQE \cref{DQE} when using the AGSP from \cref{AGSP_product}.
    The qubits are passed through the circuit of \cref{measurement_circuit}, repeated for each local term in the Hamiltonian.
    The top-most qubit is the ancilla qubit that gets measured in \cref{measurement_circuit} to produce the classical measurement outcomes, indicated by the double-wires.
    This ancilla is then reset to $\ket{0}$ by the controlled-$X$ gate and reused in the subsequent measurement.
    (Alternatively, a fresh $\ket{0}$ ancilla could also be used for each measurement.)
    The qubits are repeatedly cycled through this circuit until the measurement outcomes satisfy the stopping condition of \cref{DQE}.}
  \label{fig:DQE_product_schematic}
\end{figure}

\begin{figure}
  \center
  \makebox[\textwidth][c]{
  \Qcircuit @C=1em @R=.7em {
    &&&&&        &                                                                 &     &     &     &     &     &     &        & \cwx[4]                                &            \\
    \\
    &&&&&        &                                                                 &     &     &     &     &     &     &        &                                        & \cctrl{2}  \\
    \\
    &&&&&        &                                                                 &     &     &     &     &     &     & \qw    & \multigate{5}{\{\cE_{2,0},\cE_{2,1}\}} & \gate{X}   & \qw \ar@{-}[ddddddddddddddrrrrr] \\
    &&&&&        &                                                                 &     &     &     &     &     &     &        & \nghost{\{\cE_{2,0},\cE_{2,1}\}}       &            \\
    &&&&&        &                                                                 &     &     &     &     &     &     & \qw    & \ghost{\{\cE_{2,0},\cE_{2,1}\}}        & \qw        & \qw \ar@{-}[ddddddddddddddrrrrr] \\
    &&&&&        &                                                                 &     &     &     &     &     &     & \qw    & \ghost{\{\cE_{2,0},\cE_{2,1}\}}        & \qw        & \qw \ar@{-}[ddddddddddddddrrrrr] \\
    &&&&&        &                                                                 &     &     &     &     &     &     & \vdots & \nghost{\{\cE_{2,0},\cE_{2,1}\}}       & \vdots     \\
    &&&&&        &                                                                 &     &     &     &     &     &     & \qw    & \ghost{\{\cE_{2,0},\cE_{2,1}\}}        & \qw        & \qw \ar@{-}[ddddddddddddddrrrrr] \\
    \\
    &&&&&        &                                                                 &     &     &     &     &     &     &        & \vdots \\
    \\
    \\
    &&&&&        &                                                                 &     &     &     &     &     &     &        & \cwx[4]                                &            \\
    \\
    &&&&&        &                                                                 &     &     &     &     &     &     &        &                                        & \cctrl{2}  \\
    \\
    &&& \ustick{\ket{0}} &&        & \qw \ar@{-}[uuuuuuuuuuuuuurrrrrr] \ar@{-}[ddddddddddddddrrrrrr] & \qw & \qw & \qw & \qw & \qw & \qw & \qw    & \multigate{5}{\{\cE_{i,0},\cE_{i,1}\}} & \gate{X}   & \qw & \qw & \qw & \qw & \qw & \qw & \qw    & \qw & \ustick{\ket{0}} \qw & \qw & \qw \ar@{->} `r/.5em[dddddddddddddddddddddddddd]`[dddddddddddddddddddddddddd]`[ddddddddddddddddddddddddddllllllllllllllllllllllllll]`[llllllllllllllllllllllllll] [lllllllllllllllllllll]\\
    &&&                  &&        &                                                                 &     &     &     &     &     &     &        & \nghost{\{\cE_{i,0},\cE_{i,1}\}}       &            \\
    &&&                  &&        & \qw \ar@{-}[uuuuuuuuuuuuuurrrrrr] \ar@{-}[ddddddddddddddrrrrrr] & \qw & \qw & \qw & \qw & \qw & \qw & \qw    & \ghost{\{\cE_{i,0},\cE_{i,1}\}}        & \qw        & \qw & \qw & \qw & \qw & \qw & \qw & \qw    & \qw & \qw & \qw \ar@{->} `r/.5em[ddddddddddddddddddddddd]`[ddddddddddddddddddddddd]`[dddddddddddddddddddddddllllllllllllllllllllllll]`[llllllllllllllllllllllll] [llllllllllllllllllll]\\
    &&&                  &&        & \qw \ar@{-}[uuuuuuuuuuuuuurrrrrr] \ar@{-}[ddddddddddddddrrrrrr] & \qw & \qw & \qw & \qw & \qw & \qw & \qw    & \ghost{\{\cE_{i,0},\cE_{i,1}\}}        & \qw        & \qw & \qw & \qw & \qw & \qw & \qw & \qw    & \qw & \qw \ar@{->} `r/.5em[ddddddddddddddddddddd]`[ddddddddddddddddddddd]`[dddddddddddddddddddddllllllllllllllllllllll]`[llllllllllllllllllllll] [lllllllllllllllllll]\\
    &&&                  && \vdots &                                                                 &     &     &     &     &     &     & \vdots & \nghost{\{\cE_{i,0},\cE_{i,1}\}}       & \vdots     &     &     &     &     &     &     & \vdots &     &     &     &     \\
    &&&                  &&        & \qw \ar@{-}[uuuuuuuuuuuuuurrrrrr] \ar@{-}[ddddddddddddddrrrrrr] & \qw & \qw & \qw & \qw & \qw & \qw & \qw    & \ghost{\{\cE_{i,0},\cE_{i,1}\}}        & \qw        & \qw & \qw & \qw & \qw & \qw & \qw & \qw \ar@{->} `r/.5em[ddddddddddddddddd]`[ddddddddddddddddd]`[dddddddddddddddddllllllllllllllllll]`[llllllllllllllllll] [lllllllllllllllll]\\
    \\
    &&&                  &&        &                                                                 &     &     &     &     &     &     &        & \vdots \\
    \\
    \\
    &&&                  &&        &                                                                 &     &     &     &     &     &     &        & \cwx[4]                                &            & \\
    \\
    &&&                  &&        &                                                                 &     &     &     &     &     &     &        &                                        & \cctrl{2}  & \\
    \\
    &&&                  &&        &                                                                 &     &     &     &     &     &     & \qw    & \multigate{5}{\{\cE_{m,0},\cE_{m,1}\}} & \gate{X}   & \qw \ar@{-}[uuuuuuuuuuuuuurrrrr] \\
    &&&                  &&        &                                                                 &     &     &     &     &     &     &        & \nghost{\{\cE_{m,0},\cE_{m,1}\}}       &            &     \\
    &&&                  &&        &                                                                 &     &     &     &     &     &     & \qw    & \ghost{\{\cE_{m,0},\cE_{m,1}\}}        & \qw        & \qw \ar@{-}[uuuuuuuuuuuuuurrrrr] \\
    &&&                  &&        &                                                                 &     &     &     &     &     &     & \qw    & \ghost{\{\cE_{m,0},\cE_{m,1}\}}        & \qw        & \qw \ar@{-}[uuuuuuuuuuuuuurrrrr] \\
    &&&                  &&        &                                                                 &     &     &     &     &     &     & \vdots & \nghost{\{\cE_{m,0},\cE_{m,1}\}}       & \vdots     &     \\
    &&&                  &&        &                                                                 &     &     &     &     &     &     & \qw    & \ghost{\{\cE_{m,0},\cE_{m,1}\}}        & \qw        & \qw \ar@{-}[uuuuuuuuuuuuuurrrrr] \\
    &&&&&&&&&&&&&&&&&&&&&&&&&&\\
    &&&&&&&&&&&&&&&&&&&&&&&&&&\\
    &&&&&&&&&&&&&&&&&&&&&&&&&&\\
    &&&&&&&&&&&&&&&&&&&&&&&&&&\\
    &&&&&&&&&&&&&&&&&&&&&&&&&&\\
    &&&&&&&&&&&&&&&&&&&&&&&&&&\\
    &&&&&&&&&&&&&&&&&&&&&&&&&&
  }}
  \caption{Schematic illustration of the DQE \cref{DQE} when using the AGSP from \cref{AGSP_mixture}.
    The qubits are passed through the circuit of \cref{measurement_circuit}, selected at random from each local term in the Hamiltonian.
    The top-most qubit in any block is the ancilla qubit that gets measured in circuit \cref{measurement_circuit}, producing the classical measurement outcomes indicated by the double-wires.
    This ancilla is then reset to $\ket{0}$ by the controlled-$X$ gate and reused in the subsequent iteration.
    (Alternatively, a fresh $\ket{0}$ ancilla could also be used for each measurement.)
    The qubits are repeatedly cycled through this circuit until the measurement outcomes satisfy the stopping condition of \cref{DQE}.}
  \label{fig:DQE_mixture_schematic}
\end{figure}

Although this algorithm can be implemented on any quantum hardware able to carry out intermediate measurements during the circuit, the structure of the DQE \cref{DQE} lends itself particularly well to hardware architectures with ``flying'' qubits which can readily be cycled through the same circuit repeatedly.
In particular, many of the features of integrated photonic quantum hardware seem especially well-suited to DQE~\cite{PsiQ_architecture}:
\begin{itemize}
\item Repeatedly cycling photonic qubits through the same quantum circuit module is a natural part of photonic quantum computing hardware.
\item Clifford circuits are inherently less costly to implement in the measurement- and fusion-based quantum computation schemes being implemented on photonic hardware; most of the DQE circuit consists of Clifford gates.
\item The only non-Clifford gate required---the single-qubit $R$ rotation---is the natural single-qubit gate produced by a non-equal beam-splitter~\cite{PsiQ_gates}.
\item Moreover, this single-qubit rotation is always applied to the same ancilla qubit (if the same ancilla is reused for each measurement, as described above).
  So the non-Clifford gate in the circuit only needs to be implemented for this one qubit line.
\end{itemize}

\clearpage

\section{Extensions}\label{sec:conclusions}

There are a number of natural extensions of the DQE \cref{DQE}, beyond the general algorithm discussed so far.
We describe some of these briefly here.

\subsection{Stabiliser-encoded Hamiltonians}

Stabiliser Hamiltonians, formed by taking the sum over the projectors onto the -1~eigenspaces of the local stabilisers in a stabiliser code~\cite[\S 10.5]{Nielsen+Chuang}, are a special case whose ground states can be prepared \emph{efficiently} using the original dissipative state engineering algorithm~\cite{VWC}.
(Note that this is not the only way to efficiently prepare stabiliser states.)

Hamiltonians consisting of a stabiliser part, and a non-stabiliser part which commutes with this, arise naturally in various contexts:
\begin{gather}
  H = H_{\text{stab}} + H_{\text{enc}}, \\
  \text{where } H_{\text{stab}} = \sum_{s\in S} \Pi_s, \quad
  \norm{H_{\text{enc}}} \leq 1, \quad
  [H_{\text{enc}}, H_{\text{stab}}] = 0,
\end{gather}
where $S$ is a set of generators of a stabiliser code and $\Pi_s$ is the projector onto the -1~eigenspace of $s\in S$.
For example, fermion-to-qubit encodings generate Hamiltonians of this form~\cite{VerstraeteCirac,BravyiKitaev,compact_encoding}, as do recent rigorous constructions of toy models of AdS/CFT boundary Hamiltonians~\cite{AdS-CFT,AdS-CFT_random}.

For such Hamiltonians, instead of running the DQE algorithm on the full Hamiltonian, we can first efficiently prepare a state in the stabiliser code subspace,\footnote{Or in any other subspace with a given stabiliser syndrome.}
then run the DQE algorithm starting from that state just on the non-stabiliser part of the Hamiltonian $H_{\text{enc}}$.
Since $H_{\text{enc}}$---end hence the generalised measurement operators in \cref{DQE}---commute with $H_{\text{stab}}$, in the absence of errors the state will remain in the stabiliser subspace, and will converge to the lowest energy state within that subspace.
(Which in many such cases is guaranteed to be the global ground state.)

However, in the presence of errors and noise, the state will in general leak out of the stabiliser subspace.
Therefore, better than only running the stabiliser state preparation at the beginning, is to simultaneously run both the dissipative stabiliser state preparation algorithm on $H_{\text{stab}}$ and the DQE algorithm on $H_{\text{enc}}$.
This amounts to running the DQE \cref{DQE} on the whole Hamiltonian, but setting $\epsilon = 1$ for the $H_{\text{stab}}$ terms so that we perform a projective measurement of $\Pi_s$ for those terms.
(One can also choose a different resampling procedure for the stabiliser terms, such as locally rotating to the +1~eigenspace of $s$ when obtaining the wrong outcome of the $\Pi_s$ measurement, rather than locally resetting to the maximally mixed state.)

In this way, the state is continuously driven back to the correct stabiliser subspace, whilst converging within that subspace to the lowest energy state.

\subsection{Probabilistic gates}

Many photonic architectures can only naturally carry out non-Clifford gates probabilistically, but where failures are heralded~\cite{KLM,PsiQ_gates}.
To construct a scalable quantum computer, one typically uses measurement-based quantum computation schemes, with sophisticated error correcting codes and fault-tolerant gate implementations.
However, the DQE algorithm allows a more simplistic approach: if a probabilistic gate fails, one can simply treat this as a measurement failure (a ``1'' outcome) in the DQE \cref{DQE}, and carry on running.
With global resampling (see \cref{sec:stopping}, this is equivalent to the naive strategy of retrying the while quantum algorithm until all the probabilistic gates happen to succeed.
In the context of DQE, this increases (by an exponential factor) the expected stopping time under any of the stopping strategies in \cref{sec:stopped}.

However, with local resampling strategies (see \cref{sec:resampling}), treating gate failures are measurement failures approach no longer throws away all progress towards the ground state upon a single failure.
Thus it may give reasonable performance without the overhead of general fault-tolerant schemes.
This need not be all-or-nothing: tolerating some additional probability of measurement failures due to gate failures in the DQE algorithm, whilst using fault-tolerant measurement-based schemes to reduce the probability of gate failures, allows a trade-off between increased DQE run-time and increased measurement-based quantum computation overhead.

\subsection{Variational DQE and VQE Initial States}

Whilst the DQE \cref{DQE} does not require any optimisation over parameters to succeed (\cref{decaying_CPT_map}), there is also a natural extension to a variational version of the algorithm.
This may have advantages in achieving faster convergence to the ground state on certain systems, albeit at additional computation cost in implementing the variational optimisation step.

In the variational version of the DQE algorithm (VDQE), the total number of iterations is fixed in advance to some value $t$, e.g.\ determined by what is feasible on available hardware, and an initial sequence of $\epsilon_i$ in \cref{DQE} is chosen.
The DQE algorithm is then run, and the energy of the output state measured as in VQE.
The sequence of $\epsilon_i$ values are then variationally optimised over to minimise the energy.

Another possibility is to use VQE to generate the \emph{starting state} for the DQE \cref{DQE}, instead of initialising and resetting to the maximally mixed state.
This will not affect the proofs of convergence to the ground state in \cref{stopped_CPT_map,decaying_CPT_map} (note that the initial state $\rho_0$ is left unspecified until the end of those proofs).
But starting from a state that is closer to the ground state may reduce the run-time.
This approach can be thought of as using DQE to ``polish up'' a VQE state, potentially combining the speed benefits of VQE (once the VQE circuit has been optimised) with the convergence guarantees of DQE.
Exploring whether this is in fact the case is an interesting future direction.

\subsection{Combinatorial optimisation}

Combinatorial optimisation problems can equivalently be formulated as ground state problems for (classical) Hamiltonians.
For example, in the prototypical \textsc{MAXSAT} problem on boolean variables, we are given a set of clauses $C$ each of which constrains a subset of the variables to a given bit-string.
The task is then to find an assignment to the boolean variables that minimises the number of violated constraints.
This is exactly equivalent to finding the ground state of the Hamiltonian
\begin{equation}
  H_{\text{MAXSAT}} = \sum_{c\in C}\Pi_c
\end{equation}
where $\Pi_c$ is the projector onto the bit strings that violate the constraint $c\in C$.

Many local search algorithms for combinatorial optimisation problems can be viewed as special cases of the DQE \cref{DQE}, where $\epsilon = 1$ so that all the measurements become projective measurements in the computational basis and a particular resampling map is chosen which resamples states in the computational basis.
Just as quantum algorithms such as QAOA~\cite{QAOA} may give an algorithmic speedup on classical constraint optimisation problems~\cite{Sami_QAOA} (albeit still exponential-time, but potentially smaller exponent), the more general quantum DQE algorithm with $\epsilon < 1$ may give a similar computational advantage over classical local search algorithms.

\printbibliography

\end{document}